\documentclass{article}

\PassOptionsToPackage{numbers, compress}{natbib}

\usepackage[main, final]{neurips_2026}

\usepackage[utf8]{inputenc} 
\usepackage[T1]{fontenc}    
\usepackage{hyperref}       
\usepackage{url}            
\usepackage{booktabs}       
\usepackage{amsfonts}       
\usepackage{nicefrac}       
\usepackage{microtype}      
\usepackage{xcolor}         
\usepackage{amsmath}
\usepackage{bbm}
\usepackage{amsthm}
\usepackage{multirow}
\usepackage{graphicx}
\usepackage{subcaption}
\usepackage[ruled,vlined,linesnumbered]{algorithm2e} 
\usepackage{array}
\usepackage{svg}
\usepackage{wrapfig}
\usepackage[utf8]{inputenc}
\usepackage{xcolor}
\usepackage{array}
\usepackage{colortbl}   

\newtheorem{theorem}{Theorem}[section]
\newtheorem{proposition}[theorem]{Proposition}
\newtheorem{lemma}[theorem]{Lemma}
\newtheorem{corollary}[theorem]{Corollary}
\newtheorem{remark}[theorem]{Remark}

\newtheorem{assumption}{Assumption}

\DeclareMathOperator*{\argmax}{arg\,max}

\title{Steer-to-Detect: Probing Hidden Representations for Detection of LLM-Generated Texts}

\author{
  Luxu Liang \\
  School of Economics and Management \\
  and Shenzhen International Graduate School \\
  Tsinghua University \\
  \texttt{liang-lx25@mails.tsinghua.edu.cn} \\
  \And
  Xiang Li \\
  Perelman School of Medicine \\
  University of Pennsylvania \\
  \texttt{lx10077@upenn.edu} \\
}

\begin{document}
\maketitle

\begin{abstract}
The rapid advancement of large language models (LLMs) has made machine-generated text increasingly difficult to distinguish from human-written text. While recent studies explore leveraging internal representations of language models to uncover deeper detection signals, these raw features often exhibit substantial overlap between classes, limiting their discriminative power. To address this challenge, we propose Steer-to-Detect (\texttt{S2D}), a two-stage framework for detecting LLM-generated text. In the first stage, \texttt{S2D} learns a steering vector that is injected into the hidden states of a frozen observer LLM, producing representations with improved class separability. In the second stage, detection is performed via a hypothesis testing procedure based on the steered representations. 
We establish finite-sample, high-probability guarantees for Type I and Type II errors, providing a theoretical characterization of the procedure. Empirically, \texttt{S2D} achieves strong and consistent performance across a range of settings, including out-of-distribution scenarios and adversarial perturbations.

\end{abstract}



\section{Introduction}

Large language models (LLMs) are increasingly deployed across a broad range of domains, including education, finance, legal services, customer support, and writing assistance~\citep{cheong2024not,iacovides2024finllama,kobak2025delving,su2025llm,batra2025review}. Despite their widespread adoption, their use has raised growing societal concerns, including the spread of misinformation, academic misconduct, and the erosion of trust in written content~\citep{lee2023language,zugecova2025evaluation,xu2024instructions,yang2025adversarial}. A key reason is that LLM-generated text often closely resembles human writing, making it difficult to distinguish between the two. This challenge highlights the need for reliable detection of LLM-generated text.

Early work on this problem suggests that as long as distributional differences exist, reliable detection is possible in principle, particularly with sufficiently long samples \citep{mitchell2023detectgpt}.
Building on this observation, a prominent line of work seeks to explicitly introduce such differences during generation, most notably through watermarking methods that embed structured patterns identifiable with statistical guarantees~\citep{kirchenbauer2023watermark,li2025statistical}. While effective under controlled settings, these approaches rely on access to the generation process and are not applicable when watermarking is absent or not widely adopted. This limitation motivates passive methods, which aim to infer the origin of a text directly from observed samples without any control over the generation process. A detailed taxonomy is provided in Section~\ref{sec_related_work}.

A key challenge in the passive setting is that the distributional differences between human-written and LLM-generated text can be subtle, particularly for short or moderate-length samples. This raises a natural question: rather than relying solely on these weak signals, can we amplify them to improve detectability? Recent studies suggest that hidden representations of LLMs encode behavior-relevant information that is not fully captured by final-layer outputs~\citep{zou2023representation,orgad2024llms,li2023inference,skean2025layer,park2025steer}. These hidden representations provide a richer feature space in which the differences between human- and machine-generated text may become more pronounced. This observation motivates our central question: \textbf{\textit{Can we reshape LLM hidden representations to amplify such differences and enable reliable detection of machine-generated text?}}

\vspace{-0.1in}
\paragraph{Our Contributions.}
In this paper, we propose \emph{Steer-to-Detect} (\texttt{S2D}), a novel method for detecting LLM-generated text. At a high level, \texttt{S2D} employs a surrogate model as an ``observer'' and intervenes in its representation extraction process. Specifically, during the forward pass on input text, we steer hidden representations by adding a universal vector, thereby modifying how representations are formed. The vector is learned to enhance the separability between human-written and LLM-generated text. This approach avoids expensive model fitting and extends to other binary detection tasks.

On the theoretical side, we establish finite-sample control of the Type I (or false positive) error and derive an explicit upper bound on the excess Type II (or false negative) error. In addition, we provide a quantitative characterization of robustness, demonstrating the reliability of the detector in risk-sensitive settings.

Empirically, we conduct extensive experiments and show that \texttt{S2D} achieves strong and consistent performance across both in-distribution (ID) and out-of-distribution (OOD) settings. Further analyses demonstrate robustness to paraphrasing, adversarial perturbations, and variations in input length. In short, we propose a simple yet effective approach that enhances detectability by steering hidden representations.

\vspace{-0.05in}
\subsection{Related Works}\label{sec_related_work}
\vspace{-0.05in}

\paragraph{LLM-generated Text Detection.}
\begin{wraptable}{r}{0.48\linewidth}
\vspace{-8pt}
\centering
\small
\renewcommand{\arraystretch}{1.1}
\caption{Taxonomy of passive detection.}
\label{tab:llm_detection_taxonomy}
\vspace{2pt}
\resizebox{\linewidth}{!}{
\begin{tabular}{c|c|c}
\toprule
 & \textbf{Rewrite} & \textbf{Non-rewrite} \\
\midrule
\textbf{Train-free}
& \citep{mitchell2023detectgpt,sun2025zero}
& \citep{su2023detectllm,bao2023fast,hans2024spotting} \\
\midrule
\textbf{Train-based}
& \citep{mao2024raidar,zhou2026learn,huang2025magret}
& Ours \& \citep{solaiman2019release,ippolito2020automatic,mitrovic2023chatgpt,zhou2025adadetectgpt} \\
\bottomrule
\end{tabular}
}
\vspace{-10pt}
\end{wraptable}
Active methods modify the generation process and are beyond the scope of this work~\citep{yang2025watermarking,wang2025building,ye2026securing}, so we focus on passive detection. As summarized in Table~\ref{tab:llm_detection_taxonomy}, passive detectors can be organized along two axes: whether labeled training data are required and whether detection relies on auxiliary rewritten or perturbed variants of the input text. This yields four broad categories: train-free rewrite-based~\citep{mitchell2023detectgpt,sun2025zero}, train-free non-rewrite-based~\citep{su2023detectllm,bao2023fast,hans2024spotting}, train-based rewrite-based~\citep{mao2024raidar,huang2025magret}, and train-based non-rewrite-based methods~\citep{solaiman2019release,ippolito2020automatic,mitrovic2023chatgpt,zhou2025adadetectgpt}. Rewrite-based methods incur additional cost due to auxiliary text generation, whereas non-rewrite-based methods avoid this overhead. Our method falls into the \textit{train-based, non-rewrite-based} category. Prior work in this setting typically relies on final-layer outputs such as logits or derived scores, leaving intermediate representations underexplored~\citep{skean2025layer}. More recent studies have begun to use hidden representations for detection~\citep{yu2024text,chen2025repreguard}, with RepreGuard~\citep{chen2025repreguard} showing that hidden representations can provide stronger signals than final-layer outputs. Unlike these methods, we do not treat representations as fixed features. Instead, we shape their geometry to better separate human-written and LLM-generated text.

\vspace{-0.1in}
\paragraph{Representation Engineering.}
Representation engineering studies how information encoded in the hidden states of LLMs can be analyzed, interpreted, and manipulated to understand, monitor, or control model behavior~\citep{durrani2020analyzing,zou2023representation,hedstrom2025steer,skean2025layer,voita2024neurons,ghandeharioun2024patchscopes,sun2025layernavigator,nguyen2025activation,venkateswaran2026spotlight,beaglehole2026toward,davarmanesh2026efficient}. This line of work includes both methods for characterizing what hidden representations encode and intervention-based approaches that modify them. Among these, \emph{activation steering} has emerged as a prominent technique that perturbs intermediate activations along task-specific directions (often represented as steering vectors) in representation space, enabling targeted behavioral changes without updating model parameters~\citep{turner2023steering,im2025unified}. Recent studies apply activation steering to improve factual reliability, control generation style, and monitor model behaviors such as hallucinations and toxic content using internal representations~\citep{rimsky2024steering,wu2025sharp,sheng2025alphasteer,qiu2025hallucination,hua2025steering,beaglehole2026toward}.
Despite these advances, the use of representation-level interventions for LLM-generated text detection remains relatively underexplored. In this work, we repurpose activation steering to reshape the geometry of hidden representations, transforming raw hidden states into features that better separate human-written and LLM-generated text.



\section{Our Method: Steer-to-Detect (\texttt{S2D})}
\begin{figure}[t]
    \centering
    \includegraphics[width=1\textwidth]{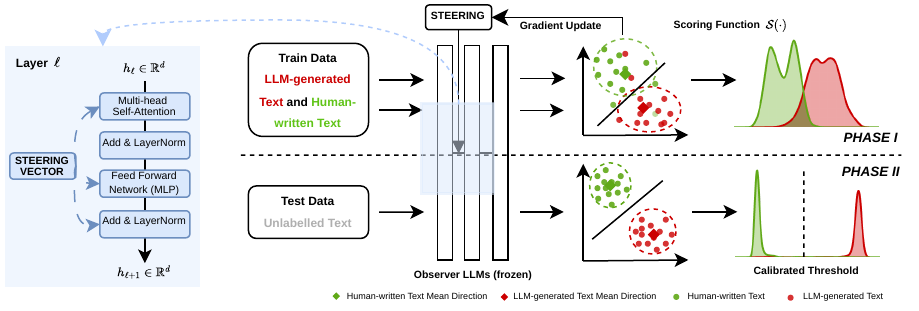}
    \caption{\textbf{Overview of Steer-to-Detect (\texttt{S2D}).} \textit{Phase I} (top row) applies a steering vector to reshape the observer LLM’s hidden representations, enhancing the separation between human-written and LLM-generated text. 
    \textit{Phase II} (bottom row) scores unseen texts and rejects the null hypothesis when the score exceeds a calibrated threshold.}
    \label{fig:s2d_overview}
    \vspace{-0.05in}
\end{figure}



\paragraph{Problem Formulation.}
Let $\mathcal{X}$ denote the space of all text sequences (possibly of varying length). We assume that human-written and LLM-generated texts are drawn from two distributions, $\mathbb{P}_0$ and $\mathbb{P}_1$, respectively. Given an observation $x \in \mathcal{X}$, detection is formulated as the binary hypothesis test:
\begin{equation}
\label{eq:hypothesis-testing}
\mathcal{H}_0: x \sim \mathbb{P}_0 \quad (\text{human-written})
\quad \textit{versus} \quad
\mathcal{H}_1: x \sim \mathbb{P}_1 \quad (\text{LLM-generated}).
\end{equation}

\vspace{-0.05in}
\subsection{Method Overview}
\label{sec:problem_formulation}
We begin with an overview of our method. We employ a surrogate model as an ``observer'' and extract hidden representations during the forward pass on input text. Representations of LLM-generated and human-written text exhibit systematic differences, which, although subtle, provide a sufficient signal for detection~\citep{chen2025repreguard}. Motivated by this observation, we intervene in the representation formation process by injecting a lightweight, learnable steering vector $\mathbf{v}$ into intermediate hidden states during the forward pass. This modifies how representations are formed and enhances their separability between the two types of text. Based on the resulting representations, we construct a hypothesis test for detection. The overall framework is illustrated in Figure~\ref{fig:s2d_overview}.

More specifically, let $f_{\theta}: \mathcal{X} \to \mathbb{S}^{d-1}$ denote the representation map of an observer model, which embeds a text sequence into a $d$-dimensional vector on the unit hypersphere (to be specified later). To enable intervention in the representation space, we introduce a steering vector $\mathbf{v} \in \mathbb{R}^d$ and denote the resulting steered representation map by $f_{\theta,\mathbf{v}}$. 

Our method proceeds in two stages, as introduced earlier. In the first stage, we estimate the steering vector $\mathbf{v}$ from training data $\mathcal{S}_{\text{train}}=\{(x_i,y_i)\}_{i=1}^{n_1}$ and fit parametric class-conditional models to the steered representations. In the second stage, we construct a likelihood-ratio statistic based on the estimated models and use it to test whether a given text is human-written or LLM-generated. 
We next describe (i) the computation of $f_{\theta}$ and the introduction of $\mathbf{v}$ in Section~\ref{sec:representation}, (ii) the learning of $\mathbf{v}$ in Section~\ref{sec:phase_1}, and (iii) the detection procedure in Section~\ref{sec:phase_2}.

\vspace{-0.05in}
\subsection{Representation Extraction and Steering}
\label{sec:representation}

\paragraph{Representation Extraction.}\label{sec_repre_ext}
We define the mapping $f_{\theta}$ by extracting hidden states from an observer model with $L$ Transformer blocks and hidden dimension $d$. For a sequence $x \in \mathcal{X}$, let $h_{\ell,t}(x) \in \mathbb{R}^d$ denote the hidden state at layer $\ell (\le L)$ and token index $t$. Prior work~\citep{xiao2023efficient,zou2023representation,skean2025layer,chen2025repreguard} suggests that tokens near the end of a sequence encode richer contextual information and exhibit more pronounced differences between human-written and LLM-generated text, as later tokens aggregate broader preceding context in causal LLMs. Motivated by this observation, we focus on the final $K \in (0,1]$ fraction of valid (non-padding) tokens,\footnote{Non-padding tokens refer to tokens corresponding to actual text content, excluding padding tokens.} indexed by $\mathcal{I}_{K}(x)$. 
In addition, intermediate layers capture complementary information beyond the final layer~\citep{skean2025layer}, so we aggregate hidden states from the last $N$ layers to obtain a more stable sequence-level representation. Mathematically, we define
\begin{equation*}
    f_{\theta}(x) = \frac{\bar{m}(x)}{\|\bar{m}(x)\|} \in \mathbb{S}^{d-1}
    \quad \text{with} \quad
    \bar{m}(x)
    :=
    \frac{1}{N \cdot |\mathcal{I}_K(x)|}
    \sum_{\ell=L-N+1}^{L}
    \sum_{t \in \mathcal{I}_K(x)}
    h_{\ell,t}(x) \in \mathbb{R}^d,
\end{equation*}
where $|\mathcal{I}_K(x)|$ denotes the number of selected non-padding tokens.

\begin{remark}[Representation Norms and Directionality]
In Appendix~\ref{app_repre_norm}, we provide empirical support for this modeling choice $f_{\theta}$. 
In modern Transformer models, representation norms are often approximately uniform after root mean square normalization, leaving directional variation as the primary source of information~\citep{park2025steer,grattafiori2024llama,yang2025qwen3}. 
For example, in \textit{meta-llama/Llama-3.1-8B}, the mean norm is around 18.8 across four different datasets; see Figure \ref{fig_norm_distribution} in the appendix for details. 
We also observe that $f_{\theta}(x)\mid y$ exhibits unimodal structure (see Figure~\ref{fig_vmf_domain} in the appendix), which motivates the use of the vMF model in deriving \eqref{eq:postvMF}.
\end{remark}

\vspace{-0.05in}
\paragraph{Representation Steering.}
Although the base extractor $f_{\theta}(\cdot)$ captures high-level semantic information, the resulting class-conditional distributions may still overlap in latent space, even in out-of-distribution settings (see Appendix~\ref{appendix_score_dist} for evidence). 

To improve separability, we introduce a learnable steering vector $\mathbf{v} \in \mathbb{R}^d$ and define a steered extractor $f_{\theta,\mathbf{v}}(\cdot)$.
We implement steering as an additive intervention on intermediate activations of the observer model by injecting a vector $\mathbf{v}$ at a chosen layer $\ell$:
\begin{equation*}
    h_{\ell,t}(\cdot) \leftarrow h_{\ell,t}(\cdot) + \mathbf{v}.
\end{equation*}
In the default setting, the intervention is applied at an intermediate layer $\ell \le L-N+1$. Importantly, the method is not sensitive to the exact choice of $\ell$: as long as the intervention is applied within intermediate layers, performance remains stable (see Section~\ref{understanding_s2d}). 

This is the only modification to the model: after the injection, the observer model remains fixed, and the perturbed activations propagate through subsequent Transformer blocks without further intervention, producing a new representation.
Intuitively, since the model captures rich semantic structure from large-scale pre-training, the steering vector $\mathbf{v}$ introduces a direction that enhances its intrinsic discriminative capability. Applying the same token selection, layer aggregation, and normalization yields the steered representation $f_{\theta,\mathbf{v}}(\cdot) \in \mathbb{S}^{d-1}$ for the downstream hypothesis test.


\vspace{-0.05in}
\subsection{Phase I: Learning the Steering Vector via Likelihood Maximization}
\label{sec:phase_1}

\paragraph{Likelihood-Based Modeling.}
A remaining question is how to learn the steering vector $\mathbf{v}$, which is addressed in Phase I of our method. The steering vector aims to enhance the separability of latent representations across classes. To quantify this, we formulate a classification problem where the feature is given by $f_{\theta, \mathbf{v}}$ and the response is the ground-truth label ($y_i=1$ for LLM-generated text and $y_i=0$ for human-written text).
More specifically, we model the conditional distribution $f_{\theta, \mathbf{v}}(x_i) \mid y_i$ as a von Mises--Fisher (vMF) distribution, a directional analogue of a Gaussian on the unit sphere.\footnote{Under the vMF model, $p(f_{\theta,\mathbf v}(x_i)\mid y_i=c)
\propto \exp\left( \kappa\boldsymbol{\mu}_c^\top
f_{\theta,\mathbf v}(x_i) \right)$, where $\kappa$ is the concentration parameter. The Phase-I objective retains the resulting softmax form but uses a separate inverse temperature $\beta$.} Under a shared-concentration vMF model, the class posterior takes a softmax form. Motivated by this structure, we use the same functional form for Phase-I optimization, with the inverse temperature $\beta>0$. Specifically, under a uniform class prior, we define
\begin{equation}
\label{eq:postvMF}
    q_{\beta}\left(y_i \mid f_{\theta, \mathbf{v}}(x_i), \boldsymbol{\mu}_0, \boldsymbol{\mu}_1\right)
    =
    \frac{\exp\left(\beta \boldsymbol{\mu}_{y_i}^\top f_{\theta, \mathbf{v}}(x_i)\right)}
    {\sum_{c \in \{0,1\}} \exp\left(\beta \boldsymbol{\mu}_c^\top f_{\theta, \mathbf{v}}(x_i)\right)}.
\end{equation}
Here, $\boldsymbol{\mu}_c \in \mathbb{S}^{d-1}$ denotes the mean direction for class $c \in \{0,1\}$, while $\beta$ is an optimization hyperparameter and is distinct from the shared vMF concentration parameter $\kappa$. Appendix~\ref{app:derivation} derives the corresponding softmax posterior under the vMF model and explains this temperature-scaled formulation. To estimate the steering vector, we collect a labeled training set $\mathcal{S}_{\text{train}}=\{(x_i,y_i)\}_{i=1}^{n_1}$ and optimize the following objective:
\begin{equation}
\label{eq:loss-for-v}
(\widehat{\mathbf{v}}, \widehat{\boldsymbol{\mu}}_0, \widehat{\boldsymbol{\mu}}_1) = 
\argmax\limits_{
\mathbf{v}\in\mathbb{R}^d,
\boldsymbol{\mu}_0,\boldsymbol{\mu}_1\in\mathbb{S}^{d-1}
}
\frac{1}{n_1}\sum_{i=1}^{n_1}
\log q_\beta\left(
y_i\mid
f_{\theta,\mathbf v}(x_i),
\boldsymbol{\mu}_0,
\boldsymbol{\mu}_1
\right).
\end{equation}

\vspace{-0.05in}
\paragraph{Optimization Procedure.}
In practice, we adopt a two-timescale optimization scheme to optimize \eqref{eq:loss-for-v}. 
Specifically, the steering vector $\mathbf{v}$ is updated via gradient ascent on a slower timescale (i.e., using a smaller step size $\eta$), while the class mean directions $\boldsymbol{\mu}_c$ are updated on a faster timescale (i.e., using an exponential moving average (EMA) with coefficient $\rho \in (0,1)$) over normalized class-specific embeddings~\citep{wang2021pico}. 
The observer model remains fixed throughout the optimization, and we choose $0<\eta \ll \rho \le 1$ to introduce two time scales. The full procedure is provided in Appendix~\ref{app_algo} due to space constraints. 

\begin{remark}[Gradient analysis]
 Intuitively, this objective in \eqref{eq:loss-for-v} encourages the steering vector to reshape the representation space so that class-wise representations become more separable. In Appendix~\ref{app_prop_1}, we characterize the gradient of the population version of \eqref{eq:loss-for-v}, which provides theoretical support for this effect.   
\end{remark}

\vspace{-0.05in}
\subsection{Phase II: Detection via the Log-Likelihood Ratio Test}
\label{sec:phase_2}

\paragraph{Log-Likelihood Ratio Test.}
Detection naturally reduces to a log-likelihood ratio test. Under the shared-concentration vMF generative model, the Neyman-Pearson (NP) lemma~\citep{lehmann2005testing} implies that the most powerful test is based on the log-likelihood ratio:
\begin{equation}
\label{eq:score-for-detection}
\mathcal{S}(x)
=
\log\left(
\frac{p(f_{\theta,\mathbf{v}}(x)\mid y=1,\boldsymbol{\mu}_1,\boldsymbol{\mu}_0)}
{p(f_{\theta,\mathbf{v}}(x)\mid y=0,\boldsymbol{\mu}_1,\boldsymbol{\mu}_0)}
\right)
=
\kappa (\boldsymbol{\mu}_1-\boldsymbol{\mu}_0)^\top f_{\theta,\mathbf{v}}(x).
\end{equation}
Since $\kappa>0$ only rescales the score, it does not affect the rejection region after threshold calibration. Thus, the detector is equivalently based on projecting the steered representation onto the discriminative direction $\boldsymbol{\mu}_1-\boldsymbol{\mu}_0$, with unscaled projection used in practice. In practice, we plug in the estimates $(\widehat{\mathbf{v}}, \widehat{\boldsymbol{\mu}}_0, \widehat{\boldsymbol{\mu}}_1)$ obtained from \eqref{eq:loss-for-v} to form the empirical score $\widehat{\mathcal{S}}(x)$. The resulting detector is $\widehat{\phi}(x) := \mathbbm{1}\bigl(\widehat{\mathcal{S}}(x) > \widehat{\tau}\bigr)$, which rejects the null hypothesis $\mathcal{H}_0$ when the score exceeds the threshold $\widehat{\tau}$. Accordingly, $x$ is classified as LLM-generated if $\widehat{\phi}(x)=1$, and as human-written otherwise.

\vspace{-0.05in}
\paragraph{Threshold Calibration.}
In high-stakes scenarios such as academic integrity assessment, falsely labeling human-written text as LLM-generated corresponds to a Type-I error with substantial practical and ethical consequences. In such settings, controlling the false positive rate (FPR) is more important than optimizing a balanced trade-off between detection power and false alarms. Accordingly, we calibrate the threshold using an independent calibration sample $\mathcal S_{\mathrm{cal}} = \{x_i^{-}\}_{i=1}^{n_2}$ drawn from the null distribution, and define
\begin{equation}\label{type_1_error}
\widehat{\tau}_\alpha := \inf\left\{
\tau: \frac{1}{n_2}\sum_{i=1}^{n_2}
\mathbbm{1}\bigl(\widehat{\mathcal S}(x_i^{-})>\tau\bigr)
\le \alpha \right\}.
\end{equation}
This empirical calibration enforces the target Type-I error level on the calibration set. Moreover, Theorem~\ref{thm_1} shows that the resulting test achieves approximate control of the population Type-I error in finite samples with high probability.

\begin{remark}[Threshold from Youden index]
In Appendix~\ref{app:threshold_calibration}, we describe an alternative threshold selection strategy based on the Youden index, which is commonly used in standard benchmarking settings~\citep{chen2025repreguard,ruopp2008youden,yin2014joint,liu2012classification}. While this approach is not designed to control the Type I error, our empirical results indicate that it can nevertheless achieve competitive false positive control in practice.
\end{remark}


\section{Theoretical Analysis}\label{sec_thm}

In this section, we provide a theoretical analysis of our method. 
We begin by establishing finite-sample guarantees for Type I error control, and then characterize the excess Type II error arising from parameter estimation. 
Finally, we study the effect of null distribution shifts on the resulting test. All proofs are deferred to Appendix~\ref{app_pf}.

\vspace{-0.05in}
\paragraph{Type-I Error Control.}
We show that empirical threshold calibration ensures Type I error control with high probability. Let $\widehat{\mathcal{S}}_t$ denote the score function at training iteration $t$, trained on $\mathcal S_{\mathrm{train}}$, and define the detector $\widehat{\phi}(x) := \mathbbm{1}\bigl( \widehat{\mathcal S}(x)>\widehat{\tau}_\alpha \bigr)$, where $\widehat{\tau}_{\alpha,t}$ is the empirical threshold at iteration $t$, computed from an independent null sample $\mathcal S_{\mathrm{cal}}=\{x_i^{-}\}_{i=1}^{n_2}$ as in~\eqref{type_1_error}.
Theorem~\ref{thm_1} shows that this procedure controls the Type I error up to a finite-sample deviation that vanishes as $n_2$ increases.
\begin{theorem}[Type-I Error Control]\label{thm_1}
Let $\alpha,\delta\in(0,1)$ and suppose that, conditional on the training sample, the null score $\widehat{\mathcal S}_t(X^-)$ has a continuous distribution. Then, with probability at least $1-\delta$ over the independent calibration sample, $\left|
\mathbb P_0\left(\widehat{\phi}(X_{\mathrm{test}}^-)=1
\right)-\alpha \right|\le \sqrt{\log(2/\delta) / (2n_2)}
+1 / n_2$.
\end{theorem}
\begin{remark}[Ties]
The continuity assumption ensures almost surely distinct calibration scores. With ties, one may conservatively exclude the boundary group for one-sided Type-I control or use randomized tie-breaking for exact empirical calibration.
\end{remark}



\vspace{-0.05in}
\paragraph{Excess Type-II Error.}
Next, we study the Type-II error. 
In general, analyzing the power of the detector is challenging due to the complex and implicit behavior of LLM-extracted representations. 
To facilitate analysis, we consider a simplified setting in which the extracted features $f_{\theta,\mathbf v}(X)$ follow a class-conditional vMF model with certain ground-truth parameters $\boldsymbol{\mu}_0$ and $\boldsymbol{\mu}_1$. This assumption allows us to define the corresponding oracle likelihood ratio test (i.e., $\mathcal{S}$ in \eqref{eq:score-for-detection}) based on this true model, and to study the excess Type-II error of the plug-in detector $\widehat{\mathcal{S}}$ relative to this oracle benchmark. In particular, it enables us to disentangle the contributions of estimation and calibration.

Importantly, the oracle and plug-in scores are evaluated at the same steering iterate $\mathbf v_t$ and operate on the same steered representation $f_{\theta,\mathbf v_t}(x)$. The following result therefore isolates prototype-estimation and threshold-calibration errors, rather than comparing $\mathbf v_t$ with a globally optimal steering vector or establishing global convergence of the steering optimization.

\begin{theorem}[Excess Type-II Error (informal)]\label{thm_2}
Assume the extracted features $f_{\theta, \mathbf{v}}$ follow a vMF distribution, and the EMA coefficient $\rho$ and learning rate $\eta$ are chosen sufficiently small, satisfying $0 < \eta \ll \rho < 1$. Under regularity conditions, there exists a constant $c \in (0,1)$ such that, with probability at least $1-2\delta$, for any fixed $0 \le t \le T$, the gap $\mathbb{P}_1(\widehat{\phi}(X_{\mathrm{test}}^{+})=0)-\mathbb{P}_1({\phi}(X_{\mathrm{test}}^{+})=0)$ is bounded by
$$
\begin{aligned}
\mathcal{O}\Bigg( \underbrace{\left((1- {c}\rho)^t+\sqrt{\rho \log\frac{4T}{\delta}}+ \frac{\eta^2}{\rho^2}\right)^{\frac{1+\bar{\gamma}}{2}}}_{\text{Estimation Error}}  + \underbrace{ 
\sqrt{\frac{\log(2/\delta)}{n_2}} + \frac{1}{n_2}}_{\text{Calibration Error}} \Bigg)
\end{aligned}
$$
where $n_2$ is the calibration sample size and $\bar{\gamma}$ characterizes the local concentration of the score distribution (i.e., $\mathcal{S}$ under $\mathbb{P}_0$) around the threshold $\tau_\alpha^\star$.\footnote{A larger value of $\bar{\gamma}$ corresponds to less mass near the threshold and hence an easier estimation problem. See Assumption \ref{asump3:local_mass} for the detailed definition.}
\end{theorem}
Theorem~\ref{thm_2} establishes an upper bound on the excess Type-II error, which characterizes the suboptimality of the learned score function $\widehat{\mathcal{S}}$. 
The bound consists of two components: an estimation error arising from learning the ground-truth parameters $\boldsymbol{\mu}_0$ and $\boldsymbol{\mu}_1$, and a calibration error due to estimating the critical value $\widehat{\tau}_{\alpha,t}$. 

The second term follows from Theorem~\ref{thm_1} and reflects the finite-sample effect of threshold calibration; it decreases as the calibration sample size $n_2$ increases. 
The first term can be further decomposed into several parts: an optimization term exhibiting geometric decay $(1-c\rho)^t$, a variance term induced by stochastic training $\sqrt{\rho \log(4T/\delta)}$, and a higher-order lag term $\eta^2/\rho^2$. 
When $t$ is sufficiently large and the two-timescale condition $\eta \ll \rho < 1$ holds with $\rho$ sufficiently small, the estimation error becomes negligible.

\vspace{-0.05in}
\paragraph{Effects of Distributional Shift.}
Prior work has highlighted performance degradation under domain shifts in LLM-generated text detection~\citep{mao2025general,zhou2026detecting}, and proposed database-based strategies to mitigate such a mismatch. However, these approaches do not explicitly characterize the impact of distribution shift. 
In Appendix~\ref{app_prop_2}, we quantify this effect by analyzing the detector under a Wasserstein perturbation with magnitude $\mathcal{E}$. Proposition~\ref{prop_3} provides finite-sample bounds for both Type-I and Type-II errors under the shifted distributions. In particular, when $\mathcal{E} > 0$, both errors may deteriorate, reflecting the impact of distribution mismatch. 


\section{Experiment}\label{sec:exp}
In this section, we present a comprehensive evaluation of the proposed \texttt{S2D} detector. Our code is available at: \url{https://github.com/LuxLiang/steer-to-detect-release/tree/main}.


\vspace{-0.05in}
\paragraph{Datasets.}
Following previous setups~\citep{chen2025repreguard,liu2025mgt}, we evaluate \texttt{S2D} on the DetectRL benchmark~\citep{wu2024detectrl}, which covers four domains that are particularly vulnerable to misuse by LLMs, including academic writing (arXiv Archive\footnote{\url{https://www.kaggle.com/datasets/spsayakpaul/arxiv-paper-abstracts/data.}}), news summarization (XSum~\citep{narayan2018don}), creative writing (WritingPrompts~\citep{fan2018hierarchical}), and user reviews (Yelp~\citep{zhang2015character}). For each domain, the benchmark provides 2,800 pairs of human-written and LLM-generated texts. The LLM-generated texts are produced by four widely used LLMs: \textit{GPT-3.5-Turbo}~\citep{openai2023chatgpt}, \textit{Claude-Instant}~\citep{anthropic2023claudeinstant}, \textit{Google-PaLM}~\citep{anil2023palm}, and \textit{Llama-2-70B}~\citep{touvron2023llama}. To ensure robust evaluation, we repeat the experiment over five independent runs with different random samplings of the training data. In each run, we randomly sample 512 human-written/LLM-generated text pairs to form the training set, while using a fixed, disjoint test set of 1,000 pairs. We further evaluate on the recent OpAI-Bench~\citep{bsharat2026operation}, which contains progressive human-to-AI revision trajectories across four domains and includes texts revised by more recent generators, including \textit{GPT-5.4}, \textit{GPT-5.4-nano}, \textit{Gemini-2.5-Flash}, and \textit{Qwen3-8B}. Full experimental details and results are provided in Appendix~\ref{additional_exp}.

\vspace{-0.05in}
\paragraph{Baselines.}
In our primary DetectRL evaluation, we compare \texttt{S2D} against thirteen baselines across train-free and train-based categories. The train-free baselines include: \underline{\textit{Likelihood}}~\citep{gehrmann2019gltr}, Log Rank Ratio (\underline{\textit{LRR}})~\citep{gehrmann2019gltr}, Fast-DetectGPT (\underline{\textit{FDGPT}})~\citep{bao2023fast}, \underline{\textit{Binoculars}}~\citep{hans2024spotting}, Implicit Reward Model (\underline{\textit{IRM}})~\citep{liu2025implicit},
\underline{\textit{Lastde++}}~\citep{xu2025training},
\underline{\textit{SpecDetect++}}~\citep{luo2026specdetect}, and Normalized Temperature Sensitivity (\underline{\textit{NTS}})~\citep{ma2026zero}. The train-based baselines consist of: \underline{\textit{RoBERTa-Large}}~\citep{solaiman2019release}, \underline{\textit{RAIDAR}}~\citep{mao2024raidar}, Imitate Before Detection (\underline{\textit{ImBD}})~\citep{chen2025imitate}, \underline{\textit{RepreGuard}}~\citep{chen2025repreguard}, and Learn-to-Distance (\underline{\textit{L2D}})~\citep{zhou2026learn}. We further conduct controlled comparisons against logistic regression, a linear SVM, and a small MLP trained on the same pooled observer representations without steering, to assess whether a standard discriminative readout alone can account for the observed performance. The complete comparisons are reported in Appendix~\ref{additional_exp}.

\vspace{-0.05in}
\paragraph{Evaluation Metrics.}
We evaluate detection performance using AUROC and the true positive rate (TPR) at two fixed false positive rate (FPR) levels, namely TPR@1\% and TPR@0.01\%\footnote{Since the DetectRL test sets contain
1,000 human-written samples, TPR@0.01\% is below the empirical FPR resolution and is obtained by linear interpolation of the ROC curve. We further examine the accuracy of this interpolation using a substantially larger null set in Appendix~\ref{app:low_fpr_validation}.}. These metrics capture both overall discrimination and performance in the high-precision regime where false positives must be tightly controlled. Further experimental details are provided in Appendix \ref{sec:experiment_details}.

\vspace{-0.05in}
\subsection{Detection Performance and Robustness}
\label{sec_rob_ana}

\paragraph{Detection under ID and OOD settings.}
We first evaluate detection performance under both in-distribution (ID) and out-of-distribution (OOD) settings, as summarized in Table~\ref{tab:main_experiment}. The detector is trained on text generated by a specific source model and evaluated either on the same generator (ID, diagonal entries) or on different generators (OOD, off-diagonal entries).


Under the ID setting, \texttt{S2D} performs strongly across all generators when the training and test distributions are aligned. RepreGuard and L2D are the most competitive baselines, often achieving TPRs above $99\%$ at low FPR, while other training-based methods are less consistent across generators. L2D relies on rewriting, whereas \texttt{S2D} uses a single forward pass of the frozen observer model at inference time; efficiency comparisons are provided in Appendix~\ref{sec:computational_efficiency}. In the OOD setting, performance gaps become more pronounced, and \texttt{S2D} remains competitive against most baselines.

\begin{table*}[!t]
\centering
\renewcommand{\arraystretch}{1.2}
\setlength{\tabcolsep}{3pt}
\resizebox{\textwidth}{!}{
\begin{tabular}{ll ccc ccc ccc ccc}
\toprule
\multirow{2.5}{*}{\bf Train $\downarrow$} & \multirow{2.5}{*}{\bf Detector $\downarrow$} & \multicolumn{3}{c}{\bf Test: ChatGPT} & \multicolumn{3}{c}{\bf Test: Llama-2-70B} & \multicolumn{3}{c}{\bf Test: Google-PaLM} & \multicolumn{3}{c}{\bf Test: Claude-Instant} \\
\cmidrule(lr){3-5} \cmidrule(lr){6-8} \cmidrule(lr){9-11} \cmidrule(lr){12-14}
& & \textit{ AUROC} & \textit{ TPR@1\%} & \textit{ TPR@.01\%} & \textit{ AUROC} & \textit{ TPR@1\%} & \textit{ TPR@.01\%} & \textit{ AUROC} & \textit{ TPR@1\%} & \textit{ TPR@.01\%} & \textit{ AUROC} & \textit{ TPR@1\%} & \textit{ TPR@.01\%} \\
\midrule

\rowcolor[gray]{0.96} \multicolumn{14}{l}{\textit{\textbf{Train-free Methods}}} \\
\multirow{8}{*}{N/A}
 & Likelihood   & 89.91 & 73.00 & 0.40 & 89.32 & 80.30 & 5.80 & 85.73 & 69.30 & 13.10 & 76.72 & 5.50 & 0.40 \\
 & LRR          & 88.95 & 78.00 & 0.01 & 89.51 & 81.50 & 38.80 & 85.66 & 73.70 & 21.30 & 76.75 & 7.00 & 0.20 \\
 & FDGPT        & 94.70 & 51.90 & 0.06 & 96.85 & 70.10 & 32.90 & 95.26 & 66.00 & 50.40 & 58.37 & 4.10 & 0.70 \\
 & Binoculars   & 94.71 & 90.67 & 87.00 & 99.15 & 96.70 & 88.90 & 94.50 & 91.40 & 89.30 & 87.70 & 37.50 & 13.70 \\
 & IRM          & 99.08 & 90.70 & 83.10 & 99.20 & 98.20 & 95.50 & 99.11 & 91.30 & 85.20 & 95.50 & 69.20 & 35.10 \\
 & Lastde++     & 94.32 & 53.10 & 4.27 & 96.07 & 65.77 & 20.30 & 94.85 & 63.07 & 47.80 & 69.38 & 10.23 & 2.33 \\
 & SpecDetect++ & 94.19 & 38.90 & 3.07 & 94.76 & 54.33 & 2.60 & 89.46 & 42.30 & 15.07 & 69.90 & 5.70 & 0.40 \\
 & NTS          & 93.48 & 56.60 & 6.00 & 96.93 & 76.40 & 57.30 & 94.75 & 71.60 & 36.80 & 66.72 & 8.90 & 2.30 \\
\midrule

\rowcolor[gray]{0.96} \multicolumn{14}{l}{\textit{\textbf{Train-based Methods}}} \\

\multirow{6}{*}{ChatGPT} 
 & RoBERTa-L  & $98.99_{\pm .12}$ & $98.90_{\pm .25}$ & $98.10_{\pm .45}$ & $98.82_{\pm .15}$ & $63.70_{\pm 1.5}$ & $44.80_{\pm 2.1}$ & $98.08_{\pm .22}$ & $45.60_{\pm 2.0}$ & $25.40_{\pm 2.5}$ & \cellcolor{blue!25}$\mathbf{98.95_{\pm .18}}$ & $76.90_{\pm 1.6}$ & $15.80_{\pm 1.5}$ \\
 & RAIDAR     & $99.89_{\pm .05}$ & $99.90_{\pm .08}$ & $99.10_{\pm .15}$ & $98.49_{\pm .21}$ & $92.20_{\pm .85}$ & $32.60_{\pm 1.8}$ & $94.53_{\pm .60}$ & $75.40_{\pm 1.4}$ & $51.90_{\pm 1.9}$ & $87.40_{\pm 1.1}$ & $67.80_{\pm 1.9}$ & $0.10_{\pm .05}$ \\
 & ImBD       & $99.94_{\pm .04}$ & $99.90_{\pm .06}$ & $98.72_{\pm .20}$ & \cellcolor{blue!10}$\underline{99.98_{\pm .02}}$ & \cellcolor{blue!10}$\underline{99.90_{\pm .05}}$ & $93.64_{\pm .90}$ & \cellcolor{blue!25}$\mathbf{99.65_{\pm .10}}$ & \cellcolor{blue!25}$\mathbf{95.30_{\pm .65}}$ & \cellcolor{blue!10}$\underline{84.82_{\pm 1.1}}$ & $86.18_{\pm 1.3}$ & $25.20_{\pm 2.4}$ & $9.34_{\pm .80}$ \\
 & L2D        & \cellcolor{blue!10}$\underline{99.99_{\pm .01}}$ & \cellcolor{blue!25}$\mathbf{100.0_{\pm .00}}$ & \cellcolor{blue!10}$\underline{99.90_{\pm .05}}$ & $99.80_{\pm .11}$ & $99.50_{\pm .12}$ & $95.50_{\pm .70}$ & $86.55_{\pm 1.2}$ & $57.80_{\pm 1.8}$ & $31.90_{\pm 2.2}$ & $83.80_{\pm 2.1}$ & $40.10_{\pm 4.5}$ & $12.30_{\pm 2.40}$ \\
 & RepreGuard & $99.84_{\pm .10}$ & $99.80_{\pm .00}$ & $99.70_{\pm .00}$ & $99.54_{\pm .12}$ & $99.43_{\pm .04}$ & \cellcolor{blue!10}$\underline{99.40_{\pm .00}}$ & $97.19_{\pm .12}$ & $81.40_{\pm .21}$ & $79.70_{\pm .41}$ & $98.25_{\pm .16}$ & \cellcolor{blue!10}$\underline{78.53_{\pm 1.1}}$ & \cellcolor{blue!10}$\underline{70.40_{\pm .90}}$ \\
 & \texttt{S2D} (ours) & \cellcolor{blue!25}$\mathbf{99.99_{\pm .00}}$ & \cellcolor{blue!25}$\mathbf{100.0_{\pm .00}}$ & \cellcolor{blue!25}$\mathbf{99.90_{\pm .00}}$ & \cellcolor{blue!25}$\mathbf{100.0_{\pm .00}}$ & \cellcolor{blue!25}$\mathbf{99.95_{\pm .00}}$ & \cellcolor{blue!25}$\mathbf{99.90_{\pm .05}}$ & \cellcolor{blue!10}$\underline{98.93_{\pm .25}}$ & \cellcolor{blue!10}$\underline{91.93_{\pm 1.2}}$ & \cellcolor{blue!25}$\mathbf{90.83_{\pm 1.1}}$ & \cellcolor{blue!10}$\underline{98.31_{\pm .34}}$ & \cellcolor{blue!25}$\mathbf{83.02_{\pm 3.8}}$ & \cellcolor{blue!25}$\mathbf{79.48_{\pm 4.8}}$ \\
\cmidrule{2-14}

\multirow{6}{*}{Llama-2} 
 & RoBERTa-L  & $98.97_{\pm .20}$ & $98.90_{\pm .30}$ & $98.10_{\pm .40}$ & $99.90_{\pm .05}$ & \cellcolor{blue!10}$\underline{99.90_{\pm .04}}$ & $97.10_{\pm .60}$ & \cellcolor{blue!10}$\underline{99.82_{\pm .06}}$ & \cellcolor{blue!10}$\underline{96.80_{\pm .50}}$ & $83.50_{\pm 1.2}$ & \cellcolor{blue!10}$\underline{99.11_{\pm .15}}$ & \cellcolor{blue!10}$\underline{90.70_{\pm .80}}$ & $57.90_{\pm 2.2}$ \\
 & RAIDAR     & $99.98_{\pm .02}$ & \cellcolor{blue!25}$\mathbf{100.0_{\pm .00}}$ & $92.40_{\pm 1.2}$ & $99.58_{\pm .12}$ & $96.80_{\pm .50}$ & $70.80_{\pm 1.5}$ & $97.56_{\pm .40}$ & $84.10_{\pm 1.1}$ & $61.30_{\pm 1.8}$ & $88.58_{\pm 1.0}$ & $43.70_{\pm 2.0}$ & $6.30_{\pm .80}$ \\
 & ImBD       & \cellcolor{blue!10}$\underline{99.99_{\pm .01}}$ & \cellcolor{blue!10}$\underline{99.90_{\pm .05}}$ & \cellcolor{blue!10}$\underline{99.80_{\pm .05}}$ & \cellcolor{blue!10}$\underline{99.99_{\pm .01}}$ & \cellcolor{blue!10}$\underline{99.90_{\pm .04}}$ & \cellcolor{blue!10}$\underline{99.70_{\pm .05}}$ & \cellcolor{blue!25}$\mathbf{99.89_{\pm .04}}$ & \cellcolor{blue!25}$\mathbf{97.50_{\pm .40}}$ & \cellcolor{blue!25}$\mathbf{94.51_{\pm .50}}$ & $94.39_{\pm .60}$ & $49.70_{\pm 1.8}$ & $27.36_{\pm 1.9}$ \\
 & L2D        & $99.98_{\pm .02}$ & \cellcolor{blue!10}$\underline{99.90_{\pm .05}}$ & $95.60_{\pm .80}$ & $99.80_{\pm .08}$ & $99.30_{\pm .15}$ & $77.60_{\pm 1.2}$ & $93.82_{\pm .90}$ & $77.40_{\pm 1.3}$ & $51.40_{\pm 2.1}$ & $94.47_{\pm .55}$ & $70.80_{\pm 1.5}$ & $2.40_{\pm .50}$ \\
 & RepreGuard & $99.94_{\pm .00}$ & $99.78_{\pm .04}$ & $99.62_{\pm .08}$ & $99.94_{\pm .01}$ & $99.40_{\pm .00}$ & $99.35_{\pm .09}$ & $97.35_{\pm .14}$ & $82.85_{\pm .52}$ & $81.32_{\pm .78}$ & $98.41_{\pm .20}$ & $81.05_{\pm .50}$ & \cellcolor{blue!10}$\underline{74.28_{\pm 2.6}}$ \\
 & \texttt{S2D} (ours) & \cellcolor{blue!25}$\mathbf{100.0_{\pm .00}}$ & \cellcolor{blue!25}$\mathbf{100.0_{\pm .00}}$ & \cellcolor{blue!25}$\mathbf{99.90_{\pm .00}}$ & \cellcolor{blue!25}$\mathbf{99.99_{\pm .00}}$ & \cellcolor{blue!25}$\mathbf{99.90_{\pm .00}}$ & \cellcolor{blue!25}$\mathbf{99.80_{\pm .00}}$ & $99.48_{\pm .09}$ & $95.42_{\pm .41}$ & \cellcolor{blue!10}$\underline{94.50_{\pm .56}}$ & \cellcolor{blue!25}$\mathbf{99.36_{\pm .10}}$ & \cellcolor{blue!25}$\mathbf{91.20_{\pm .65}}$ & \cellcolor{blue!25}$\mathbf{88.63_{\pm .95}}$ \\
\cmidrule{2-14}

\multirow{6}{*}{PaLM} 
 & RoBERTa-L  & \cellcolor{blue!10}$\underline{99.99_{\pm .01}}$ & \cellcolor{blue!10}$\underline{99.90_{\pm .03}}$ & $99.30_{\pm .15}$ & \cellcolor{blue!10}$\underline{99.97_{\pm .01}}$ & \cellcolor{blue!10}$\underline{99.90_{\pm .04}}$ & $80.20_{\pm 1.5}$ & \cellcolor{blue!25}$\mathbf{99.98_{\pm .01}}$ & \cellcolor{blue!10}$\underline{99.60_{\pm .10}}$ & \cellcolor{blue!25}$\mathbf{99.50_{\pm .15}}$ & $99.70_{\pm .10}$ & \cellcolor{blue!10}$\underline{99.70_{\pm .08}}$ & $88.70_{\pm 1.1}$ \\
 & RAIDAR     & $99.96_{\pm .02}$ & \cellcolor{blue!10}$\underline{99.90_{\pm .03}}$ & $71.10_{\pm 1.8}$ & $99.96_{\pm .02}$ & \cellcolor{blue!10}$\underline{99.90_{\pm .04}}$ & $60.90_{\pm 2.0}$ & \cellcolor{blue!10}$\underline{99.97_{\pm .02}}$ & $99.40_{\pm .15}$ & $93.00_{\pm .80}$ & $99.71_{\pm .08}$ & $98.10_{\pm .30}$ & $11.30_{\pm 1.5}$ \\
 & ImBD       & $99.85_{\pm .05}$ & $96.45_{\pm .50}$ & $85.53_{\pm 1.1}$ & $99.93_{\pm .03}$ & $98.20_{\pm .25}$ & $95.32_{\pm .80}$ & $99.78_{\pm .08}$ & $96.00_{\pm .50}$ & $88.01_{\pm 1.1}$ & $93.67_{\pm .70}$ & $44.20_{\pm 1.8}$ & $25.12_{\pm 2.1}$ \\
 & L2D        & \cellcolor{blue!25}$\mathbf{100.0_{\pm .00}}$ & \cellcolor{blue!25}$\mathbf{100.0_{\pm .00}}$ & \cellcolor{blue!25}$\mathbf{100.0_{\pm .00}}$ & $99.92_{\pm .03}$ & \cellcolor{blue!25}$\mathbf{100.0_{\pm .00}}$ & $24.50_{\pm 2.5}$ & \cellcolor{blue!25}$\mathbf{99.98_{\pm .01}}$ & \cellcolor{blue!25}$\mathbf{99.70_{\pm .08}}$ & $95.70_{\pm .60}$ & \cellcolor{blue!25}$\mathbf{99.98_{\pm .01}}$ & \cellcolor{blue!25}$\mathbf{99.90_{\pm .03}}$ & \cellcolor{blue!25}$\mathbf{97.70_{\pm .40}}$ \\
 & RepreGuard & $99.92_{\pm .00}$ & $99.58_{\pm .13}$ & \cellcolor{blue!10}$\underline{99.53_{\pm .04}}$ & $99.94_{\pm .00}$ & $99.62_{\pm .11}$ & \cellcolor{blue!10}$\underline{99.50_{\pm .12}}$ & $96.95_{\pm .17}$ & $85.38_{\pm .80}$ & $84.08_{\pm .81}$ & $96.28_{\pm .50}$ & $76.97_{\pm 2.2}$ & $72.40_{\pm 2.0}$ \\
 & \texttt{S2D} (ours) & \cellcolor{blue!25}$\mathbf{100.0_{\pm .00}}$ & \cellcolor{blue!25}$\mathbf{100.0_{\pm .00}}$ & \cellcolor{blue!25}$\mathbf{100.0_{\pm .00}}$ & \cellcolor{blue!25}$\mathbf{99.99_{\pm .00}}$ & \cellcolor{blue!25}$\mathbf{100.0_{\pm .00}}$ & \cellcolor{blue!25}$\mathbf{99.60_{\pm .00}}$ & $99.88_{\pm .02}$ & \cellcolor{blue!10}$\underline{98.25_{\pm .40}}$ & \cellcolor{blue!10}$\underline{97.60_{\pm .52}}$ & \cellcolor{blue!10}$\underline{99.74_{\pm .11}}$ & $94.52_{\pm 1.6}$ & \cellcolor{blue!10}$\underline{92.60_{\pm 1.8}}$ \\
\cmidrule{2-14}

\multirow{6}{*}{Claude} 
 & RoBERTa-L  & $80.91_{\pm 1.5}$ & $14.00_{\pm 1.2}$ & $9.70_{\pm .80}$ & $80.26_{\pm 1.4}$ & $18.20_{\pm 1.6}$ & $16.30_{\pm 1.5}$ & $74.82_{\pm 1.8}$ & $26.95_{\pm 1.8}$ & $25.00_{\pm 2.0}$ & $99.97_{\pm .01}$ & $99.80_{\pm .05}$ & \cellcolor{blue!10}$\underline{99.60_{\pm .10}}$ \\
 & RAIDAR     & $99.91_{\pm .03}$ & $99.70_{\pm .10}$ & $46.70_{\pm 2.2}$ & $97.32_{\pm .50}$ & $88.90_{\pm 1.2}$ & $56.20_{\pm 2.0}$ & $93.60_{\pm .80}$ & $73.50_{\pm 1.4}$ & $20.20_{\pm 1.5}$ & \cellcolor{blue!10}$\underline{99.98_{\pm .01}}$ & \cellcolor{blue!10}$\underline{99.90_{\pm .04}}$ & $94.70_{\pm .80}$ \\
 & ImBD       & \cellcolor{blue!25}$\mathbf{99.99_{\pm .01}}$ & \cellcolor{blue!10}$\underline{99.90_{\pm .04}}$ & \cellcolor{blue!10}$\underline{99.71_{\pm .10}}$ & \cellcolor{blue!10}$\underline{99.97_{\pm .02}}$ & \cellcolor{blue!10}$\underline{99.90_{\pm .04}}$ & $84.31_{\pm 1.5}$ & \cellcolor{blue!25}$\mathbf{99.88_{\pm .04}}$ & \cellcolor{blue!25}$\mathbf{97.40_{\pm .35}}$ & \cellcolor{blue!25}$\mathbf{91.60_{\pm .80}}$ & $99.05_{\pm .15}$ & $86.70_{\pm 1.1}$ & $60.14_{\pm 1.8}$ \\
 & L2D        & $99.92_{\pm .02}$ & \cellcolor{blue!10}$\underline{99.90_{\pm .04}}$ & $82.20_{\pm 1.5}$ & $99.90_{\pm .05}$ & $99.80_{\pm .08}$ & \cellcolor{blue!10}$\underline{99.10_{\pm .20}}$ & $97.88_{\pm .40}$ & $92.20_{\pm .80}$ & $79.80_{\pm 1.3}$ & $99.96_{\pm .02}$ & \cellcolor{blue!10}$\underline{99.90_{\pm .04}}$ & $97.60_{\pm .40}$ \\
 & RepreGuard & $99.92_{\pm .00}$ & $99.47_{\pm .11}$ & $99.18_{\pm .19}$ & $99.89_{\pm .01}$ & $98.82_{\pm .26}$ & $98.10_{\pm .19}$ & $95.08_{\pm .27}$ & $70.38_{\pm .93}$ & $68.45_{\pm .54}$ & $98.42_{\pm .26}$ & $85.85_{\pm .62}$ & $81.47_{\pm 2.0}$ \\
 & \texttt{S2D} (ours) & \cellcolor{blue!10}$\underline{99.98_{\pm .03}}$ & \cellcolor{blue!25}$\mathbf{99.98_{\pm .04}}$ & \cellcolor{blue!25}$\mathbf{99.93_{\pm .04}}$ & \cellcolor{blue!25}$\mathbf{99.99_{\pm .01}}$ & \cellcolor{blue!25}$\mathbf{100.0_{\pm .00}}$ & \cellcolor{blue!25}$\mathbf{99.98_{\pm .04}}$ & \cellcolor{blue!10}$\underline{98.92_{\pm .20}}$ & \cellcolor{blue!10}$\underline{92.58_{\pm .69}}$ & \cellcolor{blue!10}$\underline{89.45_{\pm 1.6}}$ & \cellcolor{blue!25}$\mathbf{99.99_{\pm .00}}$ & \cellcolor{blue!25}$\mathbf{99.95_{\pm .09}}$ & \cellcolor{blue!25}$\mathbf{99.95_{\pm .09}}$ \\
\bottomrule
\end{tabular}}
\caption{\textbf{In-distribution and out-of-distribution performance comparison.} Methods are categorized into train-free and train-based settings. Our proposed \texttt{S2D} demonstrates robust generalization, achieving state-of-the-art or highly competitive results across the majority of evaluated scenarios. Best and second-best results are highlighted in \colorbox{blue!25}{\textbf{darker}} and \colorbox{blue!10}{\underline{lighter}} blue, respectively.}
\label{tab:main_experiment}
\vspace{-0.1in}
\end{table*}

\begin{table*}[!t]
\vspace{-5pt}
\centering
\small
\setlength{\tabcolsep}{3.2pt}
\renewcommand{\arraystretch}{1.2}
\resizebox{\textwidth}{!}{%
\begin{tabular}{
l
cc | cc cc cc | cc cc cc
}
\toprule
\multirow{2}{*}{\textbf{Detector} $\downarrow$}
& \multicolumn{2}{c|}{\textbf{No Attack}}
& \multicolumn{6}{c|}{\textbf{Paraphrasing Attacks}}
& \multicolumn{6}{c}{\textbf{Perturbations}} \\

\cmidrule(lr){2-3} \cmidrule(lr){4-9} \cmidrule(lr){10-15}

&
&
& \multicolumn{2}{c}{Polish}
& \multicolumn{2}{c}{Back Trans.}
& \multicolumn{2}{c}{DIPPER}
& \multicolumn{2}{c}{Character}
& \multicolumn{2}{c}{Word}
& \multicolumn{2}{c}{Char + Word} \\

\cmidrule(lr){4-5} \cmidrule(lr){6-7} \cmidrule(lr){8-9} \cmidrule(lr){10-11} \cmidrule(lr){12-13} \cmidrule(lr){14-15}

& \textit{AUROC} & \textit{TPR@1\%}
& \textit{AUROC} & \textit{TPR@1\%}
& \textit{AUROC} & \textit{TPR@1\%}
& \textit{AUROC} & \textit{TPR@1\%}
& \textit{AUROC} & \textit{TPR@1\%}
& \textit{AUROC} & \textit{TPR@1\%}
& \textit{AUROC} & \textit{TPR@1\%} \\

\midrule
Binoculars    & 95.40 & 62.80
& 73.54 \textcolor{red}{\scriptsize ($-$22.9\%)} & 0.53 \textcolor{red}{\scriptsize ($-$99.2\%)}
& 90.93 \textcolor{red}{\scriptsize ($-$4.7\%)} & 53.73 \textcolor{red}{\scriptsize ($-$14.4\%)}
& 93.67 \textcolor{red}{\scriptsize ($-$1.8\%)} & 58.14 \textcolor{red}{\scriptsize ($-$7.4\%)}
& 94.76 \textcolor{red}{\scriptsize ($-$0.7\%)} & 57.73 \textcolor{red}{\scriptsize ($-$8.1\%)}
& 94.23 \textcolor{red}{\scriptsize ($-$1.2\%)} & 58.73 \textcolor{red}{\scriptsize ($-$6.5\%)}
& 95.01 \textcolor{red}{\scriptsize ($-$0.4\%)} & 61.07 \textcolor{red}{\scriptsize ($-$2.8\%)} \\

L2D   & 95.80 & \cellcolor{blue!10}\underline{88.33}
& 56.14 \textcolor{red}{\scriptsize ($-$41.4\%)} & 1.20 \textcolor{red}{\scriptsize ($-$98.6\%)}
& \cellcolor{blue!10}\underline{92.36} \textcolor{red}{\scriptsize ($-$3.6\%)} & 15.73 \textcolor{red}{\scriptsize ($-$82.2\%)}
& 87.72 \textcolor{red}{\scriptsize ($-$8.4\%)} & 3.20 \textcolor{red}{\scriptsize ($-$96.4\%)}
& 83.43 \textcolor{red}{\scriptsize ($-$12.9\%)} & 15.20 \textcolor{red}{\scriptsize ($-$82.8\%)}
& \cellcolor{blue!10}\underline{95.70} \textcolor{red}{\scriptsize ($-$0.1\%)} & \cellcolor{blue!10}\underline{80.13} \textcolor{red}{\scriptsize ($-$9.3\%)}
& 94.21 \textcolor{red}{\scriptsize ($-$1.7\%)} & 42.13 \textcolor{red}{\scriptsize ($-$52.3\%)} \\

RepreGuard    & \cellcolor{blue!10}\underline{97.02} & 85.30
& \cellcolor{blue!25}\textbf{83.88} \textcolor{red}{\scriptsize ($-$13.5\%)} & 10.53 \textcolor{red}{\scriptsize ($-$87.7\%)}
& 91.67 \textcolor{red}{\scriptsize ($-$5.5\%)} & 45.33 \textcolor{red}{\scriptsize ($-$46.9\%)}
& 94.23 \textcolor{red}{\scriptsize ($-$2.9\%)} & 73.67 \textcolor{red}{\scriptsize ($-$13.6\%)}
& \cellcolor{blue!10}\underline{96.32} \textcolor{red}{\scriptsize ($-$0.7\%)} & \cellcolor{blue!10}\underline{83.87} \textcolor{red}{\scriptsize ($-$1.7\%)}
& 92.88 \textcolor{red}{\scriptsize ($-$4.3\%)} & 78.67 \textcolor{red}{\scriptsize ($-$7.8\%)}
& 94.29 \textcolor{red}{\scriptsize ($-$2.8\%)} & \cellcolor{blue!10}\underline{85.00} \textcolor{red}{\scriptsize ($-$0.4\%)} \\

\texttt{S2D} (ours)
& \cellcolor{blue!25}\textbf{98.87} & \cellcolor{blue!25}\textbf{96.68}
& 82.06 \textcolor{red}{\scriptsize ($-$17.0\%)} & \cellcolor{blue!25}\textbf{23.60} \textcolor{red}{\scriptsize ($-$75.6\%)}
& \cellcolor{blue!25}\textbf{93.45} \textcolor{red}{\scriptsize ($-$5.5\%)} & \cellcolor{blue!25}\textbf{67.33} \textcolor{red}{\scriptsize ($-$30.4\%)}
& \cellcolor{blue!25}\textbf{98.66} \textcolor{red}{\scriptsize ($-$0.2\%)} & \cellcolor{blue!25}\textbf{84.91} \textcolor{red}{\scriptsize ($-$12.2\%)}
& \cellcolor{blue!25}\textbf{97.76} \textcolor{red}{\scriptsize ($-$1.1\%)} & \cellcolor{blue!25}\textbf{90.53} \textcolor{red}{\scriptsize ($-$6.4\%)}
& \cellcolor{blue!25}\textbf{96.73} \textcolor{red}{\scriptsize ($-$2.2\%)} & \cellcolor{blue!25}\textbf{90.87} \textcolor{red}{\scriptsize ($-$6.0\%)}
& \cellcolor{blue!25}\textbf{95.77} \textcolor{red}{\scriptsize ($-$3.1\%)} & \cellcolor{blue!25}\textbf{91.47} \textcolor{red}{\scriptsize ($-$5.4\%)} \\
\bottomrule
\end{tabular}}
\caption{\textbf{Robustness evaluation.} “No Attack” denotes performance on original text, and values in parentheses indicate the relative percentage change $\left(\left(\text{Attack} - \text{No attack}\right) / \text{No attack}\times 100\%\right)$ compared to this baseline (positive values indicate performance gains under attack). Best and second-best results are highlighted in \colorbox{blue!25}{\textbf{darker}} and \colorbox{blue!10}{\underline{lighter}} blue, respectively.}
\label{tab:attack_robustness}
\vspace{-0.1in}
\end{table*}

\vspace{-0.1in}
\paragraph{Robustness to Adversarial Attacks.}
We then evaluate robustness under adversarial manipulations, including three paraphrasing attacks (i.e., \textit{Polish}, \textit{Back Translation}, \textit{DIPPER}) and three perturbation attacks (i.e., \textit{Character}, \textit{Word}, and \textit{Character+Word}), following the previous setup~\citep{bao2023fast}. As shown in Table~\ref{tab:attack_robustness}, paraphrasing attacks substantially degrade all detectors, highlighting the challenge of semantic-preserving rewrites. Nevertheless, \texttt{S2D} is highly competitive and often achieves the best or near-best performance, whereas baseline methods experience severe degradation. This robustness stems from \texttt{S2D}’s representation reshaping, which enforces a more discriminative feature space and preserves separability even under attacks. In contrast, perturbation-based attacks have a much milder impact, as they primarily introduce local or surface-level noise.
\paragraph{Robustness to Text Length.}
Our main test set has an average length of approximately $267$ tokens, with $95\%$ of samples shorter than $435$ tokens. Figure~\ref{fig:text_length} shows that detection generally benefits from additional context within this range. We further examine substantially longer sequences, up to $4096$ tokens, using a separate diagnostic set. Performance continues to improve with context length before plateauing, while the score overlap between human and LLM-generated text remains low. In particular, we observe no evidence that repeatedly injecting the same steering vector causes representation collapse or performance degradation within the evaluated range. Full long-context results are reported in Appendix~\ref{app:long_context}.

\vspace{-0.1in}
\paragraph{Comparison with Watermarking Methods.}
An interesting question is whether \texttt{S2D} can outperform existing watermarking methods. Although this comparison is inherently unfair---since watermarking methods actively embed detectable signals while \texttt{S2D} is a passive detector---it helps position \texttt{S2D} relative to these approaches.\footnote{To reflect a realistic setting, \texttt{S2D} is trained once on human-written and non-watermarked LLM-generated text, without retraining for specific watermarking schemes, and is directly applied at evaluation.}

\begin{figure}[t]
\centering
\begin{subfigure}{0.32\linewidth}
\centering
\includegraphics[width=\linewidth]{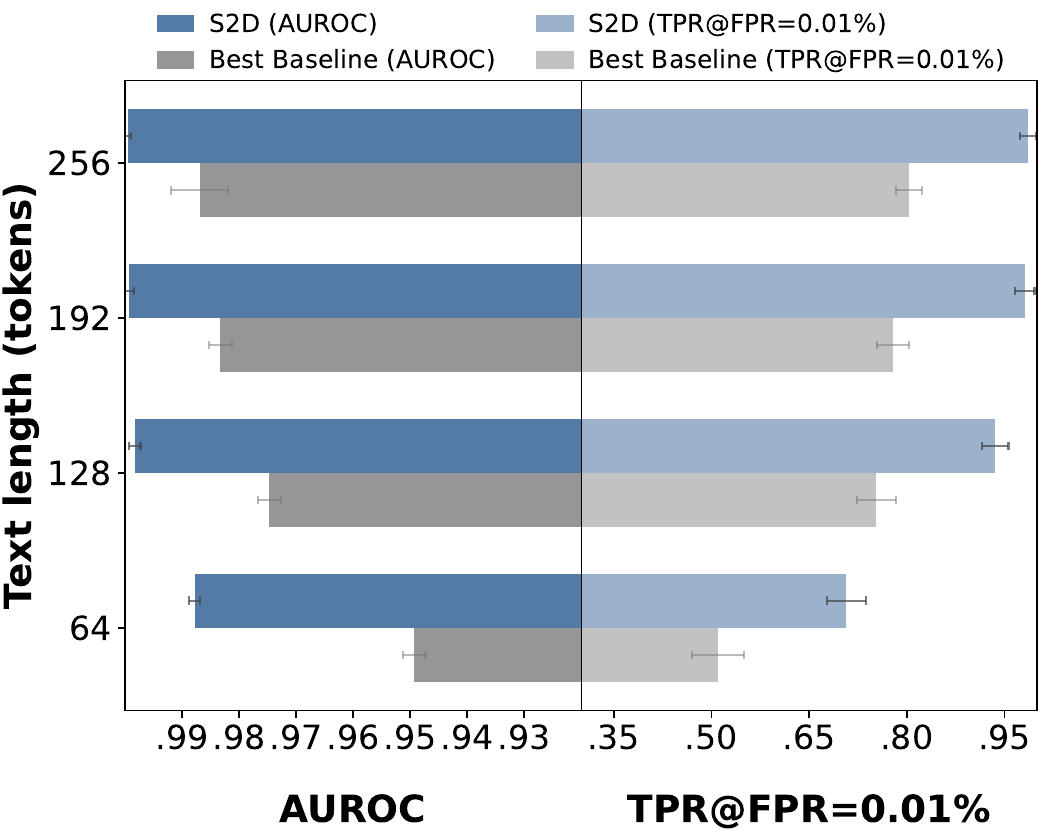}
\caption{Robustness to text length}
\label{fig:text_length}
\end{subfigure}
\begin{subfigure}{0.32\linewidth}
\centering
\includegraphics[width=\linewidth]{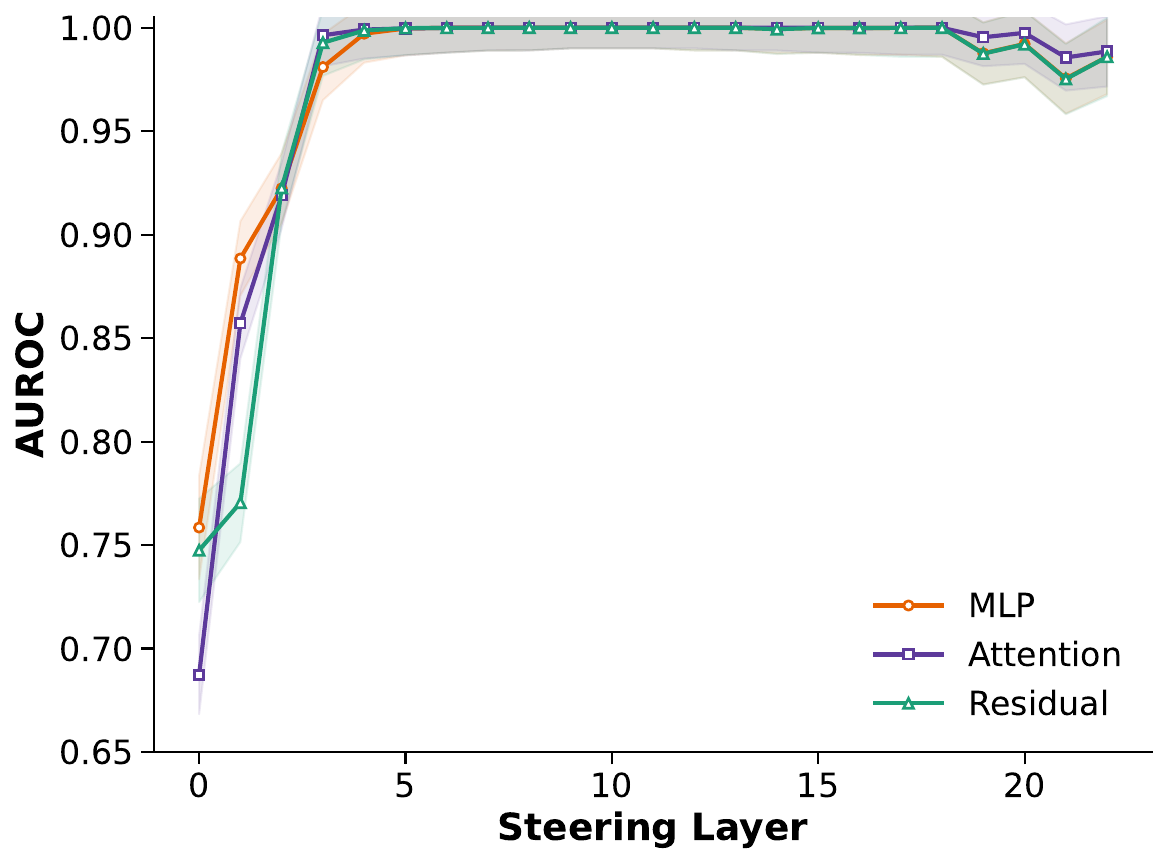}
\caption{Impact of steering position}
\label{fig:steer_position}
\end{subfigure}
\hfill
\begin{subfigure}{0.32\linewidth}
\centering
\includegraphics[width=\linewidth]{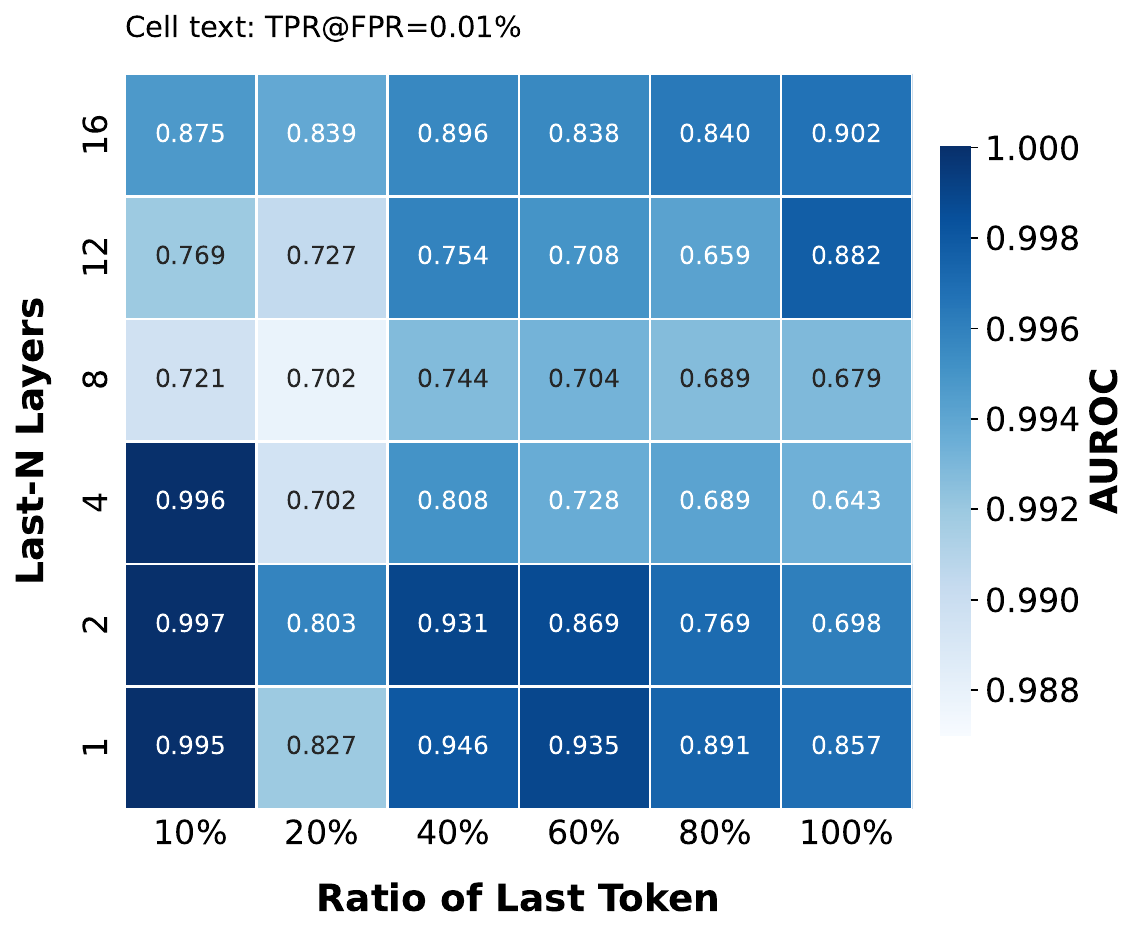}
\caption{Token ratio vs. Layer index}
\label{fig:ratio_number}
\end{subfigure}
\caption{
    \textbf{Analysis of \texttt{S2D} performance.} 
    (a) Comparison of detection stability across varying input lengths.
    (b) Detection performance across steering layers, showing that intermediate layers consistently achieve the best performance.
    (c) Performance heatmap as a function of last-token selection ratio and the number of aggregated layers. 
}
\label{fig:detailed_analysis}
\end{figure}


Table~\ref{tab:watermark_length_temp} compares \texttt{S2D} with three watermarking methods, namely Gumbel-max~\cite{aaronson2023watermarking}, Green-red~\cite{kirchenbauer2023watermark}, and SynthID~\cite{dathathri2024scalable}, using their default detectors, which have full knowledge of the watermarking mechanisms. Interestingly, \texttt{S2D} achieves stronger performance at lower temperature ($T=0.3$), where watermark signals are weaker due to reduced generation entropy. In contrast, watermarking methods perform better at higher temperatures, where the embedded signals are more pronounced.
The performance of \texttt{S2D} in this setting remains reasonable: even without knowledge of the watermarking scheme, it still captures distributional artifacts induced by the watermarking generation process. These artifacts are more pronounced in low-entropy settings, where watermark signals are weak, while at higher entropy levels, stronger embedded watermark signals favor watermark-specific detectors.

\begin{table}[t]
\vspace{-0.1in}
\centering
\small
\setlength{\tabcolsep}{3.5pt} 
\renewcommand{\arraystretch}{1.3}
\resizebox{\textwidth}{!}{%
\begin{tabular}{c c ccc ccc ccc}
\toprule
\multirow{3}{*}{\bf Temp.} & \multirow{3}{*}{\bf Method}
& \multicolumn{3}{c}{\bf 64 tokens}
& \multicolumn{3}{c}{\bf 128 tokens} 
& \multicolumn{3}{c}{\bf 256 tokens} \\
\cmidrule(lr){3-5} \cmidrule(lr){6-8} \cmidrule(lr){9-11}
& 
& Gumbel & Green-red & SynthID
& Gumbel & Green-red & SynthID
& Gumbel & Green-red & SynthID \\
& & \multicolumn{9}{c}{\textit{AUROC / TPR@1\%FPR (\%)}} \\
\midrule

\multirow{2}{*}{0.3}
& Watermark Det.$^\dagger$
& \bf 99.9 / 98.0 &  82.0 / 12.0  & 81.8 / 15.9
& \bf 99.9 / 99.7 &  \textbf{99.6} / \textbf{95.3}  & 84.8 / 28.0
& 99.9 / {\bf 99.9} &  99.8 /{\bf 99.4} & 84.5 / 34.0 \\

& \texttt{S2D} (ours)
& 77.6 / 20.0 & \bf 96.7 / 60.7 & \bf 85.1 / 26.0
& 96.6 / 62.8 & 98.4 / 81.3 & \bf 96.5 / 69.0
& {\bf 100} / 98.5 & {\bf 99.9} / 98.6 & \bf 99.9 / 98.1 \\
\midrule

\multirow{2}{*}{0.7}
& Watermark Det.$^\dagger$
& \bf 99.9 / 97.0 &  80.2 / 9.3 & \bf 99.3 / 89.8
& \bf 99.9 / 99.8 &  \bf 99.6 / 95.7  & \bf 99.7 / 95.8
& \bf 99.9 / 99.9 &  \bf 100 / 99.8  & \bf 99.5 / 97.5 \\

& \texttt{S2D} (ours)
& 65.4 / 12.3 & \bf 90.5 / 33.0 & 70.0 / 10.9
& 89.3 / 44.7 & 89.0 / 41.7 & 84.2 / 32.7
& 99.6 / 91.7 & 95.9 / 70.4  & 97.6 / 79.0 \\

\bottomrule
\multicolumn{11}{l}{\footnotesize $^\dagger$ The defaulted detectors provided by the original watermarking algorithms.}
\end{tabular}}
\vspace{4pt}
\caption{\textbf{Detection performance across different text lengths and temperatures.} Comparison with default watermark detectors. \textbf{Bold} indicates better performance within each setting.}
\label{tab:watermark_length_temp}
\vspace{-0.1in}
\end{table}

\vspace{-0.1in}
\subsection{Ablation Study}
\label{understanding_s2d}

\paragraph{Importance of Steering.}
We assess the role of steering by comparing \texttt{S2D} with a variant without steering, which does not incorporate the steering vector and instead estimates the vMF parameters directly from training representations. As shown in Figure~\ref{fig:main_analysis}, \texttt{S2D w/o steering} exhibits substantially weaker separation between human- and LLM-generated texts, with significant overlap between the two distributions. This highlights the importance of incorporating the steering vector.

\vspace{-0.1in}
\vspace{-0.1in}
\paragraph{Steering Position.}
We first analyze the impact of the steering position. As shown in Figure~\ref{fig:steer_position}, for the default \textit{Llama-3.1-8B} observer, intermediate-layer injection performs substantially better than very early injection, while the exact insertion point within a Transformer block (i.e., residual stream, attention output, or MLP output) has a much smaller effect.

We further examine how this preferred region changes with model depth by performing steering-layer sweeps on \textit{Llama-3.2-1B}, \textit{Meta-Llama-3-8B}, and \textit{Meta-Llama-3-70B}. Across all three scales, strong performance is attained over a range of intermediate depths rather than at a single isolated layer. However, the best-performing layer does not follow a fixed proportional scaling law with network depth. These results suggest that intermediate depth is a useful search region, while the exact steering layer remains model dependent. Full cross-scale results are provided in Appendix~\ref{app:steering_depth_scaling}.

\vspace{-0.1in}
\paragraph{Token Selection and Layer Aggregation.}
We analyze the effect of the last-token ratio $K$ and the number of aggregated layers $N$ (introduced in Section \ref{sec:representation}). As shown in Figure~\ref{fig:ratio_number}, performance exhibits a non-monotonic pattern across both factors. Strong performance is observed in two regimes: the bottom-left region, which uses a small fraction of tokens from a few top layers, and the top-right region, where deeper aggregation is combined with larger token ratios. In contrast, intermediate configurations perform worse.
This pattern suggests a trade-off between signal concentration and aggregation. Detection-relevant signals are concentrated in the final layers and later tokens. While incorporating more layers or tokens can be beneficial when aggregation is sufficiently large and helps stabilize the representation, intermediate settings tend to introduce noise without sufficient complementary signal, leading to degraded performance.

\vspace{-0.1in}

\paragraph{Effects of Different Observer Models.}
As shown in Table~\ref{tab:observer_avg}, the performance of \texttt{S2D} depends noticeably on the choice of observer model. Several observers, including \textit{Falcon-7B}, \textit{Qwen2.5-7B}, and \textit{Llama-3.1-8B}, yield near-perfect AUROC and strong TPR at low FPRs, while others, such as \textit{Gemma-2}, perform substantially worse. Importantly, this variation is not explained by model scale, recency, or instruction tuning alone, e.g., the lightweight \textit{OPT-2.7B} remains highly competitive.

This dependence is consistent with the design of \texttt{S2D}: the method operates on the hidden representations of a frozen observer and can amplify discriminative structure already present in its latent space, but its effectiveness is limited when human- and LLM-generated representations remain strongly entangled. We therefore do not expect all observer models to perform equally well. In practice, we recommend selecting the observer using a small held-out validation set, with pre-steering representation overlap serving as a useful diagnostic of observer suitability.

\vspace{-0.1in}
\paragraph{Shots of Training Dataset.}
Finally, we investigate the impact of training set size on detection performance using \texttt{S2D}, RepreGuard, and L2D, since these three methods yielded the best overall results in the preceding experiments. As shown in Figure~\ref{fig:main_analysis} (right), \texttt{S2D} achieves strong performance even in the low-shot regime, while RepreGuard shows moderate sensitivity, and L2D degrades substantially with limited data. This indicates that \texttt{S2D} can leverage intrinsic differences between human-written and LLM-generated text in the representation space, with training mainly serving to reduce distributional overlap. Detailed results are provided in the Appendix~\ref{app_shots_train}.

\begin{figure}[t!]
\vspace{-8pt}
\centering
\begin{minipage}[t]{0.58\textwidth}
\vspace{0pt}
\centering
\begin{subfigure}[t]{0.49\linewidth}
\centering
\includegraphics[width=\linewidth]{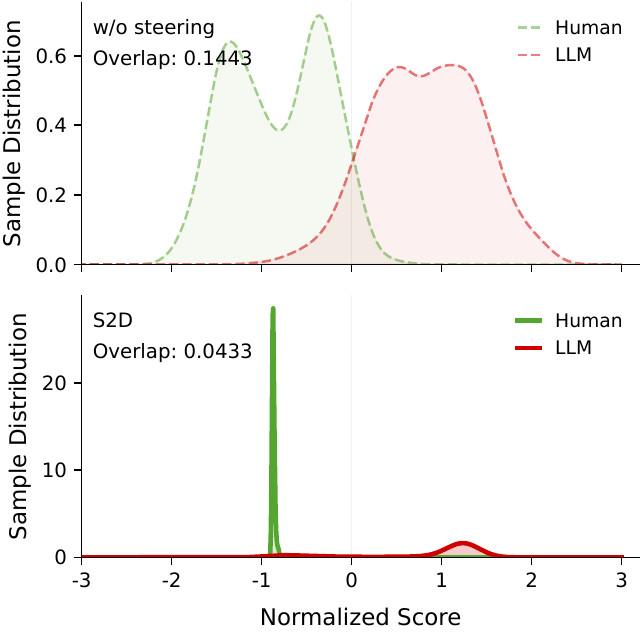}
\end{subfigure}
\hfill
\begin{subfigure}[t]{0.49\linewidth}
\centering
\includegraphics[width=\linewidth]{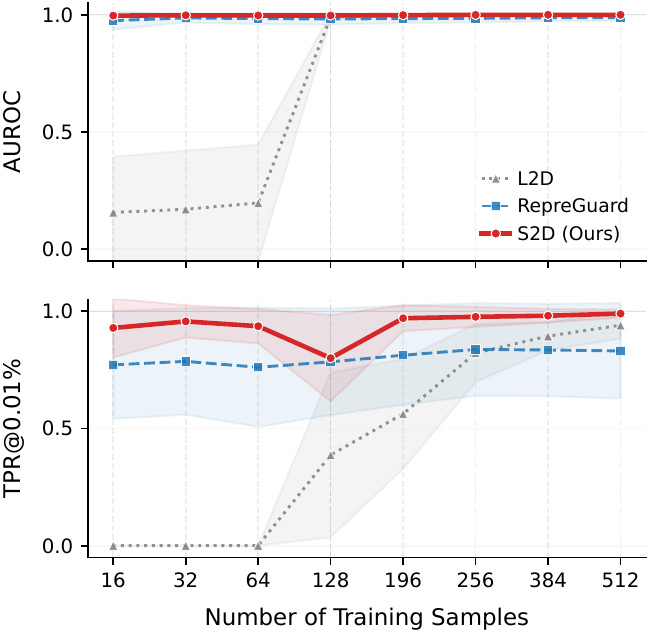}
\end{subfigure}
\vspace{-4pt}
\captionof{figure}{\textbf{Detection analysis.} 
Left: Steering leads to better separability.
Right: Detection performance across training sizes.
Full results are in Figures \ref{fig_score_distribution_grid} and \ref{fig:train_shots} in Appendix.
}
\label{fig:main_analysis}
\end{minipage}%
\hfill
\begin{minipage}[t]{0.4\textwidth}
\vspace{4pt}
\centering
\footnotesize
\setlength{\tabcolsep}{2.5pt}
\renewcommand{\arraystretch}{1.1}
\resizebox{\linewidth}{!}{
\begin{tabular}{lccc}
\toprule
\textbf{Observer Model} & \textbf{AUROC} & \textbf{TPR@1\%} & \textbf{TPR@.01\%} \\
\midrule
Llama-3.1-8B        & 99.62 \tiny$\pm$ 0.46 & 99.05 \tiny$\pm$ 0.60 & 97.98 \tiny$\pm$ 1.37 \\
Mistral-7B-v0.3     & 68.52 \tiny$\pm$ 0.94 & 30.53 \tiny$\pm$ 5.31 & 23.28 \tiny$\pm$ 6.45 \\
GPT-Neo-2.7B        & 82.38 \tiny$\pm$ 2.94 & 32.65 \tiny$\pm$ 8.91 & 11.33 \tiny$\pm$ 5.67 \\
OPT-2.7B            & 99.86 \tiny$\pm$ 0.13 & 98.30 \tiny$\pm$ 1.79 & 84.50 \tiny$\pm$ 8.24 \\
Qwen2.5-7B          & 99.93 \tiny$\pm$ 0.13 & 99.58 \tiny$\pm$ 0.65 & 98.65 \tiny$\pm$ 1.69 \\
Falcon-7B           & 99.96 \tiny$\pm$ 0.05 & 99.40 \tiny$\pm$ 0.88 & 98.28 \tiny$\pm$ 1.83 \\
Falcon-7B-Inst.     & 99.96 \tiny$\pm$ 0.05 & 99.18 \tiny$\pm$ 1.45 & 95.90 \tiny$\pm$ 2.68 \\
Gemma-2-9B          & 76.54 \tiny$\pm$ 14.6 & 19.25 \tiny$\pm$ 5.37 & 3.05 \tiny$\pm$ 2.89  \\
Gemma-2-9B-Inst.    & 79.09 \tiny$\pm$ 20.3 & 37.73 \tiny$\pm$ 26.6 & 18.90 \tiny$\pm$ 14.1 \\
\bottomrule
\end{tabular}
}
\captionof{table}{\textbf{Observer stability.} Detection performance (mean $\pm$ std), averaged across four training datasets. Full results are in Table~\ref{tab:observer_models_strict_sh} in Appendix.}
\label{tab:observer_avg}
\end{minipage}
\vspace{-8pt}
\end{figure}

\vspace{-0.10in}
\section{Discussion}\label{sec:discussion}
\vspace{-0.10in}

We present \texttt{S2D}, a representation-based framework for detecting LLM-generated text. By injecting a lightweight steering vector into hidden representations, \texttt{S2D} improves class separability and formulates detection as a hypothesis test. The resulting detector provides Type-I error control and a bound on excess Type-II error. Experiments demonstrate strong performance across ID and OOD settings, as well as robustness to adversarial perturbations, paraphrasing, varying text lengths, and different observer models. \texttt{S2D} also achieves a favorable trade-off between detection performance and computational efficiency compared with existing methods.

Several limitations suggest directions for future work. First, \texttt{S2D} is a document-level binary detector and does not localize human- and AI-written segments within mixed-authorship text. Second, a single vMF model may not fully capture latent geometry across domains; mixture models or target-domain adaptation may improve robustness. Finally, more flexible steering mechanisms, such as multiple vectors or non-additive interventions, could better capture representation geometry.

\bibliographystyle{unsrtnat}   
\bibliography{ref}  

\newpage
\appendix

\section{Algorithm}\label{app_algo}
\begin{algorithm}[H]
\caption{Overall training pipeline for \texttt{S2D}}
\label{alg:vmf_latent_steering}
\KwIn{
Frozen observer model $f_\theta$, training set $\mathcal{S}_{\text{train}}$, null calibration set $\mathcal{S}_{\text{cal}}$ (human-written text only);
steering layer $\ell_s$; inverse temperature $\beta$; 
EMA coefficient $\rho$; learning rate $\eta$; epochs $E$; batch size $B$.
}
\KwOut{
Steering vector $\mathbf{v}$, class mean directions $\widehat{\boldsymbol{\mu}}_0,\widehat{\boldsymbol{\mu}}_1$, and calibrated threshold $\widehat{\tau}$.
}
\BlankLine
\textbf{Initialize:} $\mathbf{v}\leftarrow \mathbf{0}\in\mathbb{R}^d$, and uniformly sample $\widehat{\boldsymbol{\mu}}_c \in \mathbb{S}^{d-1}$ for $c\in\{0,1\}$\;

\BlankLine
\For{$e=1$ \KwTo $E$}{
  Shuffle $\mathcal{S}_{\text{train}}$ and divide into mini-batches\;
  \For{each mini-batch $\mathcal{B} \subset \mathcal{S}_{\text{train}}$}{
    \textbf{(1) Forward with steering:} For all $x\in\mathcal{B}$, inject $\mathbf{v}$ at layer $\ell_s$:
    \[
    \tilde{h}_{\ell_s,t}(x) \leftarrow h_{\ell_s,t}(x)+  \mathbf{v},\quad \text{for all valid token position $t$}.
    \]
    Extract the steered representation $f_{\theta,\mathbf{v}}(x) \in \mathbb{S}^{d-1}$ as described in Section~\ref{sec_repre_ext}\;
    
    \textbf{(2) Temperature-Scaled Objective:} Compute the batch objective using the softmax form in~\eqref{eq:postvMF}:
    \[
    \mathcal{L}_{\mathcal{B}}(\mathbf{v})\leftarrow
    \frac{1}{|\mathcal{B}|}\sum_{(x,y)\in\mathcal{B}}
    \log q_{\beta}\left(y\mid
    f_{\theta,\mathbf{v}}(x),
    \widehat{\boldsymbol{\mu}}_0,
    \widehat{\boldsymbol{\mu}}_1\right).
    \]

    \textbf{(3) Steering Update:} Take a gradient update step: $\mathbf{v} \leftarrow \mathbf{v} + \eta\,\nabla_{\mathbf{v}} \mathcal{L}_{\mathcal{B}}(\mathbf{v})$\;

    \textbf{(4) Mean Directions Update:} For $c\in\{0,1\}$, compute the batch mean representation $\bar{\mathbf{z}}_c(\mathbf{v})$ and apply EMA:
    \[
    \widehat{\boldsymbol{\mu}}_c \leftarrow \frac{(1-\rho)\widehat{\boldsymbol{\mu}}_c + \rho\,\bar{\mathbf{z}}_c(\mathbf{v})}{\|(1-\rho)\widehat{\boldsymbol{\mu}}_c + \rho\,\bar{\mathbf{z}}_c(\mathbf{v})\|},
    \]
    where $\bar{\mathbf{z}}_c(\mathbf{v}) = \frac{1}{|\mathcal{B}_c|}\sum_{x\in \mathcal{B}_c} f_{\theta,\mathbf{v}}(x)$, and $\mathcal{B}_c = \{(x, y) \in \mathcal{B} : y=c\}$ denotes the subset of samples in the mini-batch belonging to class $c$, with $\mathcal{B}=\mathcal{B}_0 \cup \mathcal{B}_1$.
  }
}


\BlankLine
\textbf{Threshold Calibration:}
\For{each $x \in \mathcal{S}_{\text{cal}}$}{
  Compute the detection score $\widehat{\mathcal{S}}(x) =(\widehat{\boldsymbol{\mu}}_1 - \widehat{\boldsymbol{\mu}}_0)^\top f_{\theta,\mathbf{v}}(x)$\;
}
Set the threshold $\widehat{\tau}$ using $\mathcal{S}_{\text{cal}}$ according to the chosen calibration scheme\;

\BlankLine
\Return{$(\mathbf{v},\widehat{\boldsymbol{\mu}}_0,\widehat{\boldsymbol{\mu}}_1, \widehat{\tau})$}

\end{algorithm}

\section{Derivation of the Discriminative Posterior Probability}
\label{app:derivation}
In this section, we derive the class posterior probability 
$p(y_i \mid f_{\theta, \mathbf{v}}(x_i), \boldsymbol{\mu}_0, \boldsymbol{\mu}_1)$ 
under the vMF model. 
In the vMF model, given a text $x_i$ with label $y_i = c \in \{0,1\}$, its steered representation $f_{\theta, \mathbf{v}}(x_i)$ follows the following distribution on the unit hypersphere $\mathbb{S}^{d-1}$, with likelihood
\begin{equation*}
    p(f_{\theta, \mathbf{v}}(x_i) \mid y_i = c)
    =
    C_d(\kappa)\exp\left(\kappa \boldsymbol{\mu}_c^\top f_{\theta, \mathbf{v}}(x_i)\right),
\end{equation*}
where $C_d(\kappa) = \frac{\kappa^{d/2-1}}{(2\pi)^{d/2} I_{d/2-1}(\kappa)}$ is the normalization constant and $\boldsymbol{\mu}_c$ is the class mean direction. We assume a shared concentration parameter $\kappa$ across classes and a uniform prior $p(y_i = 0) = p(y_i = 1) = \frac{1}{2}$.
By Bayes' theorem,
\begin{equation*}
    p(y_i \mid f_{\theta, \mathbf{v}}(x_i))
    =
    \frac{p(f_{\theta, \mathbf{v}}(x_i) \mid y_i)p(y_i)}
    {\sum_{c \in \{0,1\}} p(f_{\theta, \mathbf{v}}(x_i) \mid y_i = c)p(y_i = c)}.
\end{equation*}
Substituting the vMF likelihood and the uniform prior gives
\begin{equation*}
    p(y_i \mid f_{\theta, \mathbf{v}}(x_i))
    =
    \frac{C_d(\kappa)\exp\left(\kappa \boldsymbol{\mu}_{y_i}^\top f_{\theta, \mathbf{v}}(x_i)\right)\cdot \frac{1}{2}}
    {\sum_{c \in \{0,1\}} C_d(\kappa)\exp\left(\kappa \boldsymbol{\mu}_c^\top f_{\theta, \mathbf{v}}(x_i)\right)\cdot \frac{1}{2}}.
\end{equation*}

Since $C_d(\kappa)$ and the prior $\frac{1}{2}$ are identical across classes, they cancel out, yielding the class posterior
\begin{equation*}
    p\left(y_i \mid f_{\theta, \mathbf{v}}(x_i), \boldsymbol{\mu}_0, \boldsymbol{\mu}_1\right)
    =
    \frac{\exp\left(\kappa \boldsymbol{\mu}_{y_i}^\top f_{\theta, \mathbf{v}}(x_i)\right)}
    {\sum_{c \in \{0,1\}} \exp\left(\kappa \boldsymbol{\mu}_c^\top f_{\theta, \mathbf{v}}(x_i)\right)}.
\end{equation*}
Thus, the shared-concentration vMF model induces a softmax posterior with scale $\kappa$. In Phase I, we retain this functional form while replacing the distributional concentration $\kappa$ with a tunable inverse temperature $\beta$, which yields the optimization objective in~\eqref{eq:postvMF}.

\section{Alternative Threshold Calibration via the Youden Index}
\label{app:threshold_calibration}

In standard benchmarking settings, where a balanced trade-off between the true positive rate (TPR) and false positive rate (FPR) is desired, one may select the threshold by maximizing the Youden index~\citep{youden1950index} on a calibration set. Let the calibration sample be given by $\mathcal S_{\mathrm{cal}} = \{(x_i, y_i)\}_{i=1}^{n_2}$, where $y_i \in \{0,1\}$ indicates whether $x_i$ is LLM-generated ($y_i=1$) or human-written ($y_i=0$). For a threshold $\tau$, the empirical TPR and FPR are defined as
$$
\widehat{\operatorname{TPR}}(\tau):=\frac{1}{n_+}\sum_{i: y_i=1}\mathbbm{1}\bigl(\widehat{\mathcal S}(x_i)\ge \tau\bigr),
\qquad
\widehat{\operatorname{FPR}}(\tau):=\frac{1}{n_-}\sum_{i: y_i=0}\mathbbm{1}\bigl(\widehat{\mathcal S}(x_i)\ge \tau\bigr),
$$
where $n_+ = \sum_{i=1}^{n_2}\mathbbm{1}(y_i=1)$ and $n_- = \sum_{i=1}^{n_2}\mathbbm{1}(y_i=0)$. The Youden-based threshold is then given by
$$
\widehat{\tau}_{\mathrm{Y}}
:=
\argmax_{\tau}
\left(
\widehat{\operatorname{TPR}}(\tau)
+ 1
- \widehat{\operatorname{FPR}}(\tau)
\right).
$$
Equivalently, this criterion maximizes $\widehat{\operatorname{TPR}}(\tau)-\widehat{\operatorname{FPR}}(\tau)$, thereby favoring thresholds that achieve high detection power while maintaining a low FPR.

\paragraph{Empirical comparison of Type-I error control and power.}
To further understand the practical behavior of different calibration strategies, we compare the threshold defined in \eqref{type_1_error}, denoted by $\widehat{\tau}_\alpha$, with the Youden-based threshold $\widehat{\tau}_{\mathrm{Y}}$ in terms of both their realized Type-I error and empirical power.

We use \textit{meta-llama/Llama-3.1-8B} as the observer model. Both the training set and the calibration set are drawn from a mixed dataset comprising outputs from four LLMs together with the corresponding human-written texts, while evaluation is conducted on four test sets generated by \textit{ChatGPT-3.5-Turbo}, \textit{Llama-2-70B}, \textit{Google-PaLM}, and \textit{Claude-Instant}, respectively. For each calibration strategy, the threshold is computed using the same mixed calibration set, and we report the realized Type-I error and empirical detection power. Since $\widehat{\tau}_{\mathrm{Y}}$ does not depend on a target level $\alpha$, it yields a single operating point once the calibration set is fixed, whereas $\widehat{\tau}_\alpha$ explicitly adapts to the desired Type-I error level. We focus on the representative operating regime $\alpha = 1\%$.

\begin{table}[h]
\centering
\small
\renewcommand{\arraystretch}{1.15}
\setlength{\tabcolsep}{4pt}
\caption{Realized type-I error (\%) and empirical detection power (\%) under different calibration strategies at target level $\alpha=1\%$.}
\label{tab:type1_power_control}
\begin{tabular}{lcccccccc}
\toprule
\multirow{2}{*}{\textbf{Method}} 
& \multicolumn{2}{c}{ChatGPT-3.5-Turbo}
& \multicolumn{2}{c}{Llama-2-70B}
& \multicolumn{2}{c}{Google-PaLM}
& \multicolumn{2}{c}{Claude-Instant} \\
\cmidrule(lr){2-3} \cmidrule(lr){4-5} \cmidrule(lr){6-7} \cmidrule(lr){8-9}
& Type-I Error & Power
& Type-I Error & Power
& Type-I Error & Power
& Type-I Error & Power \\
\midrule
$\widehat{\tau}_{\mathrm{Y}}$ & 0.20 & 99.70 & 0.00 & 99.50 & 0.20 & 85.70 & 0.20 & 99.80 \\
$\widehat{\tau}_\alpha$ & 0.90 & 100.00 & 0.80 & 99.90 & 0.80 & 92.10 & 1.30 & 99.90 \\
\bottomrule
\end{tabular}
\end{table}
As shown in Table~\ref{tab:type1_power_control}, the threshold $\widehat{\tau}_\alpha$ achieves realized Type-I errors close to the target level of $1\%$ across all test sets, with FPRs ranging from $0.80\%$ to $1.30\%$, which is consistent with the finite-sample guarantee in Theorem~\ref{thm_1}. At the same time, it maintains strong empirical power, reaching at least $99.90\%$ on three generators and substantially improving performance on \textit{Google-PaLM}. By contrast, the Youden-based threshold $\widehat{\tau}_{\mathrm{Y}}$ yields a single fixed operating point and does not explicitly target a prescribed Type-I error level. In our experiments, it behaves conservatively, with realized FPRs between $0.00\%$ and $0.20\%$, all well below the target level of $1\%$. Although this still gives strong detection performance on some generators, it also reduces power by imposing a stricter threshold than necessary; this is most evident on \textit{Google-PaLM}, where the TPR decreases from $92.10\%$ under $\widehat{\tau}_\alpha$ to $85.70\%$ under $\widehat{\tau}_{\mathrm{Y}}$. Overall, these results show that while $\widehat{\tau}_{\mathrm{Y}}$ is suitable for a fixed balanced operating point, $\widehat{\tau}_\alpha$ is more flexible with respect to the desired false positive budget and can yield better power by avoiding unnecessarily conservative thresholding.

\section{Experiment Details}
\label{sec:experiment_details}

This section outlines the experimental setup and evaluation criteria.

\paragraph{Evaluation Metrics.}
To comprehensively evaluate the discriminative capability of each detector under different security requirements, we adopt the Area Under the Receiver Operating Characteristic curve (AUROC) and the True Positive Rate (TPR) at fixed False Positive Rate (FPR) thresholds, namely TPR@FPR=1\% and TPR@FPR=0.01\%. Since the empirical ROC curve is defined over discrete test samples, we use linear interpolation to report TPR at the specified FPR levels. In particular, with 1,000 human-written test samples, the smallest nonzero empirical FPR is $0.1\%$, so TPR@0.01\% is an interpolated low-FPR summary rather than a directly observed operating point. Appendix~\ref{app:low_fpr_validation} provides a large-sample diagnostic of this interpolation.

\paragraph{Setup.}
To ensure comparability, we standardize backbone models across all detectors. We use \textit{meta-llama/Llama-3.1-8B} as the base model for scoring, rewriting, and representation extraction. For training-free methods, the same open-source LLM is used as a surrogate to compute statistics (e.g., Likelihood, LRR, FDGPT). In Binoculars, \textit{meta-llama/Llama-3.1-8B} serves as the observer and \textit{meta-llama/Llama-3.1-8B-Instruct} as the performer. For training-based approaches (e.g., RoBERTa), we report both models fine-tuned on our data and publicly available checkpoints\footnote{\url{https://huggingface.co/openai-community/roberta-large-openai-detector}}. For rewrite-based methods, rewrites are generated by \textit{meta-llama/Llama-3.1-8B-Instruct} with $4$ samples, following~\citep{zhou2026learn}.

For the experiments in Sections~\ref{sec_rob_ana} and~\ref{understanding_s2d} (excluding the ID/OOD experiments), the detectors are trained on a mixed dataset composed of outputs from four LLMs and evaluated under diverse settings. We consistently use \textit{meta-llama/Llama-3.1-8B} as the observer model, and \textit{meta-llama/Llama-3.1-8B-Instruct} for methods requiring rewriting. The rewriting process incurs a substantial computational cost, amounting to approximately 300 GPU hours on NVIDIA A100 80GB GPUs. For experiments comparing with watermarking methods, we employ \textit{Qwen/Qwen2.5-7B} as the base model to generate watermarked text using three different watermarking schemes. 

For our method, unless otherwise specified, we adopt a fixed configuration across all experiments. The steering vector is injected at layer $11$ through the residual connection. Representations are extracted from the last $8$ transformer layers, where we average the final $25\%$ tokens in each layer to obtain the representation. Training is conducted using AdamW~\citep{loshchilov2017decoupled} with a learning rate of $1\times10^{-3}$, batch size $8$, and $10$ epochs. The EMA decay for centroid updates is set to $0.9$, and the optimization inverse temperature is set to $\beta=2.5$. All hyperparameters are kept fixed across datasets and models without extensive tuning.

\section{Empirical Motivation for Representation Modeling}\label{app_repre_norm}

Our modeling choices for $f_{\theta}$ are supported by the following two empirical properties of hidden representations observed across models and domains:
\begin{enumerate}
    \item \textbf{Norm concentration.} By pooling hidden states from the final $20\%$ of tokens across the last $N=8$ layers, we find that the resulting $L_2$ norms are tightly concentrated for \textit{EleutherAI/GPT-J-6B}, \textit{Qwen/Qwen2.5-7B}, and \textit{meta-llama/Llama-3.1-8B} (Figure~\ref{fig_norm_distribution}). This supports a hyperspherical approximation, where semantic information is primarily encoded in direction rather than magnitude.
    
    \item \textbf{Directional unimodality.} On the unit sphere, the normalized representations exhibit a unimodal directional structure within each class (human or LLM), as evidenced by the distribution of $\widehat{\mu}^{\top} f_\theta(x)$ relative to the estimated class mean directions $\widehat{\mu}$ (Figure~\ref{fig_vmf_domain}). Moreover, the estimated concentration parameters $\widehat{\kappa}$ are stable across domains and remain similar between the human and LLM classes.
\end{enumerate}

These observations support the use of a von Mises--Fisher (vMF) model as a simple and effective approximation for the directional geometry, with a shared concentration parameter $\kappa$.

\begin{figure}[htbp]
    \centering
    \includegraphics[width=0.95\textwidth]{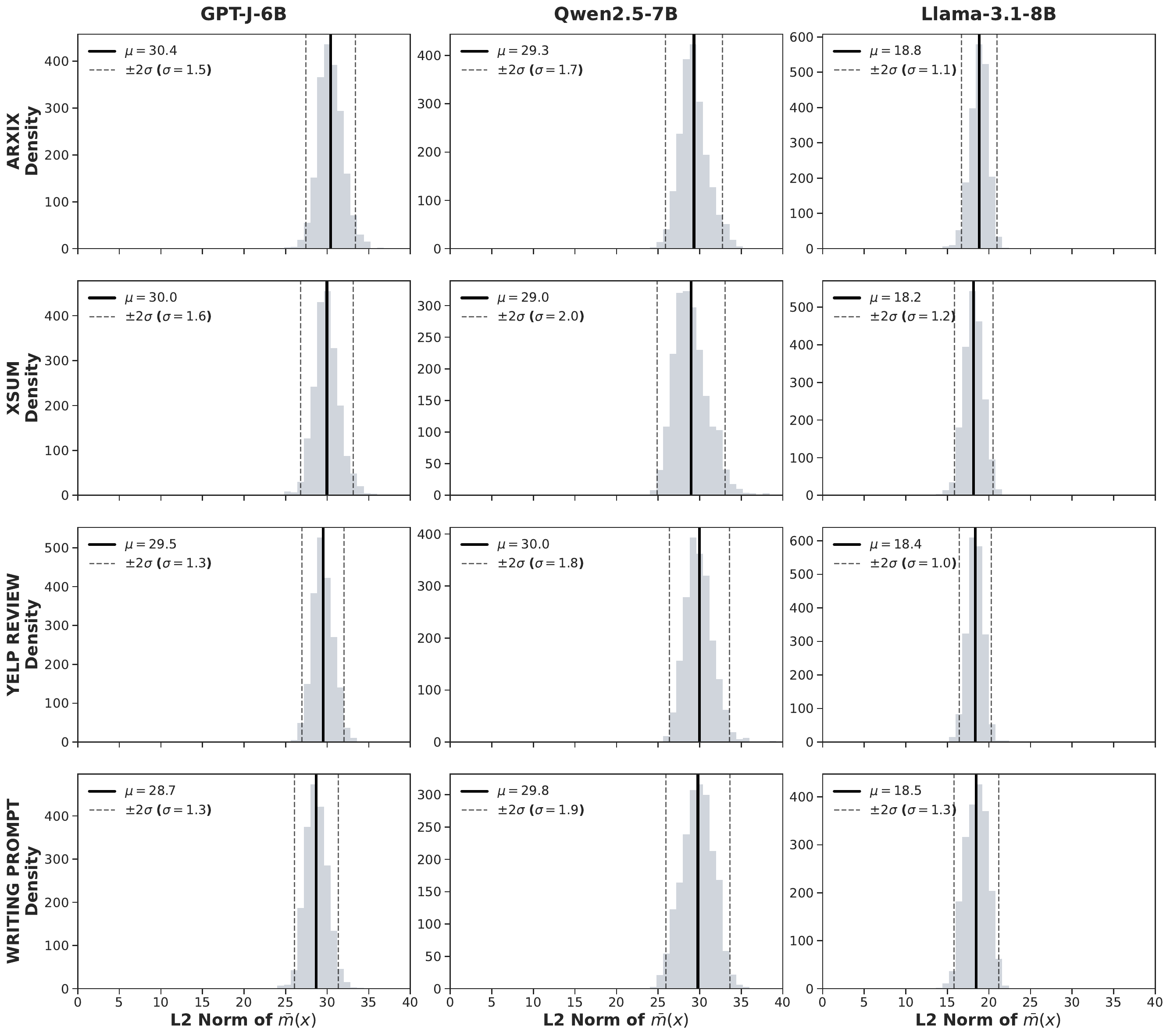}
\caption{
Empirical distributions of $L_2$ norms of representations obtained from the last 8 layers and the final 20\% of tokens, across different models and domains. Columns correspond to \textit{EleutherAI/GPT-J-6B}, \textit{Qwen/Qwen2.5-7B}, and \textit{meta-llama/Llama-3.1-8B}, while rows correspond to the Arxiv, XSum, Yelp, and Writing datasets. The solid and dashed red lines denote the mean and the $\pm 2\sigma$ intervals, respectively, with exact values reported in the upper-left legends.
}
    \label{fig_norm_distribution}
\end{figure}

\begin{figure}[htbp]
    \centering
    \includegraphics[width=0.95\linewidth]{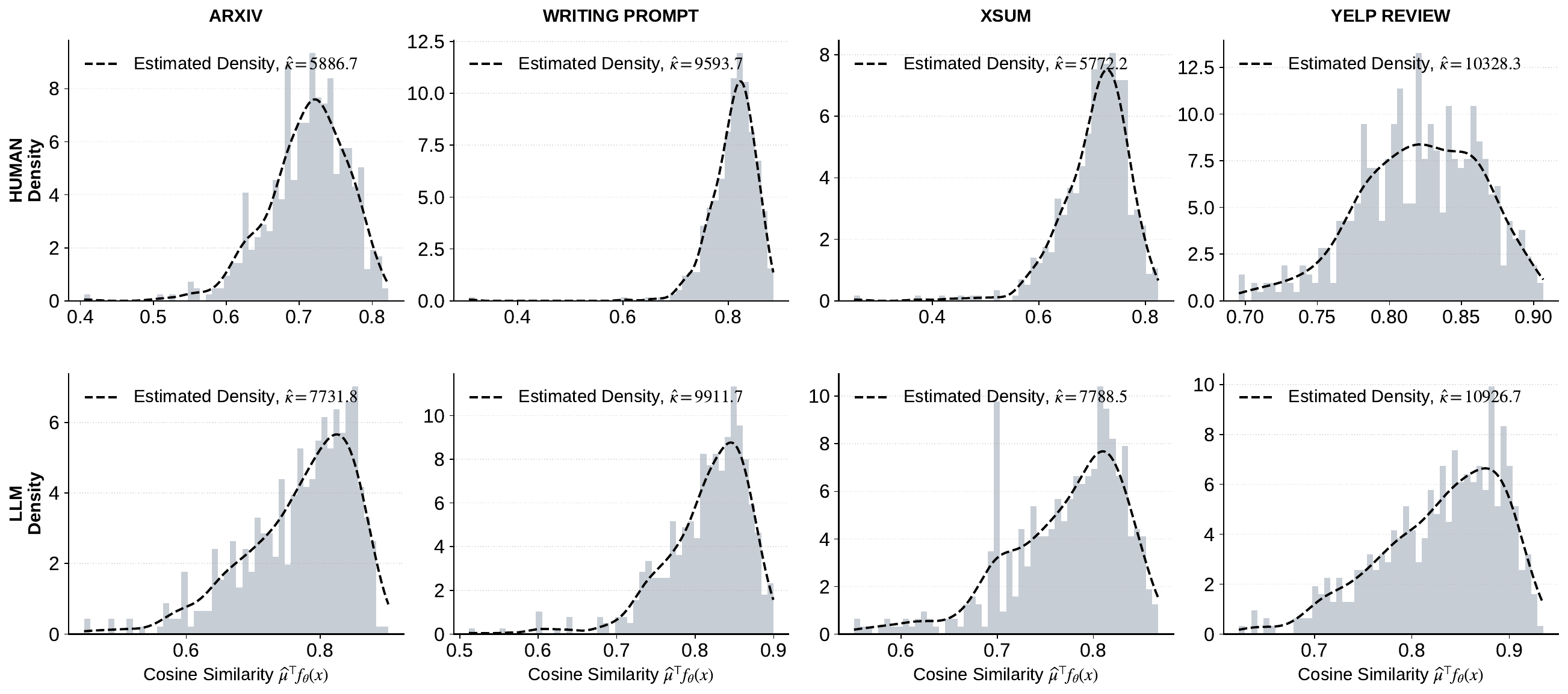}
\caption{
Empirical distributions of projected representations $\widehat{\mu}^\top f_\theta(x)$ across domains. Histograms show empirical frequencies, and dashed curves represent the corresponding estimated densities. The estimated concentration parameter $\kappa$ is reported for each case.
}
    \label{fig_vmf_domain}
\end{figure}

\section{Additional Experiment results}
\subsection{Score Distributions}\label{appendix_score_dist}
In this section, we visualize score distributions from four LLMs under different settings. We use \textit{meta-llama/Llama-3.1-8B} as the observer model. The training data is generated by \textit{GPT-3.5-Turbo}. The evaluated datasets are generated by \textit{GPT-3.5-Turbo}, \textit{Claude-Instant}, \textit{Google-PaLM}, and \textit{Llama-2-70B}.

In Figure~\ref{fig_score_distribution_grid}, the leftmost column shows scores from RepreGuard~\citep{chen2025repreguard}, the middle column corresponds to raw representations without steering, and the rightmost column shows scores produced by our method.
Although the first two approaches exploit representation-level separability, a non-negligible overlap between the two classes remains, which limits detection performance. This issue is particularly pronounced under OOD settings: RepreGuard exhibits $15.24\%$ and $17.25\%$ overlap on datasets generated by \textit{Claude-Instant} and \textit{Google-PaLM}, respectively, while \texttt{S2D} without steering yields $14.10\%$ and $18.27\%$. 
In contrast, our method achieves substantially clearer separation between human-written and LLM-generated text, even under OOD settings, leading to improved discriminability.

\begin{figure*}[t]
\centering
\resizebox{\textwidth}{!}{
\begin{tabular}{c c c c}
    & \textbf{RepreGuard} & \textbf{\texttt{S2D} w/o steering} & \textbf{\texttt{S2D} (ours)} \\

\raisebox{.20\height}{\rotatebox{90}{GPT-3.5-Turbo}} &
\includegraphics[width=0.3\textwidth]{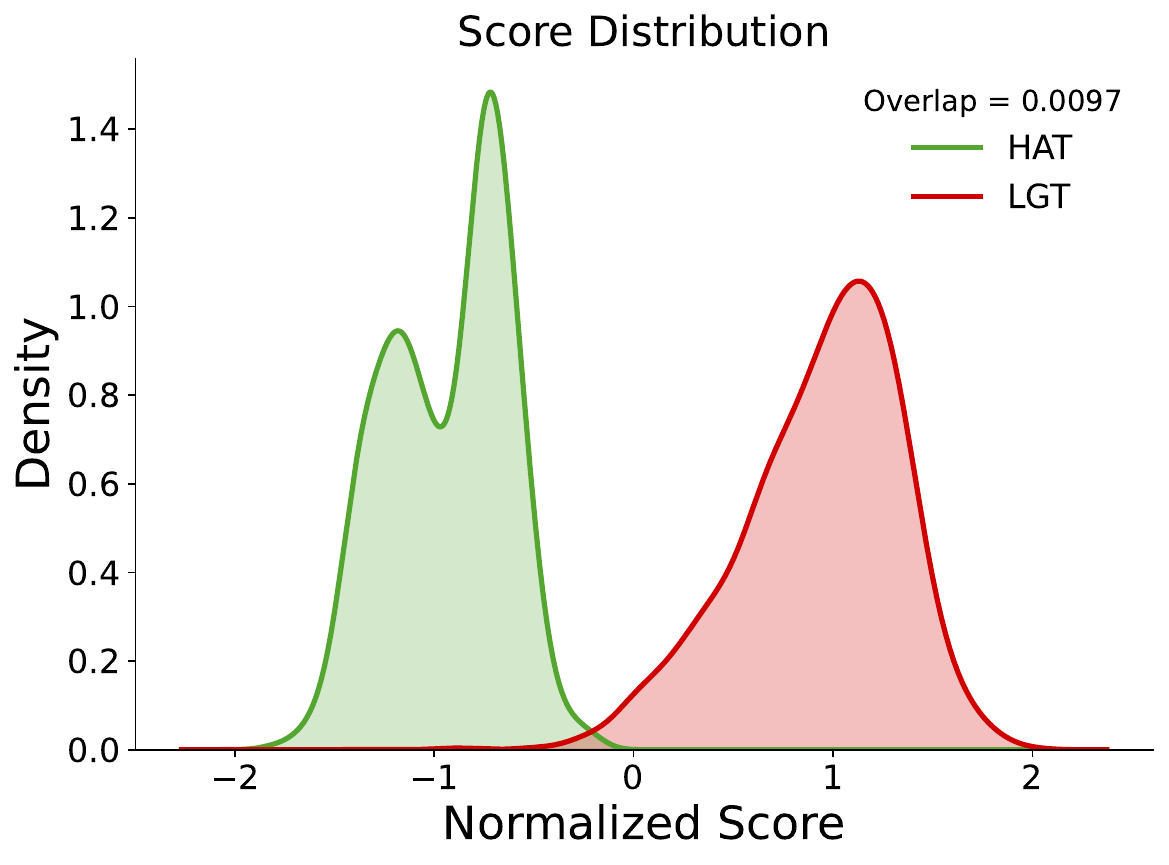} &
\includegraphics[width=0.3\textwidth]{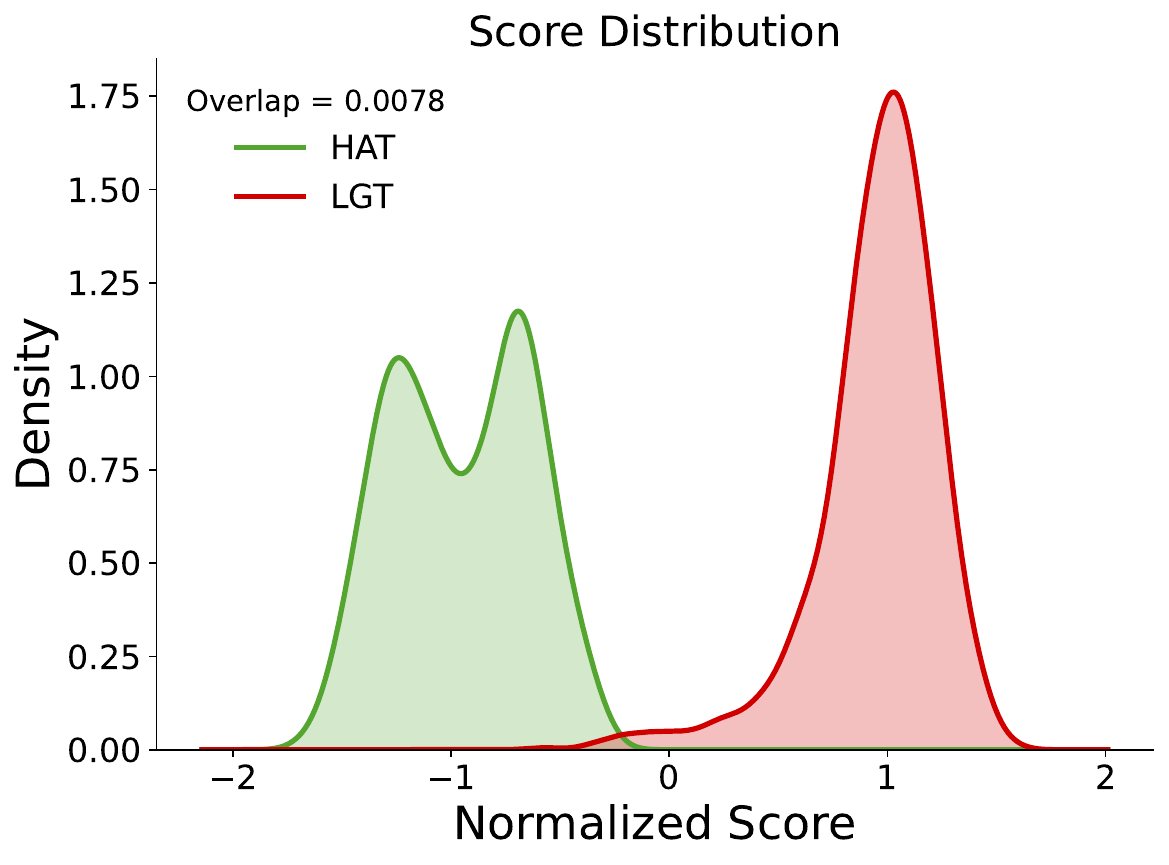} &
\includegraphics[width=0.3\textwidth]{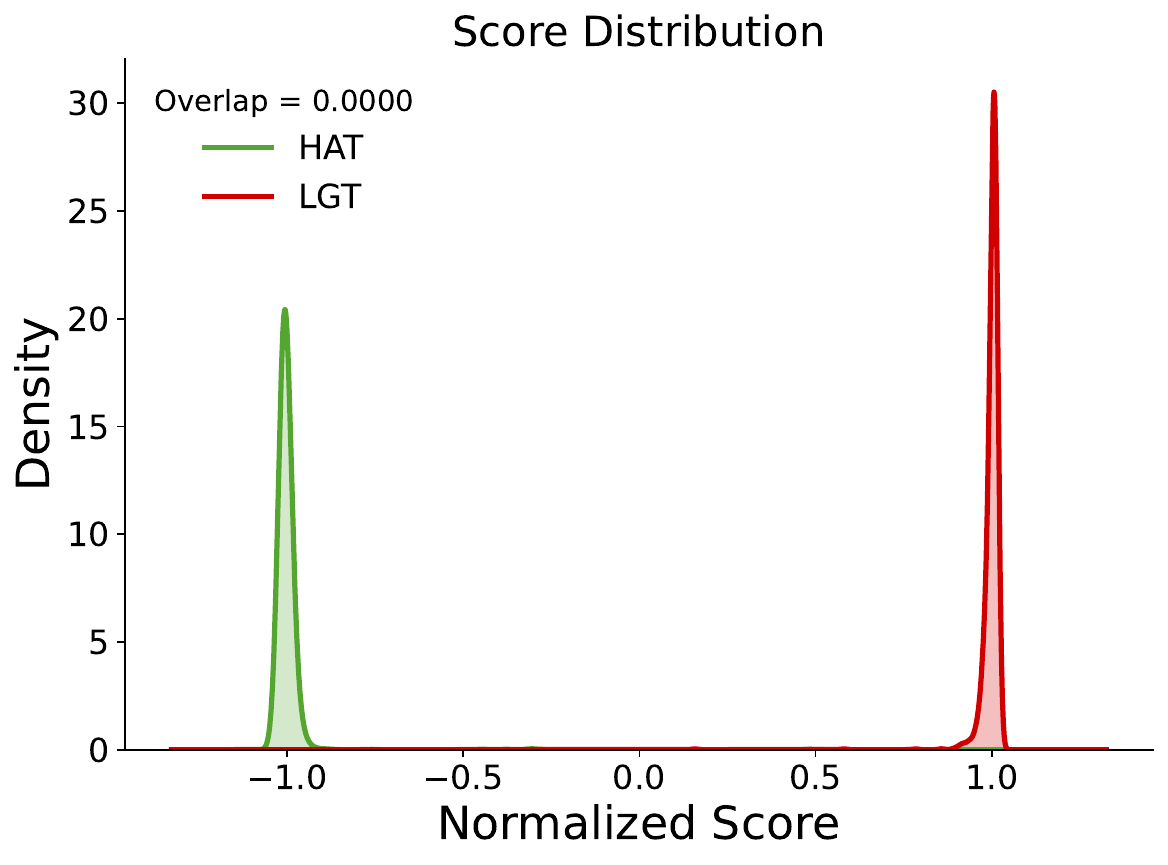} \\

\raisebox{.42\height}{\rotatebox{90}{Llama-2-70B}} &
\includegraphics[width=0.3\textwidth]{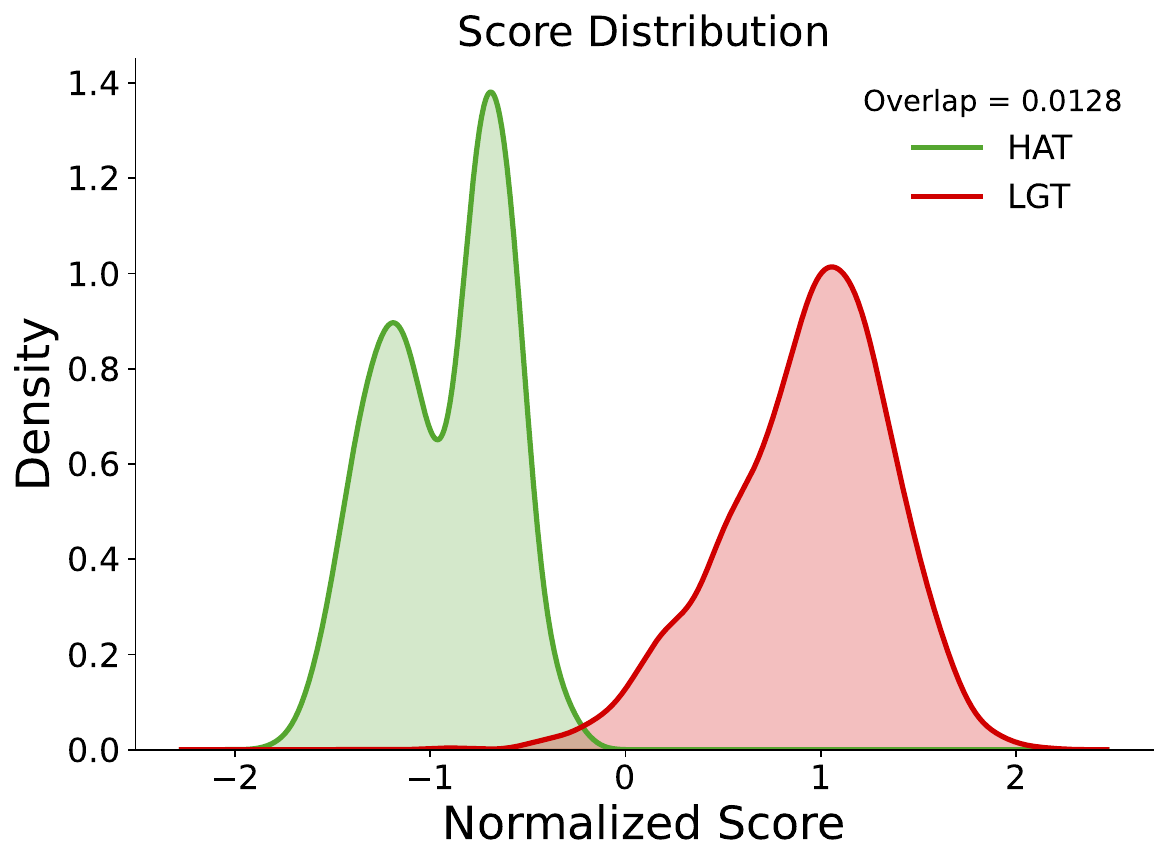} &
\includegraphics[width=0.3\textwidth]{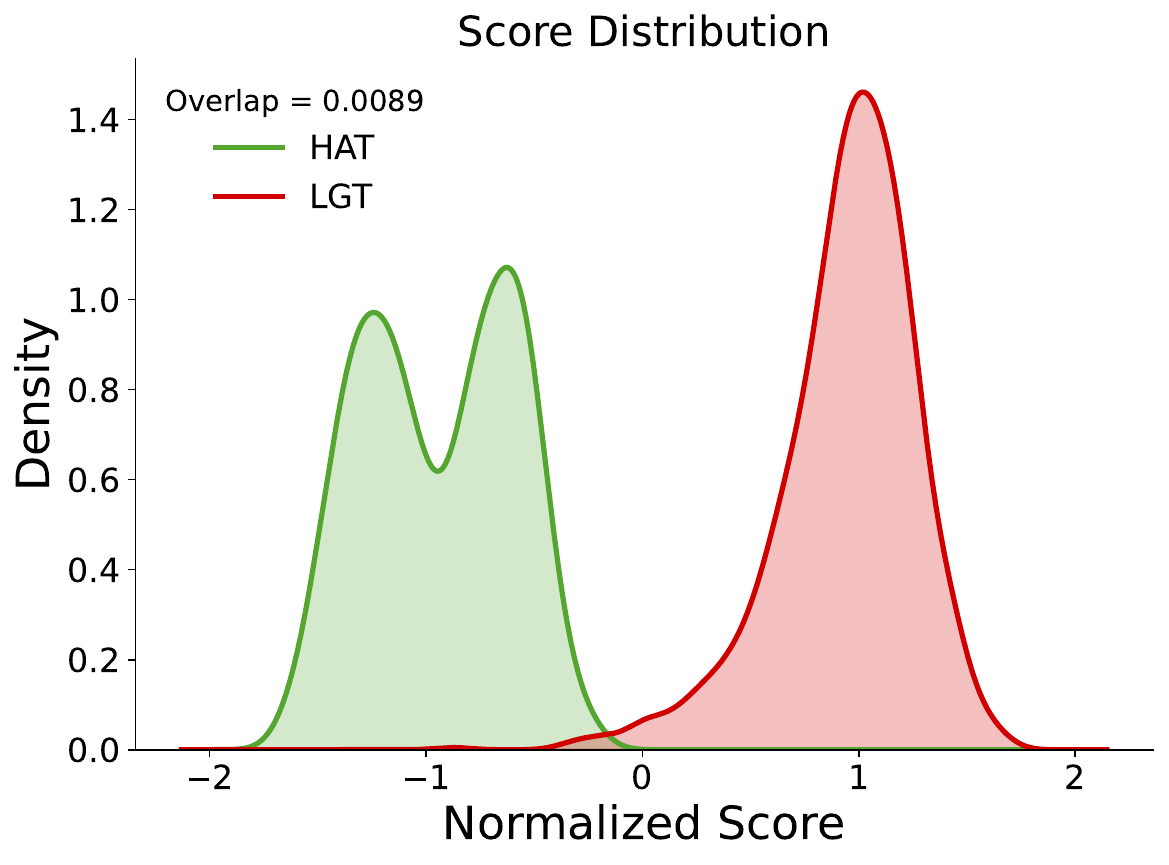} &
\includegraphics[width=0.3\textwidth]{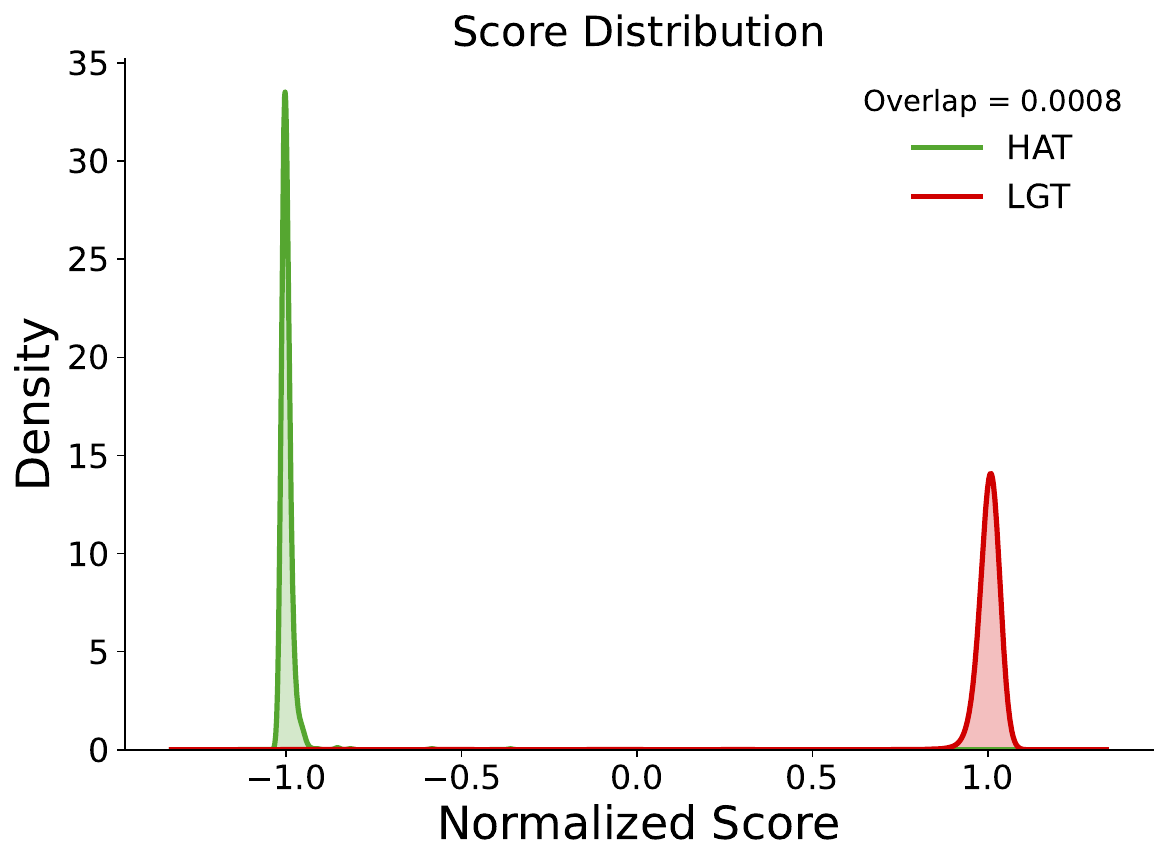} \\

\raisebox{.33\height}{\rotatebox{90}{Claude-Instant}} &
\includegraphics[width=0.3\textwidth]{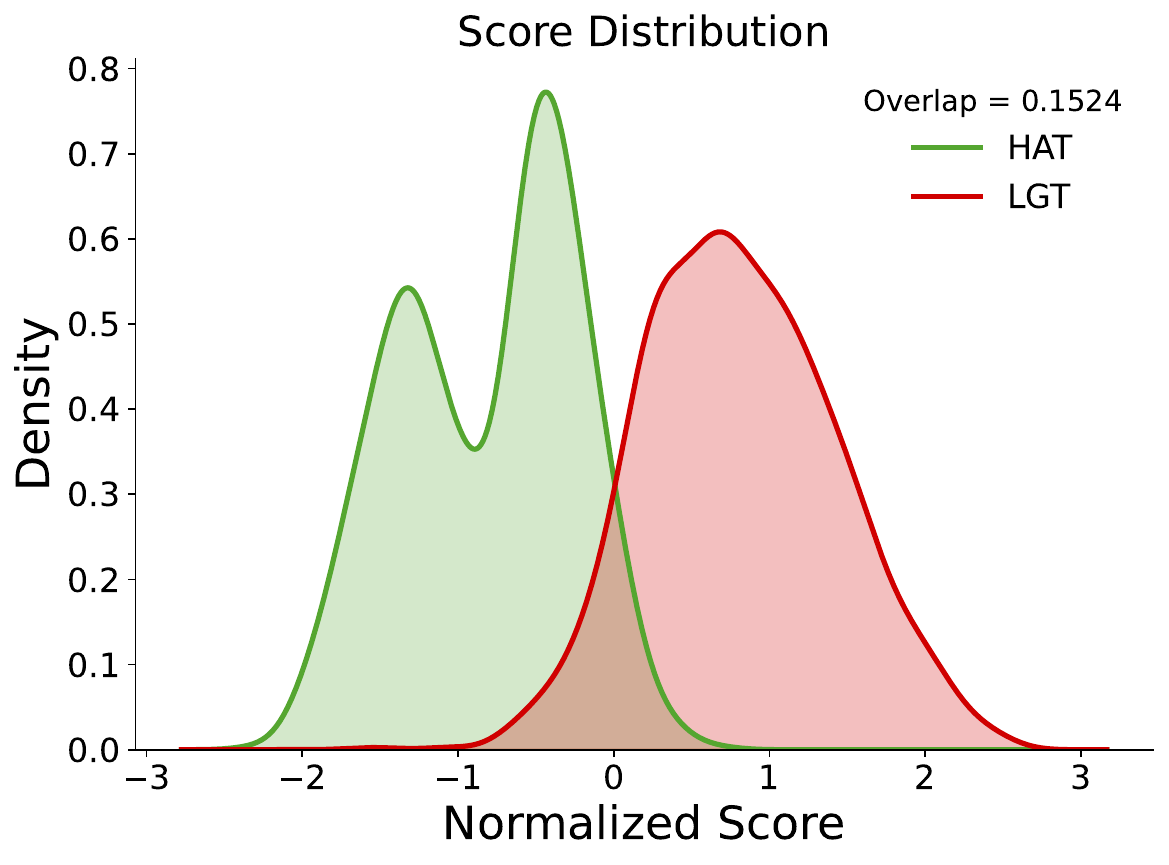} &
\includegraphics[width=0.3\textwidth]{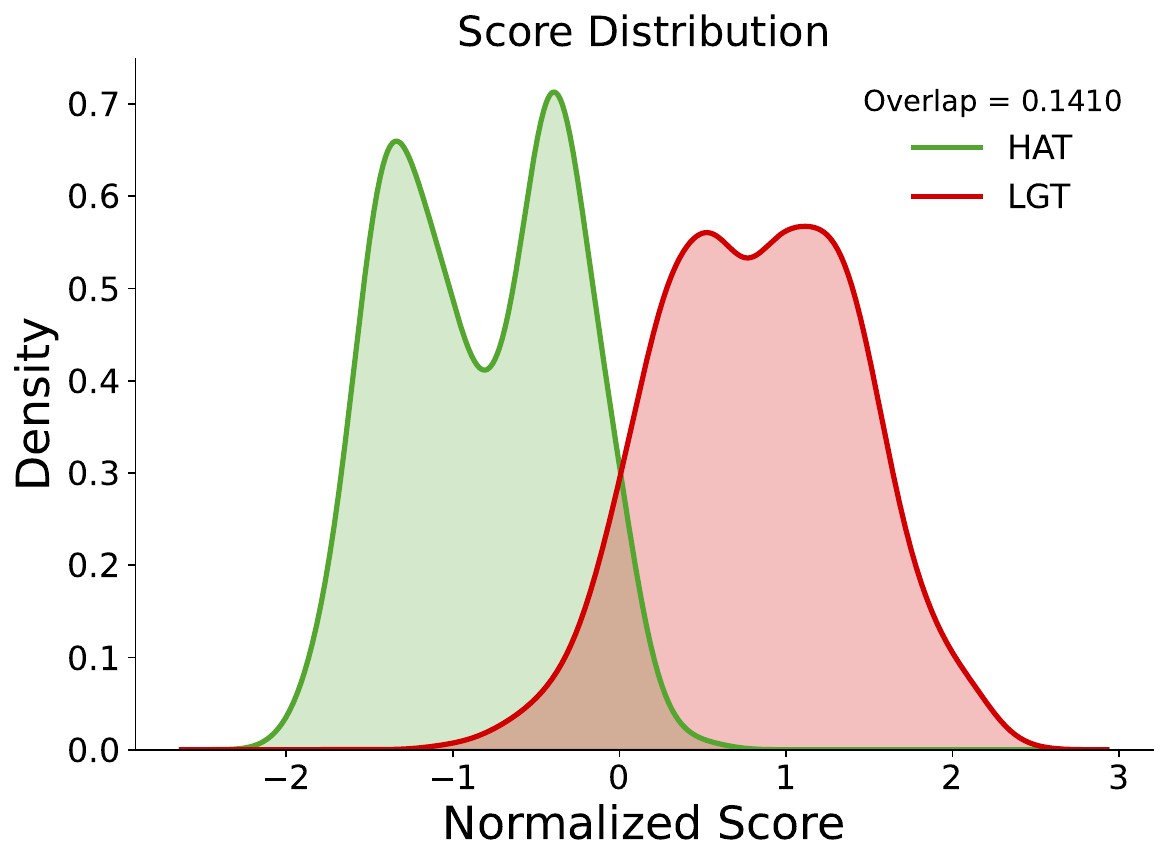} &
\includegraphics[width=0.3\textwidth]{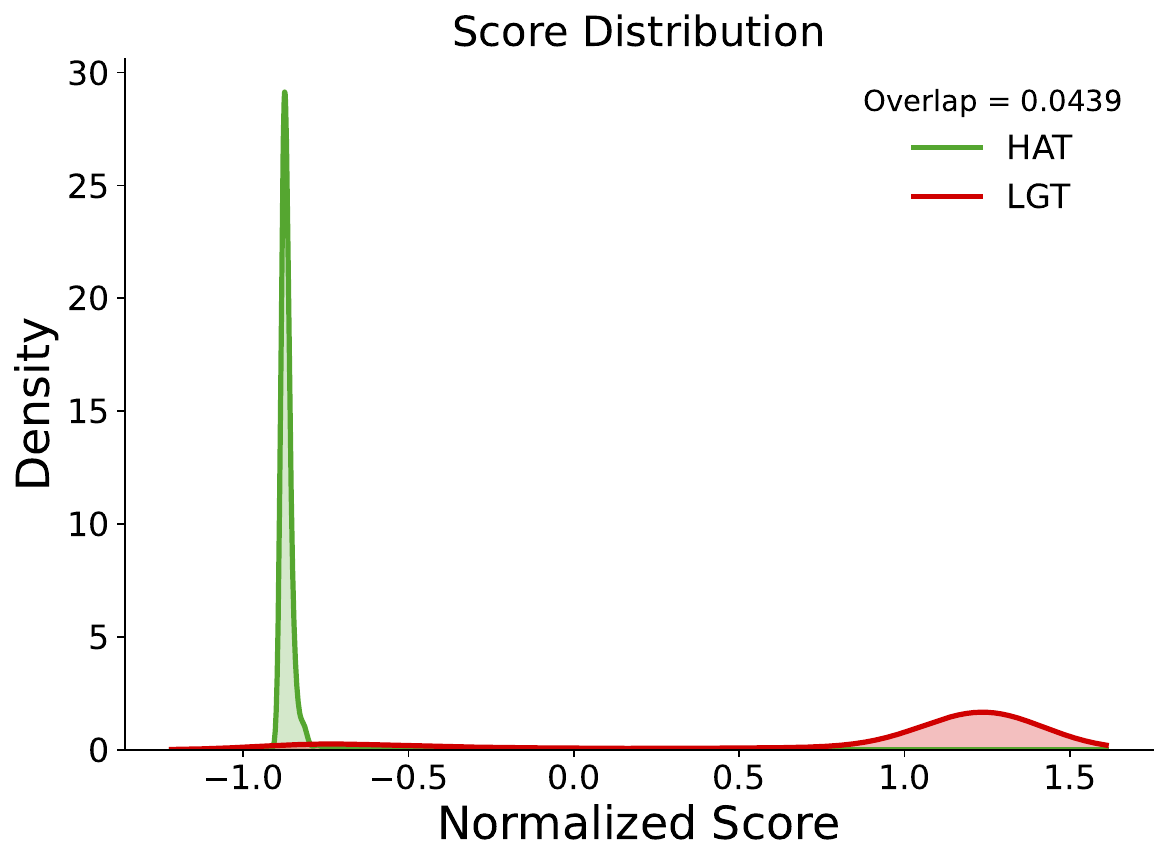} \\

\raisebox{.4\height}{\rotatebox{90}{Google-PaLM}} &
\includegraphics[width=0.3\textwidth]{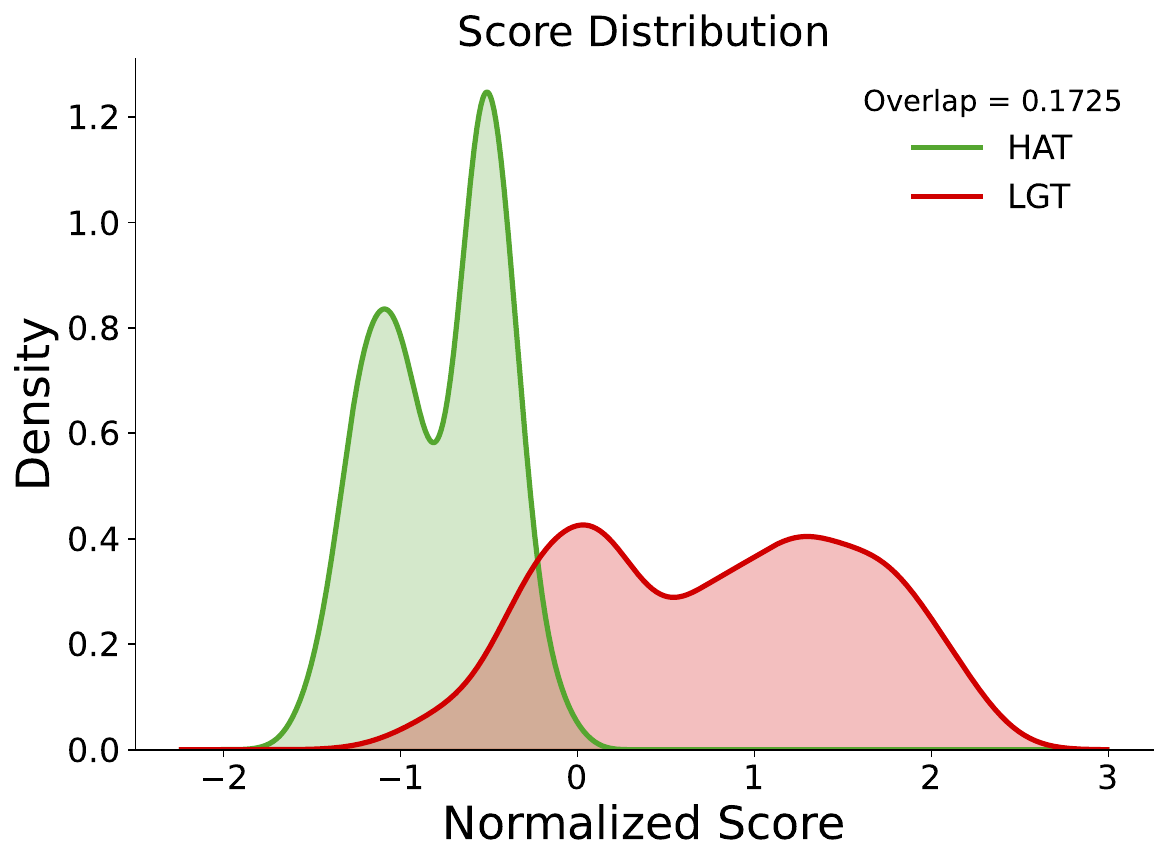} &
\includegraphics[width=0.3\textwidth]{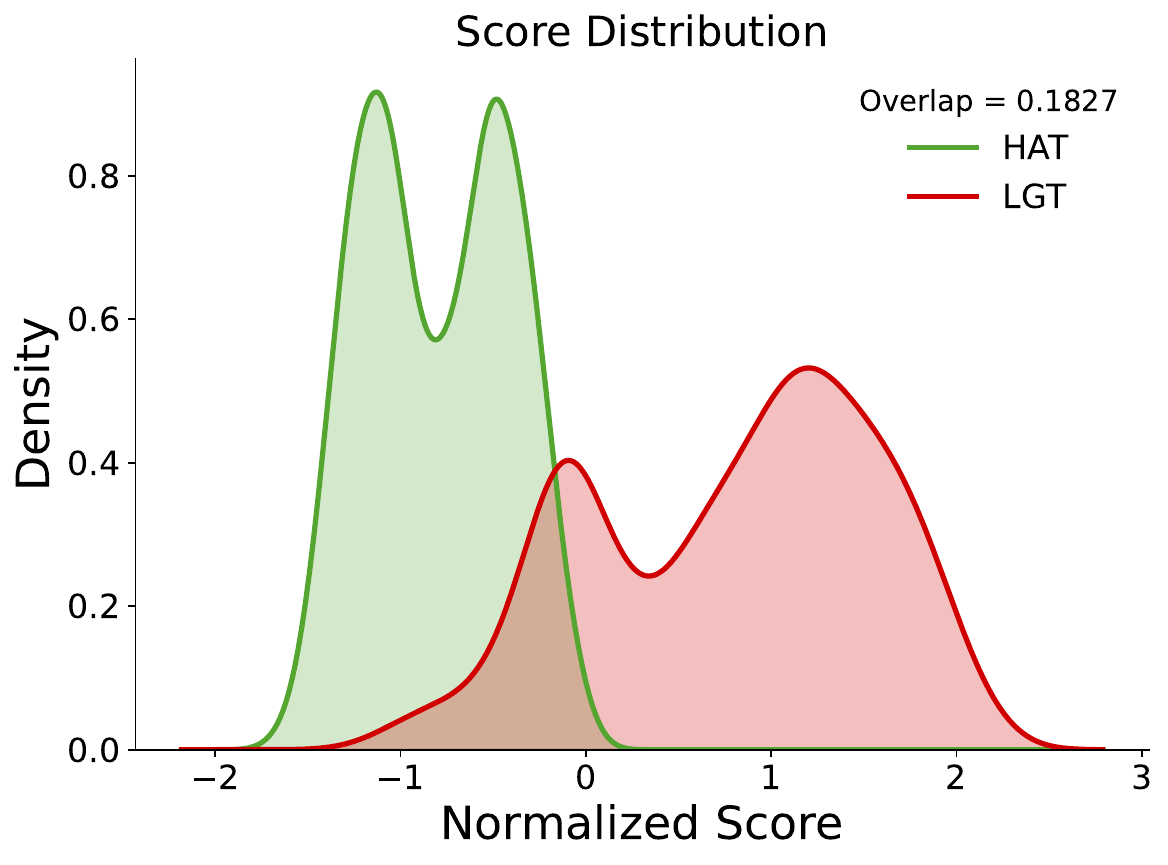} &
\includegraphics[width=0.3\textwidth]{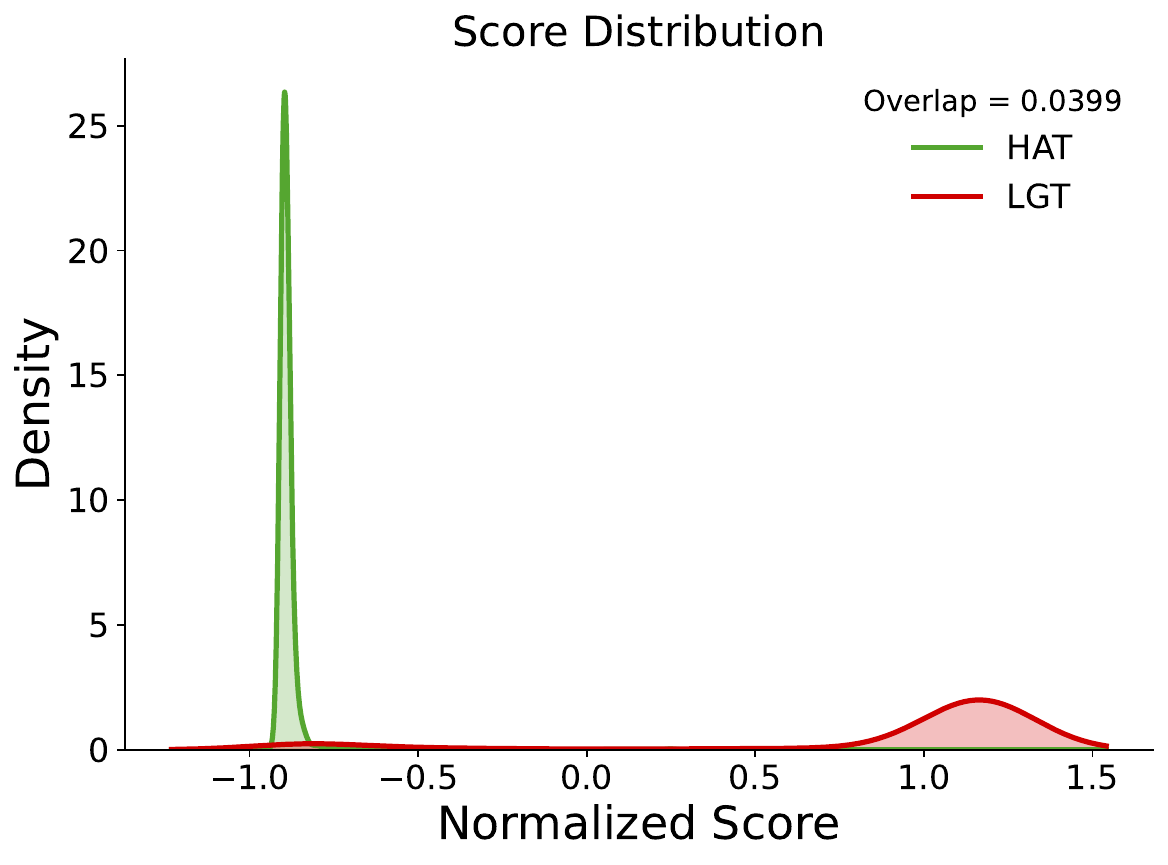} \\
\end{tabular}}
\caption{
Score distributions of different detection methods involving hidden representations across datasets. Columns denote different methods and rows denote datasets generated by different LLMs. HAT and LGT denote human-authored text and LLM-generated text, respectively.
}
\label{fig_score_distribution_grid} 
\end{figure*}

\subsection{Large-Sample Validation of Low-FPR Interpolation}
\label{app:low_fpr_validation}
\begin{wraptable}{r}{0.48\textwidth}
    \vspace{0pt}
    \centering
    \caption{
    Large-sample validation of TPR@0.01\%.
    All values are percentages.
    }
    \label{tab:low_fpr_validation}
    \small
    \setlength{\tabcolsep}{7pt}
    \renewcommand{\arraystretch}{1.05}
    \begin{tabular}{lccc}
        \toprule
        Method & Interp. & Direct & $|\mathrm{Gap}|$ \\
        \midrule
        \texttt{S2D} & 38.61 & 38.52 & 0.09 \\
        RepreGuard    & 13.99 & 13.99 & 0.00 \\
        ImBD          & 11.12 & 11.10 & 0.02 \\
        \bottomrule
    \end{tabular}
    \vspace{-10pt}
\end{wraptable}
Each DetectRL test set contains 1,000 human-written samples, so the smallest nonzero empirical FPR is $0.1\%$. Therefore, TPR@0.01\% in the main experiments is obtained by linear interpolation of the empirical ROC curve. We conduct an additional large-sample experiment to examine the effect of this interpolation. Specifically, we collect 21686 human-written texts and 20,000 LLM-generated texts from C-ReD, EvoBench, and OpAI-Bench. The resulting empirical FPR resolution is approximately $1/21686 \approx 0.0046\%$, allowing us to evaluate performance directly near FPR $=0.01\%$. We compare the interpolated and directly evaluated TPR@0.01\% for \texttt{S2D}, RepreGuard, and ImBD. As shown in Table~\ref{tab:low_fpr_validation}, the interpolated and direct estimates differ by at most $0.09$ percentage points across the three methods. This suggests that, in this larger-sample evaluation, linear interpolation has little effect on the reported TPR@0.01\% values.

\subsection{Shots of Training Dataset}\label{app_shots_train}
We evaluate the impact of training set size on detection performance. The training sets are randomly sampled from a mixed-LLM dataset with sizes $\{16, 32, 64, 128, 196, 256, 384, 512\}$ pairs. We use \textit{Llama-3.1-8B} as the base model and \textit{Llama-3.1-8B-Instruct} as the rewriter, and evaluate performance on four test sets (\textit{ChatGPT-3.5-Turbo}, \textit{Claude-Instant}, \textit{Google-PaLM}, and \textit{Llama-2-70B}). As shown in Figure~\ref{fig:train_shots}, \texttt{S2D} achieves strong data efficiency and consistent improvement over baselines, especially in low-resource settings. We observe a slight performance drop at the earliest stage as training size increases, since a very small number of samples can perturb the estimated score distribution. Nevertheless, even with as few as 16 pairs, \texttt{S2D} can still exploit the inherent score gap between human-written and LLM-generated texts. As more training data becomes available, the estimated distribution stabilizes, resulting in clearer separation and improved performance.
\begin{figure*}[htbp]
    \centering
    \includegraphics[width=0.98\textwidth]{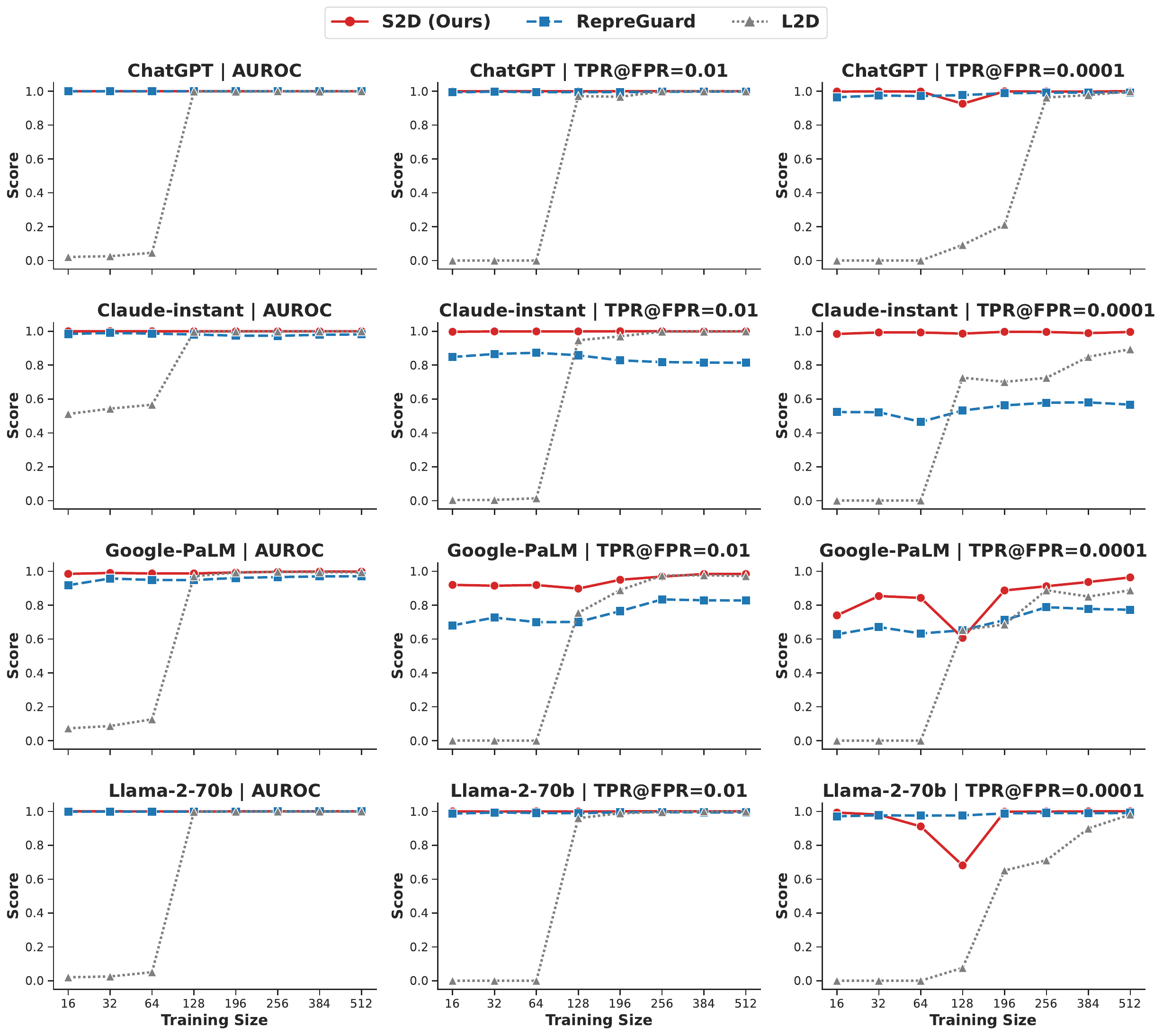}
    \caption{Detection performance across different training set sizes.}
    \label{fig:train_shots}
\end{figure*}

\subsection{Long-Context Steering Stability}
\label{app:long_context}
\begin{wraptable}{r}{0.45\textwidth}
    \vspace{-12pt}
    \centering
    \caption{
    Long-context stability of \texttt{S2D}.
    }
    \label{tab:long_context}
    \small
    \setlength{\tabcolsep}{4.5pt}
    \renewcommand{\arraystretch}{1.05}
    \begin{tabular}{rccc}
        \toprule
        Tokens & AUROC & TPR@1\% & Overlap \\
        \midrule
        256   & 0.9107 & 0.610 & 0.105 \\
        512   & 0.9667 & 0.790 & 0.096 \\
        1024  & 0.9940 & 0.960 & 0.055 \\
        2048  & 0.9998 & 0.990 & 0.035 \\
        4096  & 0.9999 & 0.990 & 0.036 \\
        \bottomrule
    \end{tabular}
    \vspace{-16pt}
\end{wraptable}
Since the same steering vector is injected at every valid token, we examine whether long sequences cause representation saturation or collapse. We train \texttt{S2D} on 512 mixed DetectRL human--LLM pairs using \textit{Meta-Llama-3-8B} and evaluate a fixed steering checkpoint on a separate long-context diagnostic set constructed from OpAI-Bench, C-ReD, and EvoBench, with 400 human--LLM pairs at each sequence length. As shown in Table~\ref{tab:long_context}, performance improves with context length and plateaus beyond 2048 tokens, while score overlap decreases from $0.105$ to about $0.035$. We therefore observe no evidence of length-induced representation collapse up to 4,096 tokens. This result is an empirical diagnostic rather than a guarantee for arbitrarily long contexts.

\subsection{Different Observers}
Table~\ref{tab:observer_models_strict_sh} shows that strong observer models (e.g., \textit{Llama-3.1-8B}, \textit{Qwen2.5-7B}, and \textit{Falcon-7B}) achieve consistently near-saturated performance across all generators, with high AUROC and stable TPR even at very low false positive rates. In contrast, weaker models (e.g., \textit{Mistral-7B-v0.3}, \textit{GPT-Neo-2.7B}, and \textit{Gemma-2-9B}) exhibit both lower accuracy and larger variance across generators, suggesting sensitivity to distributional shifts and less discriminative features. Notably, instruction-tuned variants do not consistently improve performance, suggesting that detection effectiveness in \texttt{S2D} is primarily determined by representation quality, rather than model scale or alignment alone.

\begin{table*}[thp]
\centering
\small
\setlength{\tabcolsep}{3.2pt}
\renewcommand{\arraystretch}{1.15}
\resizebox{\textwidth}{!}{
\begin{tabular}{l| ccc ccc ccc ccc | ccc}
\toprule
\multirow{2}{*}{\textbf{Observer Model}$\downarrow$}
& \multicolumn{3}{c}{ChatGPT}
& \multicolumn{3}{c}{Llama-2-70B}
& \multicolumn{3}{c}{Google-PaLM}
& \multicolumn{3}{c|}{Claude-Instant}
& \multicolumn{3}{c}{Avg.} \\
\cmidrule(lr){2-4}\cmidrule(lr){5-7}\cmidrule(lr){8-10}\cmidrule(lr){11-13}\cmidrule(lr){14-16}
& \textit{AUROC} & \textit{TPR@1\%} & \textit{TPR@0.01\%}
& \textit{AUROC} & \textit{TPR@1\%} & \textit{TPR@0.01\%}
& \textit{AUROC} & \textit{TPR@1\%} & \textit{TPR@0.01\%}
& \textit{AUROC} & \textit{TPR@1\%} & \textit{TPR@0.01\%}
& \textit{AUROC} & \textit{TPR@1\%} & \textit{TPR@0.01\%} \\
\midrule
Llama-3.1-8B  & 99.99 & 99.90 & 99.80 & 98.95 & 98.90 & 98.10 & 99.83 & 98.50 & 97.40 & 99.70 & 98.90 & 96.60 & 99.62 & 99.05 & 97.98 \\

Mistral-7B-v0.3   & 68.85 & 29.20 & 24.20 & 67.55 & 33.80 & 26.60 & 67.97 & 35.40 & 28.40 & 69.72 & 23.70 & 13.90 & 68.52 & 30.53 & 23.28 \\


GPT-Neo-2.7B    & 84.44 & 33.70 & 8.30 & 79.97 & 41.50 & 19.40 & 79.63 & 20.40 & 6.70 & 85.46 & 35.00 & 10.90 & 82.38 & 32.65 & 11.33 \\

OPT-2.7B  & 99.95 & 99.80 & 82.60 & 99.93 & 99.30 & 94.10 & 99.66 & 95.80 & 74.50 & 99.88 & 98.30 & 86.80 & 99.86 & 98.30 & 84.50 \\

Qwen2.5-7B  & 99.99 & 99.90 & 99.80 & 99.99 & 99.90 & 98.90 & 99.73 & 98.60 & 96.20 & 99.99 & 99.90 & 99.70 & 99.93 & 99.58 & 98.65 \\

Falcon-7B   & 99.99 & 99.90 & 99.80 & 99.99 & 99.90 & 99.60 & 99.89 & 98.10 & 95.70 & 99.98 & 99.70 & 98.00 & 99.96 & 99.40 & 98.28 \\

Falcon-7B-Instruct    & 99.98 & 99.90 & 99.60 & 99.99 & 99.90 & 95.70 & 99.89 & 97.00 & 93.20 & 99.97 & 99.90 & 95.10 & 99.96 & 99.18 & 95.90 \\

Gemma-2-9B   & 88.01  & 25.70 & 5.90 & 76.25 & 13.90 & 0.10 & 56.10 & 20.90 & 5.20 & 85.79 & 16.50 & 1.00 & 76.54 & 19.25 & 3.05 \\

Gemma-2-9B-Instruct   & 96.25  & 62.70 & 32.50 & 95.66 & 60.10 & 29.70 & 54.10 & 15.60 & 7.30 & 70.34 & 12.50 & 6.10 & 79.09 & 37.73 & 18.90 \\
\bottomrule
\end{tabular}}
\caption{Impact of observer model. Detection performance of \texttt{S2D} across various observer models.}
\label{tab:observer_models_strict_sh}
\vspace{-0.1in}
\end{table*}

\paragraph{Steering Depth Across Model Scales.}
\label{app:steering_depth_scaling}
We examine how the preferred steering depth changes with model scale by sweeping injection layers on \textit{Llama-3.2-1B}, \textit{Meta-Llama-3-8B}, and \textit{Meta-Llama-3-70B} under the same experimental protocol. Since AUROC is nearly saturated, we report TPR@0.01\% and human--LLM score overlap. We define the high-performing interval as the largest contiguous region among the swept layers whose mean TPR@0.01\% is within 10 percentage points of the model-specific maximum and whose overlap is at most 2.5\%.

\begin{table}[t]
    \centering
    \caption{
    Steering-layer sweeps across model scales.
    Here, $r=\ell/L$ denotes normalized injection depth.
    TPR@0.01\% and overlap are reported at the best validated layer as mean $\pm$ standard deviation over three seeds.
    }
    \label{tab:steering_depth_scaling}
    \small
    \setlength{\tabcolsep}{4.5pt}
    \resizebox{\textwidth}{!}{
    \begin{tabular}{lcccccc}
        \toprule
        Observer
        & Swept layers
        & High-performing interval
        & Best layer $\ell/L$
        & TPR@0.01\%
        & Overlap (\%) \\
        \midrule
        Llama-3.2-1B
        & 2, 4, 6, 8, 10, 12, 14
        & 4--6 ($r=0.25$--$0.375$)
        & 4/16
        & $94.09 \pm 2.94$
        & $2.29 \pm 0.10$ \\
        
        Meta-Llama-3-8B
        & 4, 8, 12, 16, 20, 24, 28
        & 8--16 ($r=0.25$--$0.50$)
        & 12/32
        & $97.14 \pm 0.66$
        & $1.92 \pm 0.36$ \\
        
        Meta-Llama-3-70B
        & 10, 20, 30, 40, 50, 60, 70
        & 20--50 ($r=0.25$--$0.625$)
        & 40/80
        & $97.52 \pm 0.54$
        & $1.97 \pm 0.20$ \\
        \bottomrule
    \end{tabular}}
\end{table}

In Table~\ref{tab:steering_depth_scaling}, across all three scales, strong performance is attained over a range of intermediate depths rather than at a single layer. The best normalized depths shift from $0.25$ to $0.375$ and $0.50$, while the high-performing region broadens from $0.25$--$0.375$ to $0.25$--$0.625$. Thus, the results do not support an exact proportional depth-scaling rule; instead, intermediate layers provide a robust search region whose effective range may broaden as model depth increases.

\subsection{Computational Efficiency}
\label{sec:computational_efficiency}
We evaluate the computational efficiency of \texttt{S2D} in terms of both inference latency and training cost, and compare it against representative baselines including L2D, RepreGuard, and Binoculars.

\paragraph{Experimental Setup.}
All efficiency profiling is conducted on a single NVIDIA A100 (80GB) GPU. To ensure a fair evaluation, the training phase for all trainable methods is conducted on a mixed dataset comprising 512 text pairs generated by multiple LLMs, and all inference metrics are evaluated over a correspondingly mixed test set. Our test set has an average length of approximately $267$ tokens, with $95\%$ of the samples shorter than $435$ tokens.

For all methods, we adopt \textit{meta-llama/Llama-3.1-8B} as the observer/proxy model. For Binoculars, \textit{meta-llama/Llama-3.1-8B} is used as the observer while \textit{meta-llama/Llama-3.1-8B-Instruct} serves as the performer. For L2D, the number of rewrites is set to 4. Unless otherwise specified, we report per-sample inference latency with a batch size of 1. Peak memory usage is measured using \texttt{torch.cuda.max\_memory\_allocated()}. Importantly, for the rewrite-based baseline (L2D), the reported training and inference times strictly \textit{exclude} the time required for the auxiliary LLM to autoregressively generate the rewrites, as these were pre-computed offline.

\paragraph{Results.}

\begin{table}[h]
\centering
\small
\setlength{\tabcolsep}{4pt}
\resizebox{\textwidth}{!}{
\begin{tabular}{lccccccc}
\toprule
\multirow{2}{*}{\textbf{Method}} & \multirow{2}{*}{\textbf{Rewrite?}} & \multicolumn{4}{c}{\textbf{Efficiency}} & \multicolumn{2}{c}{\textbf{Performance (\%)}} \\
\cmidrule(lr){3-6} \cmidrule(lr){7-8}
& & \textbf{Train Time (s)} $\downarrow$ & \textbf{Train Cost (GB)} $\downarrow$ & \textbf{Infer. Time (s)} $\downarrow$ & \textbf{Infer. Cost (GB)} $\downarrow$ & \textbf{AUROC} $\uparrow$ & \textbf{TPR@1\%} $\uparrow$ \\
\midrule
Binoculars & No  & - & - & 0.50 & 58.0 & 87.70 & 74.70 \\
L2D        & Yes & 1213.12 & 45.0 & 2.03 & 43.0 & \textbf{98.96} & 97.60 \\
RepreGuard & No  & 835.62 & 42.0 & 0.32 & \textbf{38.0} & 98.42 & 81.47 \\
\midrule
\texttt{S2D} (ours) & No & \textbf{759.46} & \textbf{40.0} & \textbf{0.30} & 39.0 & 98.90 & \textbf{97.75} \\
\bottomrule
\end{tabular}}
\vspace{3pt}
\caption{Comparison of computational cost and detection performance evaluated on a mixed test set. \textit{Train Cost} and \textit{Infer Cost} denote the peak memory allocated during the respective phases. Binoculars is a train-free method, hence its training costs are omitted. Note that for L2D, the heavy computational time required to generate rewrites is excluded from the timing metrics.}
\label{tab:efficiency_performance}
\end{table}

During \textbf{inference}, S2D requires only a single forward pass through the frozen observer LLM. Even when L2D is evaluated under the favorable assumption that rewrite generation is excluded, S2D still achieves a \textbf{6.7$\times$ speedup} (0.30s vs.\ 2.03s, $ 2.03 / 0.30\approx 6.7$).
Compared with Binoculars, S2D substantially reduces inference latency and lowers memory usage by avoiding the need to load two LLMs simultaneously. Compared with RepreGuard, S2D incurs virtually no additional runtime overhead while maintaining comparable memory usage, indicating that steering introduces negligible deployment cost. During \textbf{training}, S2D keeps the observer LLM frozen and optimizes only a lightweight steering vector. As a result, it achieves the lowest training time among the evaluated trainable detectors while maintaining low training memory cost.
Overall, these results show that S2D provides a scalable and deployable detection framework without sacrificing accuracy.

\subsection{Interpretability Analysis of the Steering Vector}
\label{app:steering_interpretability}

We further examine what the learned steering vector $\mathbf v$ represents and how it affects the observer model. For this analysis, we use the steering vector learned with \textit{Meta-Llama-3-8B} and evaluate it on an independent held-out set of 512 human--LLM text pairs from DetectRL.

\paragraph{Vocabulary Projection.}
Let $\mathbf w_j$ denote the unembedding direction of token $j$. We first examine the alignment between $\mathbf w_j$ and $\mathbf v$ using $\mathbf w_j^\top \mathbf v$ and $\cos(\mathbf w_j,\mathbf v)$. The most aligned tokens are heterogeneous, including multilingual fragments and incomplete subwords, and the maximum absolute cosine similarity is only $0.074$. This suggests that $\mathbf v$ does not correspond to a small set of identifiable lexical features. Moreover, since $\mathbf v$ is injected at an intermediate layer and subsequently transformed by the remaining frozen Transformer blocks, direct vocabulary projection provides only a coarse interpretation of its effect.

\paragraph{Effect on Output Logits.}
We next examine the change in output logits induced by steering. Let $\Delta \ell_j(x) :=\ell_j(x;\mathbf v)-\ell_j(x;\mathbf 0)$ denote the change in the output logit of token $j$. Tokens whose changes differ most between LLM-generated and human-written text include \texttt{incumbent}, \texttt{subsequently}, \texttt{aforementioned}, \texttt{plethora}, \texttt{concurrently}, and \texttt{rationale}, showing some association with formal written expression. However, direct vocabulary alignment poorly predicts these realized changes. The Spearman correlations between $\cos(\mathbf w_j,\mathbf v)$ and the overall average change $\mathbb E[\Delta \ell_j(X)]$, and between $\cos(\mathbf w_j,\mathbf v)$ and the class-conditional difference $\mathbb E[\Delta \ell_j(X)\mid Y=1]-\mathbb E[\Delta \ell_j(X)\mid Y=0]$, are $0.061$ and $-0.035$, respectively. In addition, the corresponding top-$50$ token sets have no overlap. Thus, the effect of steering cannot be explained well by nearest-token projections alone.

\paragraph{Effect on the Final Representation.}
Finally, we examine how steering changes the representation used by the detector. Recall from~\eqref{eq:score-for-detection} that the detection score depends on the projection of $f_{\theta,\mathbf v}(x)$ onto the direction $\widehat{\boldsymbol{\mu}}_1-\widehat{\boldsymbol{\mu}}_0$. We therefore consider $\left\langle \widehat{\boldsymbol{\mu}}_1-\widehat{\boldsymbol{\mu}}_0,\, f_{\theta,\mathbf v}(X)-f_{\theta}(X) \right\rangle$. Its average value is $-1.671$ for human-written text and $0.388$ for LLM-generated text, yielding an LLM-minus-human contrast of $2.059$. Hence, although the same vector $\mathbf v$ is injected for every input, its effect after propagation through the remaining Transformer layers is class dependent rather than a uniform shift of the representation space.

Overall, these results suggest that $\mathbf v$ is better viewed as a context-dependent intervention direction, whose effect emerges through the subsequent frozen Transformer layers and is reflected in the final representation used for detection.

\subsection{Additional Experiments}
\label{additional_exp}

We further evaluate \texttt{S2D} on the recent OpAI-Bench~\citep{bsharat2026operation} to examine whether the conclusions from DetectRL continue to hold for more recent generators and more challenging human-to-AI revision settings. OpAI-Bench covers four domains, including abstracts, essays, news, and reports, and contains texts produced or revised by \textit{GPT-5.4}, \textit{GPT-5.4-nano}, \textit{Gemini-2.5-Flash}, and \textit{Qwen3-8B}. Our setup uses 512 training pairs and 302 test documents. In addition to the baselines considered in the main experiments, we include several recent training-free detectors: Implicit Reward Model (\textit{IRM})~\citep{liu2025implicit}, \textit{Lastde++}~\citep{xu2025training}, \textit{SpecDetect++}~\citep{luo2026specdetect}, and Normalized Temperature Sensitivity (\textit{NTS})~\citep{ma2026zero}.

\paragraph{Evaluation on recent generators and detectors.}
Table~\ref{tab:opaibench} compares both training-free and train-based detectors on OpAI-Bench. Each entry reports AUROC / TPR@1\% / TPR@0.1\%, in percentage points. For training-free methods, no source-generator training is required. For train-based methods, we report results separately for each training generator; the corresponding diagonal entry is in-distribution, while the remaining columns measure cross-generator transfer. The recent training-free methods improve substantially over several classical likelihood-based detectors at low false-positive rates, with \textit{IRM} generally providing the strongest low-FPR performance among this group. Nevertheless, a substantial gap remains between training-free detection and the strongest train-based approaches.

\begin{wraptable}{r}{0.53\textwidth}
    \vspace{-12pt}
    \centering
    \caption{
    Comparison with discriminative probes. Each entry reports AUROC / TPR@1\%.
    }
    \label{tab:probe_comparison}
    \scriptsize
    \setlength{\tabcolsep}{10pt}
    \renewcommand{\arraystretch}{0.95}
    \begin{tabular}{llcc}
        \toprule
        Revision & Method & ID & OOD \\
        \midrule
        \multirow{4}{*}{\texttt{v4}}
        & LogReg     & 0.92 / 0.45 & 0.87 / 0.34 \\
        & Linear SVM & 0.91 / 0.43 & 0.87 / 0.33 \\
        & MLP        & 0.92 / \textbf{0.50} & 0.88 / 0.37 \\
        & \texttt{S2D}
                     & \textbf{0.94} / 0.44
                     & \textbf{0.90} / \textbf{0.38} \\
        \midrule
        \multirow{4}{*}{\texttt{v6}}
        & LogReg     & 0.98 / 0.84 & \textbf{0.96} / \textbf{0.71} \\
        & Linear SVM & 0.98 / 0.84 & \textbf{0.96} / 0.70 \\
        & MLP        & 0.98 / 0.85 & 0.95 / 0.70 \\
        & \texttt{S2D}
                     & \textbf{0.99} / \textbf{0.86}
                     & \textbf{0.96} / \textbf{0.71} \\
        \midrule
        \multirow{4}{*}{\texttt{v8}}
        & LogReg     & \textbf{0.99} / 0.89 & 0.96 / 0.79 \\
        & Linear SVM & \textbf{0.99} / 0.90 & 0.96 / 0.79 \\
        & MLP        & \textbf{0.99} / 0.91 & 0.96 / 0.77 \\
        & \texttt{S2D}
                     & \textbf{0.99} / \textbf{0.92}
                     & \textbf{0.97} / \textbf{0.81} \\
        \bottomrule
    \end{tabular}
    \vspace{-10pt}
\end{wraptable}

Across the train-based settings, \texttt{S2D} is generally the strongest method or competitive with the strongest baseline, particularly at low FPRs under cross-generator transfer. For example, when trained on \textit{Gemini-2.5-Flash}, \texttt{S2D} achieves TPR@0.1\% of 53.65, 74.92, and 89.40 on \textit{GPT-5.4}, \textit{GPT-5.4-nano}, and \textit{Qwen3-8B}, respectively, substantially exceeding the corresponding values of the other train-based detectors. Similar gains are observed for several other training--test combinations. The advantage is not uniform across every metric, as RoBERTa-Large or RepreGuard can match or slightly outperform \texttt{S2D} in a few settings. Overall, these results show that the strong performance of \texttt{S2D} persists when both the generators and competing detectors are updated.

\begin{table}[t]
    \centering
    \caption{
    Detection performance on OpAI-Bench.
    Each entry reports AUROC / TPR@1\% / TPR@0.1\%, in percentage points.
    For train-based methods, results are mean $\pm$ standard deviation.
    ``--'' denotes training-free methods that do not use a source generator.
    }
    \label{tab:opaibench}
    \scriptsize
    \newcommand{\std}[1]{{\tiny$\pm$#1}}
    \setlength{\tabcolsep}{2.3pt}
    \resizebox{\textwidth}{!}{
    \begin{tabular}{llcccc}
        \toprule
        Train & Method
        & Gemini-2.5-Flash
        & GPT-5.4
        & GPT-5.4-nano
        & Qwen3-8B \\
        \midrule

        \multirow{8}{*}{--}
        & Likelihood
        & 71.42 / 0.10 / 0.00
        & 49.45 / 0.03 / 0.00
        & 71.37 / 0.15 / 0.00
        & 59.45 / 0.04 / 0.00 \\

        & LRR
        & 77.73 / 3.66 / 0.76
        & 51.09 / 0.22 / 0.00
        & 72.25 / 0.22 / 0.00
        & 60.00 / 0.22 / 0.00 \\

        & Fast-DetectGPT
        & 73.66 / 8.78 / 2.53
        & 49.79 / 1.33 / 0.30
        & 76.62 / 12.92 / 3.06
        & 47.34 / 1.66 / 0.40 \\

        & Binoculars
        & 52.12 / 0.07 / 0.00
        & 59.18 / 1.66 / 0.43
        & 36.85 / 0.07 / 0.00
        & 74.54 / 3.98 / 0.63 \\

        & IRM
        & 72.33 / 17.93 / 6.70
        & 59.52 / 6.32 / 1.49
        & 66.28 / 8.30 / 1.69
        & 74.02 / 28.19 / 10.95 \\

        & Lastde++
        & 67.88 / 4.66 / 1.01
        & 46.62 / 1.46 / 0.54
        & 71.03 / 6.98 / 1.85
        & 41.79 / 0.46 / 0.00 \\

        & SpecDetect++
        & 58.36 / 2.07 / 0.70
        & 43.59 / 0.76 / 0.00
        & 62.46 / 4.40 / 1.75
        & 38.43 / 0.45 / 0.11 \\

        & NTS
        & 57.81 / 2.03 / 0.73
        & 51.36 / 1.36 / 0.30
        & 51.91 / 2.00 / 0.20
        & 51.36 / 1.68 / 0.40 \\
        \midrule

        \multirow{4}{*}{Gemini}
        & RoBERTa-L
        & 100.00\std{0.00} / 100.00\std{0.00} / 99.44\std{0.01}
        & 79.51\std{0.03} / 36.11\std{0.04} / 16.64\std{0.07}
        & 87.10\std{0.02} / 38.75\std{0.04} / 15.88\std{0.05}
        & 97.80\std{0.01} / 76.00\std{0.02} / 49.87\std{0.04} \\

        & ImBD
        & 81.06\std{0.02} / 13.91\std{0.02} / 10.70\std{0.02}
        & 68.81\std{0.01} / 3.87\std{0.06} / 2.26\std{0.05}
        & 47.36\std{0.03} / 0.27\std{0.02} / 0.10\std{0.07}
        & 79.58\std{0.02} / 12.59\std{0.03} / 9.51\std{0.05} \\

        & RepreGuard
        & 97.82\std{0.04} / 84.77\std{0.02} / 77.07\std{0.02}
        & 76.41\std{0.03} / 17.12\std{0.14} / 6.07\std{0.12}
        & 84.98\std{0.04} / 34.36\std{0.04} / 21.01\std{0.10}
        & 91.90\std{0.03} / 57.30\std{0.02} / 48.16\std{0.02} \\

        & \texttt{S2D}
        & \textbf{100.00}\std{0.00} / \textbf{100.00}\std{0.00} / \textbf{100.00}\std{0.00}
        & \textbf{89.81}\std{0.02} / \textbf{60.53}\std{0.05} / \textbf{53.65}\std{0.06}
        & \textbf{94.75}\std{0.01} / \textbf{80.97}\std{0.03} / \textbf{74.92}\std{0.06}
        & 97.25\std{0.00} / \textbf{92.92}\std{0.01} / \textbf{89.40}\std{0.03} \\
        \midrule

        \multirow{4}{*}{GPT-5.4}
        & RoBERTa-L
        & \textbf{98.76}\std{0.02} / 83.86\std{0.22} / 59.14\std{0.12}
        & 98.38\std{0.01} / 70.28\std{0.13} / \textbf{44.93}\std{0.25}
        & 97.38\std{0.02} / \textbf{69.62}\std{0.21} / 41.85\std{0.26}
        & 95.32\std{0.02} / 62.98\std{0.19} / 41.22\std{0.10} \\

        & ImBD
        & 72.38\std{0.03} / 2.34\std{0.24} / 1.71\std{0.10}
        & 87.99\std{0.02} / 16.72\std{0.15} / 11.03\std{0.23}
        & 83.28\std{0.03} / 9.07\std{0.23} / 7.77\std{0.28}
        & 83.43\std{0.03} / 11.31\std{0.21} / 6.39\std{0.08} \\

        & RepreGuard
        & 97.12\std{0.02} / 80.75\std{0.32} / \textbf{72.64}\std{0.02}
        & 80.46\std{0.16} / 21.09\std{0.36} / 11.54\std{0.20}
        & 87.89\std{0.02} / 44.95\std{0.01} / 29.18\std{1.20}
        & 92.14\std{0.52} / 57.42\std{0.36} / 47.41\std{1.02} \\

        & \texttt{S2D}
        & 97.67\std{0.01} / \textbf{92.13}\std{0.23} / 71.58\std{0.11}
        & \textbf{98.56}\std{0.00} / \textbf{81.04}\std{0.14} / 38.35\std{0.24}
        & \textbf{98.36}\std{0.01} / 68.29\std{0.22} / \textbf{49.94}\std{0.27}
        & \textbf{96.47}\std{0.01} / \textbf{73.09}\std{0.20} / \textbf{56.38}\std{0.09} \\
        \midrule

        \multirow{4}{*}{GPT-5.4-nano}
        & RoBERTa-L
        & 99.29\std{0.01} / 90.44\std{0.01} / 83.32\std{0.02}
        & 93.88\std{0.05} / 61.28\std{0.08} / 51.83\std{0.03}
        & \textbf{99.67}\std{0.01} / 95.04\std{0.02} / 90.77\std{0.02}
        & 97.17\std{0.02} / 84.12\std{0.02} / 78.35\std{0.02} \\

        & ImBD
        & 64.24\std{0.02} / 1.99\std{0.02} / 0.86\std{0.03}
        & 83.40\std{0.03} / 13.25\std{0.10} / 7.78\std{0.01}
        & 86.34\std{0.02} / 19.91\std{0.01} / 13.68\std{0.03}
        & 86.53\std{0.03} / 18.89\std{0.03} / 14.05\std{0.03} \\

        & RepreGuard
        & 97.55\std{0.02} / 82.15\std{0.02} / 77.35\std{0.02}
        & 81.13\std{0.02} / 27.33\std{0.12} / 23.64\std{0.10}
        & 90.24\std{0.04} / 56.27\std{0.02} / 49.08\std{0.02}
        & 92.69\std{0.03} / 57.32\std{0.03} / 52.09\std{0.04} \\

        & \texttt{S2D}
        & \textbf{99.97}\std{0.00} / \textbf{99.47}\std{0.00} / \textbf{98.57}\std{0.01}
        & \textbf{95.93}\std{0.04} / \textbf{70.32}\std{0.09} / \textbf{59.93}\std{0.02}
        & 99.64\std{0.00} / \textbf{96.20}\std{0.01} / \textbf{94.16}\std{0.01}
        & \textbf{98.16}\std{0.01} / \textbf{93.53}\std{0.01} / \textbf{89.69}\std{0.01} \\
        \midrule

        \multirow{4}{*}{Qwen3-8B}
        & RoBERTa-L
        & 99.74\std{0.01} / 97.04\std{0.01} / 68.90\std{0.02}
        & 89.82\std{0.02} / 36.47\std{0.04} / 10.37\std{0.06}
        & 94.95\std{0.02} / 64.28\std{0.04} / 26.62\std{0.08}
        & \textbf{99.42}\std{0.01} / \textbf{94.38}\std{0.01} / 75.63\std{0.02} \\

        & ImBD
        & 72.73\std{0.02} / 4.98\std{0.02} / 4.67\std{0.03}
        & 74.36\std{0.01} / 4.00\std{0.06} / 3.38\std{0.08}
        & 64.85\std{0.03} / 1.68\std{0.02} / 1.36\std{0.10}
        & 90.36\std{0.02} / 32.81\std{0.02} / 30.07\std{0.03} \\

        & RepreGuard
        & 96.53\std{0.02} / 78.81\std{0.02} / 62.18\std{0.03}
        & 71.61\std{0.12} / 11.71\std{0.23} / 3.65\std{0.12}
        & 80.13\std{0.03} / 23.87\std{0.01} / 8.29\std{0.14}
        & 91.61\std{0.02} / 61.85\std{0.02} / 44.35\std{0.02} \\

        & \texttt{S2D}
        & \textbf{100.00}\std{0.00} / \textbf{100.00}\std{0.00} / \textbf{99.46}\std{0.01}
        & \textbf{94.26}\std{0.01} / \textbf{56.79}\std{0.05} / \textbf{41.04}\std{0.07}
        & \textbf{96.69}\std{0.01} / \textbf{82.52}\std{0.03} / \textbf{65.49}\std{0.09}
        & 99.20\std{0.00} / 91.36\std{0.00} / \textbf{88.01}\std{0.01} \\
        \bottomrule
    \end{tabular}}
\end{table}

\paragraph{Comparison with discriminative probes.}
We next examine whether the performance of \texttt{S2D} can be reproduced by directly fitting a standard discriminative classifier to the observer representations. We compare against logistic regression, a linear SVM, and a small MLP. All methods use the same frozen Llama-3-8B observer, representation pooling strategy, and training data, while the probe baselines operate on the corresponding pooled representations without steering. We evaluate on the OpAI-Bench revision trajectories, where \texttt{v4}, \texttt{v6}, and \texttt{v8} correspond to progressively stronger AI revisions of the same human-written documents, with approximately 50\%, 75\%, and 100\% of the sentences edited, respectively. As shown in Table~\ref{tab:probe_comparison}, \texttt{S2D} is generally the strongest method or tied with the strongest probe. The advantage is particularly consistent under cross-generator evaluation: \texttt{S2D} achieves the highest OOD AUROC and TPR@1\% for \texttt{v4} and \texttt{v8}, while matching the best probe for \texttt{v6}. The gains are not uniform across all metrics; for example, the MLP achieves a higher ID TPR@1\% on \texttt{v4}. Overall, directly fitting a discriminative classifier to the unsteered representations does not consistently recover the performance of \texttt{S2D}, particularly under cross-generator transfer. However, since \texttt{S2D} combines steering with a vMF-based readout whereas the probe baselines use unsteered representations with different classifiers, this comparison does not fully disentangle the individual contributions of steering and the downstream classifier.

\section{Omitted Theoretical Details and Proofs}\label{app_pf}

\subsection{Proof of Theorem~\ref{thm_1}}\label{app_pf_thm1}
\begin{proof}[\textbf{Proof of Theorem~\ref{thm_1}}]
Condition on the training sample $\mathcal{S}_{\mathrm{train}}$, the estimated scoring function $\widehat{\mathcal S}_t$ is fixed, and the only remaining randomness comes from the calibration sample $\mathcal{S}_{\mathrm{cal}}=\{x_i^-\}_{i=1}^{n_2}$ with $x_i^- \overset{\mathrm{i.i.d.}}{\sim}\mathbb P_0$. 

Define the true survival function $\mathcal T(\tau):=\mathbb P_0\bigl(\widehat{\mathcal S}_t(X_{\mathrm{test}}^{-})> \tau\bigr)$ (where $X_{\mathrm{test}}^{-}$ follows the same distribution of samples in $\mathcal{S}_{\mathrm{cal}}$) and its empirical counterpart $\widehat{\mathcal T}_{n_2}(\tau):=\frac{1}{n_2}\sum_{i=1}^{n_2}\mathbbm 1\bigl(\widehat{\mathcal S}_t(x_i^-)> \tau\bigr)$. By the definition of the empirical threshold $\widehat\tau_{\alpha,t}$ in Equation~\eqref{type_1_error}, it is the infimum of $\tau$ such that $\widehat{\mathcal T}_{n_2}(\tau) \le \alpha$. On one hand, by definition, we have $\widehat{\mathcal T}_{n_2}(\widehat\tau_{\alpha,t}) \le \alpha$. On the other hand, by the continuity assumption, the calibration scores are almost surely distinct, so $\widehat{\mathcal T}_{n_2}(\tau)$ is a step function whose jumps have size $1/n_2$. The infimum definition therefore guarantees that the empirical probability does not fall below $\alpha$ by more than $1/n_2$. Therefore, $\alpha - \widehat{\mathcal T}_{n_2}(\widehat\tau_{\alpha,t}) \le \frac{1}{n_2}$.
It implies that
\[
\left|\widehat{\mathcal T}_{n_2}(\widehat\tau_{\alpha,t}) - \alpha\right| \le \frac{1}{n_2}.
\]

By applying the triangle inequality, we can bound the absolute difference between the population Type-I error and the target level $\alpha$
\begin{align*}
\left|\mathcal T(\widehat\tau_{\alpha,t}) - \alpha\right| 
&\le \left|\mathcal T(\widehat\tau_{\alpha,t}) - \widehat{\mathcal T}_{n_2}(\widehat\tau_{\alpha,t})\right| + \left|\widehat{\mathcal T}_{n_2}(\widehat\tau_{\alpha,t}) - \alpha\right| \\
&\le \sup_{\tau\in\mathbb R}\left|\mathcal T(\tau)-\widehat{\mathcal T}_{n_2}(\tau)\right| + \frac{1}{n_2}.
\end{align*}

By using the Dvoretzky--Kiefer--Wolfowitz~(DKW) inequality~\citep{massart1990tight}, the uniform deviation between the true and empirical functions is bounded by
\[
\mathbb P\!\left(
\sup_{\tau\in\mathbb R}
\left|\mathcal T(\tau)-\widehat{\mathcal T}_{n_2}(\tau)\right|
>
\sqrt{\frac{\log(2/\delta)}{2n_2}}
\,\middle|\,
\mathcal S_{\mathrm{train}}
\right)
\le
2e^{-2n_2 \left(\sqrt{\frac{\log(2/\delta)}{2n_2}}\right)^2}
=
\delta.
\]
Hence, conditional on $\mathcal S_{\mathrm{train}}$, with probability at least $1-\delta$, we have
\[
\left|\mathcal T(\widehat\tau_{\alpha,t}) - \alpha\right| \le \sqrt{\frac{\log(2/\delta)}{2n_2}} + \frac{1}{n_2}.
\]

Since $\mathbb{P}_0(\widehat{\mathcal S}_t(X_{\mathrm{test}}^{-}) > \widehat\tau_{\alpha,t}) = \mathcal T(\widehat\tau_{\alpha,t})$, we can rewrite the conditional probability as:
\[
\mathbb P\!\left(
\left|\mathbb P_0\bigl(\widehat{\mathcal S}_t(X_{\mathrm{test}}^{-})>\widehat\tau_{\alpha,t}\bigr) - \alpha\right|
\le
\sqrt{\frac{\log(2/\delta)}{2n_2}} + \frac{1}{n_2}
\,\middle|\,
\mathcal S_{\mathrm{train}}
\right)
\ge
1-\delta,
\]
which yields
\[
\mathbb P\!\left(
\left|\mathbb P_0\bigl(\widehat{\mathcal S}_t(X_{\mathrm{test}}^{-})>\widehat\tau_{\alpha,t}\bigr) - \alpha\right|
\le
\sqrt{\frac{\log(2/\delta)}{2n_2}} + \frac{1}{n_2}
\right)
\ge
1-\delta,
\]
which completes the proof.
\end{proof}

\subsection{Formal Statement for Theorem~\ref{thm_2}}
\label{append:assumptions}

Theorem~\ref{thm_2} characterizes the excess Type II error of the detector trained on $\mathcal{S}_{\mathrm{train}}$. 
To establish this result, we introduce several additional technical assumptions. 
We first present and discuss these assumptions, followed by the formal statement of Theorem~\ref{thm_2}. 


\paragraph{Two-Time-Scale Prototype Tracking.}
Algorithm~\ref{alg:vmf_latent_steering} can be viewed as a two-time-scale stochastic approximation scheme for solving the empirical objective~\eqref{eq:empirical-loss}:
\begin{equation}
\label{eq:empirical-loss}
\mathcal L_{\mathrm{vMF}}
(\mathbf v,\boldsymbol{\mu}_0,\boldsymbol{\mu}_1)
:=
\frac1{n_1}\sum_{i=1}^{n_1}
\log q_\beta\left(
y_i\mid
f_{\theta,\mathbf v}(x_i),
\boldsymbol{\mu}_0,
\boldsymbol{\mu}_1
\right),
\end{equation}
where $\mathcal{S}_{\mathrm{train}} = \{(x_i, y_i)\}_{i=1}^{n_1}$ consists of i.i.d.\ samples from $\mathcal{P}$.

We restate the key update steps for clarity and subsequent analysis.
At iteration $t$, let $\mathbf{v}_t$ denote the current steering vector. For each class $c$, define
\[
\bar{\mathbf{z}}_{c,t}(\mathbf{v}_{t}) = \frac{1}{|\{(x,y):x\in \mathcal{B}_{t}, y=c\}|}\sum_{(x,y):x\in \mathcal{B}_{t}, y=c} f_{\theta,\mathbf{v}_t}(x)
\]
as the average representation of class $c$ over a fresh mini-batch $\mathcal{B}_{t}$.
The update rule is given by
\begin{subequations}
\begin{align}
    &\textbf{Fast Process (Mean Direction Tracking)} \quad
    \widehat{\boldsymbol{\mu}}_{c,t+1}
    =
    \Pi_{\mathbb{S}^{d-1}}
    \left(
    (1-\rho)\widehat{\boldsymbol{\mu}}_{c,t}
    +
    \rho\,\bar{\mathbf{z}}_{c,t}(\mathbf{v}_t)
    \right), \label{eq:fast_dynamics1} \\
    &\textbf{Slow Process (Steering Vector Update)} \quad
    \mathbf{v}_{t+1}
    =
    \mathbf{v}_t
    +
    \eta \nabla_{\mathbf{v}}
    \mathcal{L}_{\mathrm{vMF}}
    \bigl(
    \mathbf{v}_t,
    \widehat{\boldsymbol{\mu}}_{0,t},
    \widehat{\boldsymbol{\mu}}_{1,t}
    \bigr), \label{eq:slow_dynamics1}
\end{align}
\end{subequations}
where $\Pi_{\mathbb{S}^{d-1}}(\mathbf{x})=\mathbf{x}/\|\mathbf{x}\|$ denotes projection onto the unit sphere. 
Here, $\widehat{\boldsymbol{\mu}}_{c,t}$ serves as an estimator of the population mean direction for class $c$ at iteration $t$. 
We consider a two-time-scale regime with $0<\eta \ll \rho \le 1$, under which $\{\widehat{\boldsymbol{\mu}}_{c,t}\}_{c=0}^1$ evolves on the fast timescale while $\mathbf{v}_t$ evolves on the slow timescale.

\paragraph{Assumptions.}

Assumption~\ref{asump1:true-model} introduces a generative model under which the representation follows a vMF distribution at the population optimum. This assumption facilitates the theoretical analysis.

\begin{assumption}[Local vMF Model]
\label{asump1:true-model}
For every $\mathbf v$ and each class $c\in\{0,1\}$, define $Z_{\mathbf{ v}}=f_{\theta,\mathbf v}(X)$ and assume
\[
Z_{\mathbf{ v}} \mid Y=c \sim \mathrm{vMF}(\boldsymbol\mu_c(\mathbf v),\kappa),
\]
that is, its density is given by $p(z \mid Y=c)
=
C_d(\kappa)\exp\bigl(\kappa\, \boldsymbol{\mu}_c(\mathbf{v})^{\top} z\bigr)$ for any $z \in \mathbb S^{d-1}$, where $C_d(\kappa)$ is the normalization constant.
\end{assumption}

\begin{remark}[Gradient Noises]
Let $\mathcal{F}_t$ denote the filtration generated by the optimization trajectory up to iteration $t$.
Under Assumption~\ref{asump1:true-model}, the conditional expectation of the mini-batch average satisfies
\[
\mathbb{E}\!\left[\bar{\mathbf{z}}_{c,t}(\mathbf{v}_t)\mid \mathcal{F}_t\right]
=
A_d(\kappa)\boldsymbol{\mu}_c(\mathbf{v}_t),
\]
where  $A_d(\kappa)
=
\frac{I_{d/2}(\kappa)}{I_{d/2-1}(\kappa)}$, and $I_\nu(\cdot)$ denotes the modified Bessel function of the first kind of order $\nu$. 
Based on the observation, we can define the following gradient noise
\begin{equation}
\label{eq:noise}
\boldsymbol{\zeta}_{c,t+1} := \bar{\mathbf{z}}_{c,t}(\mathbf{v}_t)-A_d(\kappa) \boldsymbol{\mu}_c(\mathbf{v}_t),
\end{equation}
for a class $c \in \{0,1\}$. 
We then have that $\mathbb{E}\!\left[\boldsymbol{\zeta}_{c,t+1}\mid \mathcal{F}_t\right] = 0$ and $\|\boldsymbol{\zeta}_{c,t+1} \| \le 2$ almost surely due to the fact that $\|\bar{\mathbf{z}}_{c,t}(\mathbf{v}_t)\|, A_d(\kappa)$, and  $\|\boldsymbol{\mu}_c(\mathbf{v}_t)\| \le 1$.
\end{remark}

\begin{assumption}[Local Tracking Conditions]\label{asump2:local_tracking}
Define $\boldsymbol{\mu}_{c,t} := \boldsymbol{\mu}_{c}(\mathbf{v}_t)$ and assume that
\begin{enumerate}
\item \textbf{(Bounded drift)} 
There exists $C_\mu > 0$ such that $\|\boldsymbol{\mu}_{c,t+1} - \boldsymbol{\mu}_{c,t}\| \le C_{\mu} \eta.$

\item \textbf{(Initialization)}  $\|\widehat{\boldsymbol{\mu}}_{c,0} - \boldsymbol{\mu}_{c,0}\| \le 1/4.$
\end{enumerate}
\end{assumption}

\begin{theorem}[Local Mean Direction Tracking]\label{thm:local_tracking}
Fix $c\in\{0,1\}$ and write
$\boldsymbol{\mu}_{c,t}:=\boldsymbol{\mu}_c(\mathbf v_t)$.
Consider the recursion
\[
\mathbf z_{c,t+1}
=
(1-\rho)\widehat{\boldsymbol{\mu}}_{c,t}
+
\rho \bar{\mathbf z}_{c,t}(\mathbf v_t),
\qquad
\widehat{\boldsymbol{\mu}}_{c,t+1}
=
\frac{\mathbf z_{c,t+1}}{\|\mathbf z_{c,t+1}\|}.
\]
Under Assumptions~\ref{asump1:true-model} and
\ref{asump2:local_tracking}, there exist constants
$c,C,\rho_0,\gamma_0>0$ such that, if
\[
0<\rho\le\rho_0,
\qquad
\frac{\eta}{\rho}\le\gamma_0,
\qquad
\rho\log\frac{2T}{\delta}\le 1,
\]
and
\[
\|\widehat{\boldsymbol{\mu}}_{c,0}-\boldsymbol{\mu}_{c,0}\|^2
+
C\left(
\sqrt{\rho\log\frac{2T}{\delta}}
+\frac{\eta^2}{\rho^2}
\right)
\le \frac{1}{16},
\]
then, with probability at least $1-\delta$, 
\begin{enumerate}
    \item For all $0 \le t\le T$, $\|\widehat{\boldsymbol{\mu}}_{c,t}-\boldsymbol{\mu}_{c,t}\|^2
    \le
    (1 - {c}\rho)^t
    \|\widehat{\boldsymbol{\mu}}_{c,0}-\boldsymbol{\mu}_{c,0}\|^2
    +
    C\left(\sqrt{\rho\log\frac{2T}{\delta}}+\frac{\eta^2}{\rho^2}\right).$
    \item For all $0 \le t\le T$, $\|\widehat{\boldsymbol{\mu}}_{c,t}-\boldsymbol{\mu}_{c,t}\|\le \frac14,~
    \|\mathbf z_{c,t+1}\|\ge \frac{1}{4},~
    \langle \widehat{\boldsymbol{\mu}}_{c,t},\boldsymbol{\mu}_{c,t}\rangle>0.$
\end{enumerate}
\end{theorem}


Theorem \ref{thm:local_tracking} establishes that the proposed exponential moving average update $\widehat{\boldsymbol{\mu}}_{c,t}$ can accurately track the evolving class prototypes ${\boldsymbol{\mu}}_{c,t}$ under the local vMF model. Despite stochastic gradient noise and slow drift in the true mean direction, the estimated prototype remains close to the target with high probability. The result shows a contraction behavior up to a steady-state error, where the averaging parameter and the drift rate govern the tracking accuracy. Its proof is deferred in Appendix \ref{proof:local_tracking}.

Since the score function depends on the class prototypes only through their difference, the tracking error bound in Theorem \ref{thm:local_tracking} directly translates into a uniform bound on the deviation between the plug-in and oracle score functions (i.e., $\widehat{\mathcal{S}}_t$ and ${\mathcal{S}}_t$), as formalized in the following corollary.

\begin{corollary}\label{cor:score_deviation}
Define the oracle and plug-in log-likelihood ratios with respect to the same steered representation $f_{\theta,\mathbf v_t}(x)$ as
\[
\mathcal{S}_t(x)
=
\kappa (\boldsymbol{\mu}_{1,t}-\boldsymbol{\mu}_{0,t})^\top f_{\theta,\mathbf v_t}(x),
\qquad
\widehat{\mathcal{S}}_t(x)
=
\kappa (\widehat{\boldsymbol{\mu}}_{1,t}-\widehat{\boldsymbol{\mu}}_{0,t})^\top f_{\theta,\mathbf v_t}(x).
\]
Under Assumptions~\ref{asump1:true-model} and \ref{asump2:local_tracking}, for any $\delta\in(0,1)$, suppose that the conditions of Theorem~\ref{thm:local_tracking} hold for both classes with failure probability $\delta/2$. Then, with probability at least $1-\delta$, there exist constants $c\in(0,1)$ and $C>0$ such that the following uniform bound holds for all $0\le t\le T$:
\[
\sup_{x\in\mathcal X} \left|\widehat{\mathcal{S}}_t(x)-\mathcal{S}_t(x)\right|
\le r_t(\delta, \rho, \eta),
\]
where 
\[
r_t(\delta, \rho, \eta)
:=
\kappa C \sqrt{
(1-c\rho)^t
+
\sqrt{\rho \log \frac{4T}{\delta}}
+
\frac{\eta^2}{\rho^2}
}.
\]
\end{corollary}

\begin{proof}[Proof of Corollary \ref{cor:score_deviation}]
Apply Theorem~\ref{thm:local_tracking} separately to $c=0$ and $c=1$, each with failure probability $\delta/2$. By a union bound, with probability at least $1-\delta$, the tracking bounds hold simultaneously for both classes and for all $0\le t\le T$.

By definition, we have
$$
\begin{aligned}
\left|\widehat{\mathcal{S}}_t(x)-\mathcal{S}_t(x)\right|
&=
\left|\kappa (\widehat{\boldsymbol{\mu}}_{1,t}-\widehat{\boldsymbol{\mu}}_{0, t})^\top f_{\theta,\mathbf v}(x)-\kappa (\boldsymbol{\mu}_{1,t}-\boldsymbol{\mu}_{0,t})^\top f_{\theta,\mathbf v}(x)
\right|\\
&\overset{(a)}{\le
}
\kappa
\left(\|\widehat{\boldsymbol{\mu}}_{0,t}-\boldsymbol{\mu}_{0,t}\| + \|\widehat{\boldsymbol{\mu}}_{1,t}-\boldsymbol{\mu}_{1,t}\|\right)\\
&\overset{(b)}{\le} 2 \kappa \sqrt{ (1 - {c}\rho)^t \cdot \frac{1}{16}
    +
    C\left(\sqrt{\rho\log\frac{4T}{\delta}}+\frac{\eta^2}{\rho^2}\right)} \\
&\overset{(c)}{\le
} \kappa C \sqrt{(1- {c} \rho)^t + \sqrt{\rho \log \frac{4 T}{\delta}}+\frac{\eta^2}{\rho^2}}.
\end{aligned}
$$
where (a) follows from the triangle inequality and the fact that $\|f_{\theta,\mathbf v}(x)\|\le 1$, (b) applies Theorem \ref{thm:local_tracking}, and (c) enlarges the constant $C$ properly.
\end{proof}

\begin{assumption}[Local Mass Condition Around the Threshold]
\label{asump3:local_mass}
There exist constants $c_0, C_0 > 0$, exponents $\underline{\gamma}, \bar{\gamma} > 0$ satisfying $\underline{\gamma}\le 1 + \bar{\gamma}$, and $\varepsilon_0 > 0$ such that for all $\varepsilon \in [0,\varepsilon_0]$,
\[
\mathbb{P}_0\!\left(
\left|\mathcal{S}(X)-\tau_{\alpha}^{\star}\right|
\le \varepsilon
\right)
\le
C_0 \varepsilon^{\bar{\gamma}},
\]
and
\[
\min\left\{
\mathbb{P}_0\!\left(
\tau_{\alpha}^{\star}-\varepsilon
\le
\mathcal{S}(X)
<
\tau_{\alpha}^{\star}
\right),
\,
\mathbb{P}_0\!\left(
\tau_{\alpha}^{\star}
<
\mathcal{S}(X)
\le
\tau_{\alpha}^{\star}+\varepsilon
\right)
\right\}
\ge
c_0 \varepsilon^{\underline{\gamma}}.
\]
\end{assumption}

This assumption imposes regularity conditions on the null score distribution in a neighborhood of $\tau_{\alpha}^{\star}$. 
The upper bound controls the concentration of probability mass near the threshold, corresponding to a standard margin-type condition in classification theory~\citep{polonik1995measuring,mammen1999smooth,tong2013plug,zhou2025adadetectgpt}. 
The lower bound ensures that there is sufficient probability mass around the threshold, preventing the distribution from being too sparse in this region.


\begin{theorem}[Formal Statement of Theorem \ref{thm_2}]\label{thm:formal}
Let Assumptions~\ref{asump1:true-model}-\ref{asump3:local_mass} hold, and suppose that the conditions of Theorem~\ref{thm:local_tracking} hold for both classes with failure probability $\delta/2$. Suppose further that the locality condition in Lemma~\ref{lem:diff-bound} holds. Then there exists a constant ${c} \in (0, 1)$ such that, with probability at least $1-2\delta$ over the randomness in the training procedure and the calibration sample, the excess Type-II error satisfies
$$
\begin{aligned}
&\mathbb{P}_1(\widehat{\mathcal R}_t)-\mathbb{P}_1(\mathcal R_t^{\star})
\le
\mathcal{O}\Bigg( \underbrace{\Big((1- {c}\rho)^t+\sqrt{\rho \log (4T / \delta)}+\eta^2 / \rho^2 \Big)^{\frac{1+\bar{\gamma}}{2}}}_{\text{Estimation Error}}  + \underbrace{ 
\sqrt{\frac{\log(2/\delta)}{n_2}} + \frac{1}{n_2}}_{\text{Calibration Error}} \Bigg)
\end{aligned}
$$
where $t\in [1, T]$ is the iteration index within the total horizon $T$, $\rho \in (0,1)$ is the EMA coefficient satisfying $\rho \le 1 / {c}$, $\eta>0$ is the steering learning rate, $n_2$ is the calibration sample size, $\delta \in (0,1)$ is the confidence parameter, and $\underline{\gamma}, \bar{\gamma}$ are the constants in Assumption~\ref{asump3:local_mass}.
\end{theorem}

\subsection{Proof of Theorem~\ref{thm:formal}}
\begin{proof}[Proof of Theorem~\ref{thm:formal}]
For notational simplicity, we suppress the time index $t$. The excess Type-II error admits the decomposition given in Lemma~\ref{lem:error-decomposition}.
We prove Lemma \ref{lem:error-decomposition} in Appendix \ref{proof:error-decomposition}.

\begin{lemma}
\label{lem:error-decomposition}
Let $p_0$ and $p_1$ denote the density functions of $\mathbb{P}_0$ and $\mathbb{P}_1$, respectively. Define the score function $\mathcal{S}(x) = \log \frac{p_1(x)}{p_0(x)}$, and the acceptance regions
\[
\mathcal{R}^{\star}
=
\left\{x : \mathcal{S}(x) \le \tau_{\alpha}^{\star} \right\},
\qquad
\widehat{\mathcal{R}}
=
\left\{x : \widehat{\mathcal{S}}(x) \le \widehat{\tau}_{\alpha} \right\}.
\]
Then,
\begin{equation}\label{type_ii_error_eq_1}
\begin{aligned}
\mathbb{P}_1(\widehat{\mathcal{R}}) - \mathbb{P}_1(\mathcal{R}^{\star}) = \underbrace{\int_{\widehat{\mathcal{R}} \Delta \mathcal{R}^{\star}}\left|\frac{p_1(x)}{p_0(x)}-e^{\tau_{\alpha}^{\star}}\right| \mathrm{d} \mathbb{P}_0(x)}_{\text{ Estimation Error}}  + \underbrace{e^{\tau_{\alpha}^{\star}} \left(\alpha- \mathbb{P}_0(\widehat{\mathcal{R}}^c)\right)}_{\text{Price of Conservativeness}}.
\end{aligned}
\end{equation}   
Here, $\widehat{\mathcal{R}} \Delta \mathcal{R}^{\star}$ denotes the symmetric difference, i.e., $\widehat{\mathcal{R}} \Delta \mathcal{R}^{\star} 
= (\widehat{\mathcal{R}} \setminus \mathcal{R}^{\star}) 
\cup (\mathcal{R}^{\star} \setminus \widehat{\mathcal{R}}).$
\end{lemma}

From \eqref{type_ii_error_eq_1}, there are two terms in the decomposition of the excess Type-II error.
\begin{itemize}
\item[(i)]  The second term quantifies the gap between empirical Type I error and the target level $\alpha$.
By Theorem~\ref{thm_1}, with probability at least $1-\delta$, we have $\left| \mathbb{P}_0(\widehat{\mathcal{R}}^c) - \alpha \right| \leq \Delta_{n_2}(\delta)$ with $\Delta_{n_2}(\delta):=\sqrt{ \frac{\log(2/\delta)}{2n_2}} + \frac{1}{n_2}$, as a result of which, 
\begin{equation}\label{type_ii_error_eq_2}
e^{\tau_{\alpha}^{\star}} \left(\alpha- \mathbb{P}_0(\widehat{\mathcal{R}}^c)\right) 
\leq e^{\tau_{\alpha}^{\star}} \Delta_{n_2}(\delta).
\end{equation}

    \item[(ii)]  The first term on the RHS of~\eqref{type_ii_error_eq_1} is due to the estimation error between the oracle score $\mathcal{S}$ and the estimated $\widehat{\mathcal{S}}$.
    To bound it, we have to bound the discrepancy of the acceptance regions in the score space. By Corollary~\ref{cor:score_deviation}, with probability at least $1-\delta$, we have $\sup_{x\in \mathcal{X}} |\widehat{\mathcal{S}}(x) - \mathcal{S}(x)| \le r(\delta, \rho, \eta)$.

    Since both $\boldsymbol{\mu}_c$ and $f_{\theta,\mathbf{v}}(x)$ lie on the unit hypersphere, the log-scores are deterministically bounded within $[-2\kappa, 2\kappa]$. Therefore, the exponential mapping is Lipschitz continuous with constant $e^{2\kappa}$ over this domain. For any $x \in \widehat{\mathcal{R}} \Delta \mathcal{R}^{\star}$, it follows that
    \begin{equation}
    \label{eq:lip-bound}
    \left| \frac{p_1(x)}{p_0(x)} - e^{\tau_{\alpha}^{\star}} \right| = \left| e^{\mathcal{S}(x)} - e^{\tau_{\alpha}^{\star}} \right| \le e^{2\kappa} \left| \mathcal{S}(x) - \tau_{\alpha}^{\star} \right|.
    \end{equation}

\begin{lemma}
\label{lem:diff-bound}
Suppose
\[
\left|
\mathbb P_0(\widehat{\mathcal R}^{c})-\alpha
\right|
\le
\Delta_{n_2}(\delta),
\qquad
\sup_{x\in\mathcal X}
\left|
\widehat{\mathcal S}(x)-\mathcal S(x)
\right|
\le
r(\delta,\rho,\eta),
\]
and
\[
2r(\delta,\rho,\eta)
+
\left(
\frac{\Delta_{n_2}(\delta)}{c_0}
\right)^{1/\underline{\gamma}}
\le
\varepsilon_0.
\]
Then, under Assumption~\ref{asump3:local_mass},
\begin{equation}
\label{eq:sym-dif-bound}
\widehat{\mathcal R}\Delta\mathcal R^\star
\subseteq
\left\{
x:
\left|\mathcal S(x)-\tau_\alpha^\star\right|
<
2r(\delta,\rho,\eta)
+
\left(
\frac{\Delta_{n_2}(\delta)}{c_0}
\right)^{1/\underline{\gamma}}
\right\}.
\end{equation}
Consequently,
\begin{equation}
\label{eq:P-diff-bound}
\mathbb P_0
\left(
\widehat{\mathcal R}\Delta\mathcal R^\star
\right)
\le
C_0
\left[
2r(\delta,\rho,\eta)
+
\left(
\frac{\Delta_{n_2}(\delta)}{c_0}
\right)^{1/\underline{\gamma}}
\right]^{\bar{\gamma}}.
\end{equation}
\end{lemma}

Now, we are ready to analyze the first term.
\begin{equation}\label{type_ii_error_eq_3}
\begin{aligned}
&\int_{\widehat{\mathcal R}\Delta \mathcal R^\star} \left| \frac{p_1(x)}{p_0(x)} - e^{\tau_{\alpha}^{\star}} \right| \mathrm{d}\mathbb{P}_0(x) 
\\&\overset{\eqref{eq:lip-bound}}{\le} e^{2\kappa} \int_{\widehat{\mathcal R}\Delta \mathcal R^\star} \left| \mathcal{S}(x) - \tau_{\alpha}^{\star} \right| \mathrm{d}\mathbb{P}_0(x) \\
&\overset{\eqref{eq:sym-dif-bound}}{\le} e^{2\kappa} \left[ 2r(\delta, \rho, \eta) + \left( \frac{\Delta_{n_2}(\delta)}{c_0} \right)^{1/\underline{\gamma}} \right] \mathbb{P}_0\left(\widehat{\mathcal R}\Delta \mathcal R^\star\right) \\
&\overset{\eqref{eq:P-diff-bound}}{\le}  e^{2\kappa} {C}_0 \left[ 2r(\delta, \rho, \eta) + \left( \frac{\Delta_{n_2}(\delta)}{{c}_0} \right)^{1/\underline{\gamma}} \right]^{1 + \bar{\gamma}}.
\end{aligned}
\end{equation}

\end{itemize}

With probability at least $1-2\delta$, we have both Theorem~\ref{thm_1} and Corollary~\ref{cor:score_deviation} hold. Conditioned on this joint event, combining \eqref{type_ii_error_eq_2} and \eqref{type_ii_error_eq_3} with the decomposition \eqref{type_ii_error_eq_1} yields
$$
\begin{aligned}
&\mathbb{P}_1(\widehat{\mathcal R}_t)-\mathbb{P}_1(\mathcal R_t^{\star})
\\&\le e^{2\kappa} {C}_0\left[
2r_t(\delta, \rho, \eta)
+
\left(\frac{1}{{c}_0}\left(\sqrt{\frac{\log(2/\delta)} {2n_2}} + \frac{1}{n_2}\right)\right)^{1/\underline{\gamma}}
\right]^{1+\bar{\gamma}}
\! \!+
e^{\tau_{\alpha,t}^{\star}} \left(\sqrt{\frac{\log(2/\delta)} {2n_2}} + \frac{1}{n_2}\right)
\end{aligned}
$$
By applying the basic inequality $\sqrt{a+b+c} \le \sqrt{a}+\sqrt{b}+\sqrt{c}$ and $(x+y)^p \le 2^{(p-1)_+}(x^p+y^p)$ to decouple the terms and using Corollary~\ref{cor:score_deviation}, we abstract the multiplicative constants (including $e^{2\kappa}$, $C_0$, $C_1$, and $e^{\tau_{\alpha,t}^{\star}}$) into the asymptotic notation to obtain the final explicit finite-sample bound
$$
\begin{aligned}
&\mathbb{P}_1(\widehat{\mathcal R}_t)-\mathbb{P}_1(\mathcal R_t^{\star}) 
\le \mathcal{O}\Bigg( \Big((1- {c}\rho)^t+\sqrt{\rho \log (4T / \delta)}+\eta^2 / \rho^2 \Big)^{\frac{1+\bar{\gamma}}{2}} + 
\sqrt{\frac{\log(2/\delta)}{n_2}} + \frac{1}{n_2} \Bigg).
\end{aligned}
$$
When $\Delta_{n_2}(\delta)\le 1$, the condition $\underline{\gamma}\le 1+\bar{\gamma}$ implies $\Delta_{n_2}(\delta)^{(1+\bar{\gamma})/\underline{\gamma}} \le \Delta_{n_2}(\delta)$. When $\Delta_{n_2}(\delta)>1$, the claimed bound is trivial since the excess Type-II error is at most one. This completes the proof.
\end{proof}

\subsection{Proof of Theorem \ref{thm:local_tracking}}
\label{proof:local_tracking}
\begin{proof}[Proof of Theorem \ref{thm:local_tracking}]
Fix $c\in\{0,1\}$ and suppress the class index throughout the proof. Write
\[
\boldsymbol{\mu}_t:=\boldsymbol{\mu}_{c,t},\qquad
\widehat{\boldsymbol{\mu}}_t:=\widehat{\boldsymbol{\mu}}_{c,t},\qquad
\bar{\mathbf z}_t:=\bar{\mathbf z}_{c,t}(\mathbf v_t),\qquad
A:=A_d(\kappa).
\]
By Assumption~\ref{asump1:true-model},
\[
\bar{\mathbf z}_t = A\boldsymbol{\mu}_t+\boldsymbol{\zeta}_{t+1},
\qquad
\mathbb E[\boldsymbol{\zeta}_{t+1}\mid\mathcal F_t]=0,
\qquad
\|\boldsymbol{\zeta}_{t+1}\|\le 2
\quad\text{a.s.}
\]
Hence
\[
\mathbf z_{t+1}
=
(1-\rho)\widehat{\boldsymbol{\mu}}_t
+
\rho A\boldsymbol{\mu}_t
+
\rho \boldsymbol{\zeta}_{t+1},
\qquad
\widehat{\boldsymbol{\mu}}_{t+1}
=
\frac{\mathbf z_{t+1}}{\|\mathbf z_{t+1}\|}.
\]
Define
\[
\mathbf e_t:=\widehat{\boldsymbol{\mu}}_t-\boldsymbol{\mu}_t,
\qquad
\mathbf y_t:=(1-\rho)\widehat{\boldsymbol{\mu}}_t+\rho A\boldsymbol{\mu}_t.
\]

\paragraph{Step 1: Bootstrap region and deterministic bounds.}

Define the bootstrap event
\[
\mathcal E_t := \left\{\|\mathbf e_{s}\|\le \frac14 \text{ for all } 0\le s \le t\right\}.
\]
Clearly, $\mathcal E_t \in \mathcal{F}_t := \sigma(\{\boldsymbol{\zeta}_{s}\}_{s=1}^t)$ is $\mathcal{F}_t$-measurable.
On $\mathcal E_t$, Lemma~\ref{lem:local_contraction} yields
\[
\|\Pi(\mathbf y_t)-\boldsymbol{\mu}_t\|
\le
(1-c_0\rho)\|\mathbf e_t\|,
\]
for some $c_0 \in (0, 1)$ and all $0\le t\le T-1$, provided $\rho\le \rho_0$ is small enough.

Also, on $\mathcal E_t$, $\langle \widehat{\boldsymbol{\mu}}_t,\boldsymbol{\mu}_t\rangle
=
1-\frac12\|\mathbf e_t\|^2
\ge
1-\frac1{32}
=
\frac{31}{32}.$
Therefore,
\[
\|\mathbf y_t\|
\ge
\langle \mathbf y_t,\boldsymbol{\mu}_t\rangle
=
(1-\rho)\langle \widehat{\boldsymbol{\mu}}_t,\boldsymbol{\mu}_t\rangle+\rho A
\ge
(1-\rho)\frac{31}{32}+\rho A.
\]
Shrinking $\rho_0$ if necessary, we may ensure $\|\mathbf y_t\|\ge \frac12$, $\forall\,0<\rho\le \rho_0.$
Since $\mathbf z_{t+1}=\mathbf y_t+\rho\boldsymbol{\zeta}_{t+1}$ and $\|\boldsymbol{\zeta}_{t+1}\|\le 2$, we further obtain $\|\mathbf z_{t+1}\|
\ge
\|\mathbf y_t\|-\rho\|\boldsymbol{\zeta}_{t+1}\|
\ge
\frac12-2\rho
\ge
\frac14$
for $\rho_0$ smaller if needed.
In short, we show that as long as $\|\mathbf{e}_t\| \le \frac{1}{4}$, we then have $\|\mathbf y_{t}\| \ge \frac{1}{2}$, $\|\mathbf z_{t+1}\| \ge \frac{1}{4}$, and $\langle \widehat{\boldsymbol{\mu}}_t,\boldsymbol{\mu}_t\rangle \ge 0$.

\paragraph{Step 2: One-step square-error decomposition.}
For any $\gamma>0$, Young's inequality gives
\begin{equation}
\label{eq:hp_sq_drift_split}
\|\mathbf e_{t+1}\|^2
=
\|\widehat{\boldsymbol{\mu}}_{t+1}-\boldsymbol{\mu}_{t+1}\|^2
\le
(1+\gamma)\|\widehat{\boldsymbol{\mu}}_{t+1}-\boldsymbol{\mu}_t\|^2
+
\left(1+\frac1\gamma\right)\|\boldsymbol{\mu}_{t+1}-\boldsymbol{\mu}_t\|^2.
\end{equation}
By Assumption~\ref{asump2:local_tracking}, $\|\boldsymbol{\mu}_{t+1}-\boldsymbol{\mu}_t\|\le C_\mu \eta.$
We now analyze the first term of \eqref{eq:hp_sq_drift_split}:
\[
\|\widehat{\boldsymbol{\mu}}_{t+1}-\boldsymbol{\mu}_t\|^2
=
\|\Pi(\mathbf z_{t+1})-\boldsymbol{\mu}_t\|^2.
\]
For simplicity, we set $\mathbf a_t := \Pi(\mathbf y_t)-\boldsymbol{\mu}_t.$
Then
\[
\Pi(\mathbf z_{t+1})-\boldsymbol{\mu}_t
=
\mathbf a_t + \bigl(\Pi(\mathbf y_t+\rho\boldsymbol{\zeta}_{t+1})-\Pi(\mathbf y_t)\bigr).
\]
Hence, it follows that
\begin{align}
\label{eq:hp_sq_expand}
\|\Pi(\mathbf z_{t+1})-\boldsymbol{\mu}_t\|^2
&=
\|\mathbf a_t\|^2
+
2\left\langle
\mathbf a_t,\,
\Pi(\mathbf y_t+\rho\boldsymbol{\zeta}_{t+1})-\Pi(\mathbf y_t)
\right\rangle \nonumber\\
&\quad+
\left\|
\Pi(\mathbf y_t+\rho\boldsymbol{\zeta}_{t+1})-\Pi(\mathbf y_t)
\right\|^2.
\end{align}

\paragraph{Step 3: Taylor expansion of the projection map.}
Since $\|\mathbf y_t\|\ge 1/2$ on $\mathcal E_t$, the projection map
$\Pi(\mathbf x)=\mathbf x/\|\mathbf x\|$ is $C^2$ on a neighborhood of the trajectory. Therefore,
\[
\Pi(\mathbf y_t+\rho\boldsymbol{\zeta}_{t+1})-\Pi(\mathbf y_t)
=
\rho J_t\boldsymbol{\zeta}_{t+1}+\mathbf R_{t+1},
\]
where $J_t$ is the Jacobian matrix
\[
J_t
:=
D\Pi(\mathbf y_t)
=
\frac{1}{\|\mathbf y_t\|}
\Bigl(
I-\Pi(\mathbf y_t)\Pi(\mathbf y_t)^\top
\Bigr),
\]
and the remainder satisfies
\[
\|\mathbf R_{t+1}\|
\le
C_R\rho^2\|\boldsymbol{\zeta}_{t+1}\|^2
\le
4C_R\rho^2
\]
for a universal constant $C_R>0$. Moreover, since $\|\mathbf y_t\|\ge 1/2$, we have $\|J_t\|\le 2$. Substituting this expansion into \eqref{eq:hp_sq_expand}, we obtain
\begin{align}
\label{eq:hp_sq_expand_2}
\|\Pi(\mathbf z_{t+1})-\boldsymbol{\mu}_t\|^2
&=
\|\mathbf a_t\|^2
+
2\rho
\langle
\mathbf a_t,
J_t\boldsymbol{\zeta}_{t+1}
\rangle
+
\Xi_{t+1},
\end{align}
where $\Xi_{t+1}$ collects the higher-order terms:
\[
\Xi_{t+1}
:=
2\langle
\mathbf a_t,\mathbf R_{t+1}
\rangle
+
\|
\rho J_t\boldsymbol{\zeta}_{t+1}
+
\mathbf R_{t+1}
\|^2.
\]
Since $\|\mathbf a_t\|\le 2$, $\|J_t\|\le 2$,
$\|\boldsymbol{\zeta}_{t+1}\|\le 2$, and
$\|\mathbf R_{t+1}\|\le 4C_R\rho^2$, there exists a universal constant
$C_1>0$ such that
\[
|\Xi_{t+1}|
\le
C_1\rho^2.
\]

Let $\mathbbm{1}_{\mathcal E_t}$ denote the indicator of $\mathcal E_t$ and define
\[
\Delta_{t+1}
:=
2\rho
\langle
\mathbf a_t,
J_t\boldsymbol{\zeta}_{t+1}
\rangle
\mathbbm{1}_{\mathcal E_t}.
\]
Since $\mathbf a_t$, $J_t$, and $\mathbbm{1}_{\mathcal E_t}$ are
$\mathcal F_t$-measurable and
$\mathbb E[\boldsymbol{\zeta}_{t+1}\mid\mathcal F_t]=0$,
$\{\Delta_{t+1}\}$ is a martingale difference sequence with respect to
$\{\mathcal F_t\}$. Moreover, there exist constants $C_2,C_3>0$ such that,
almost surely,
\[
|\Delta_{t+1}|
\le
C_2\rho,
\qquad
\mathbb E[
\Delta_{t+1}^2
\mid
\mathcal F_t
]
\le
C_3\rho^2.
\]
Hence, by \eqref{eq:hp_sq_expand_2}, on $\mathcal E_t$,
\[
\|\Pi(\mathbf z_{t+1})-\boldsymbol{\mu}_t\|^2
\le
\|\mathbf a_t\|^2
+
\Delta_{t+1}
+
C_1\rho^2.
\]
\paragraph{Step 4: One-step recursion on the bootstrap region.}
By Lemma~\ref{lem:local_contraction},
\[
\|\mathbf a_t\|^2
\le
(1-c_0\rho)^2\|\mathbf e_t\|^2
\le
(1-c_0\rho)\|\mathbf e_t\|^2.
\]
Combining this with \eqref{eq:hp_sq_drift_split}, we obtain that on $\mathcal E_{t-1}$,
\[
\|\mathbf e_{t+1}\|^2
\le
(1+\gamma)(1-c_0\rho)\|\mathbf e_t\|^2
+
(1+\gamma)C_1\rho^2
+
\left(1+\frac1\gamma\right)C_\mu^2\eta^2
+
(1+\gamma)\Delta_{t+1}.
\]
We then choose $\gamma := \frac{c_0\rho}{4}$.
As a result, for $\rho_0$ sufficiently small, we could simultaneously have that (i) $(1+\gamma)(1-c_0\rho)\le 1-c_1\rho$
for some constant $c_1>0$, (ii) $(1+\gamma)C_1\rho^2\le C_4\rho^2$, and (iii) $\left(1+\frac1\gamma\right)C_\mu^2\eta^2 \le C_5\frac{\eta^2}{\rho}.$
As a result of the above notation simplification, we may rewrite the recursion as
\begin{equation}
\label{eq:hp_sq_recursion}
\|\mathbf e_{t+1}\|^2
\le
(1-c_1\rho)\|\mathbf e_t\|^2
+
C_4\rho^2
+
C_5\frac{\eta^2}{\rho}
+
\left(1+\frac{c_0\rho}{4}\right)\Delta_{t+1},
\end{equation}
where $\{\Delta_{t+1}\}$ is a martingale difference sequence satisfying $|\Delta_{t+1}|\le  C_2 \rho$ and $\mathbb E[\Delta_{t+1}^2\mid\mathcal F_t]\le  C_3 \rho^2$.

\paragraph{Step 5: Unrolling the recursion.}
Iterating \eqref{eq:hp_sq_recursion} yields that, on $\mathcal E_{t-1}$,
\[
\|\mathbf e_t\|^2
\le
(1-c_1\rho)^t\|\mathbf e_0\|^2
+
\left(C_4\rho^2+C_5\frac{\eta^2}{\rho}\right)\sum_{j=0}^{t-1}(1-c_1\rho)^j
+
\left(1+\frac{c_0\rho}{4}\right) \sum_{s=0}^{t-1}(1-c_1\rho)^{t-1-s}\Delta_{s+1}.
\]
Since $\sum_{j=0}^{t-1}(1-c_1\rho)^j \le \frac{1}{c_1\rho}$, we obtain that on $\mathcal E_{t-1}$,
\begin{equation}
\label{eq:hp_sq_unrolled}
\|\mathbf e_t\|^2
\le
(1-c_1\rho)^t\|\mathbf e_0\|^2
+
C_6\rho
+
C_7\frac{\eta^2}{\rho^2}
+
\left(1+\frac{c_0\rho}{4}\right) \sum_{s=0}^{t-1}(1-c_1\rho)^{t-1-s} \Delta_{s+1}.
\end{equation}

Now, we apply the concentration inequality in Lemma~\ref{lem:weighted_mds} to the martingale sequence $\{\Delta_{t+1}\}$ (with $c_\star=c_1$). There exists a constant $C_{8}>0$ such that, with probability at least $1-\delta/2$,
\begin{equation}
\label{eq:concentration-event}
\left(1+\frac{c_0\rho}{4}\right) \sup_{1\le t\le T}
\left|
\sum_{s=0}^{t-1}(1-c_1\rho)^{t-1-s}\Delta_{s+1}
\right|
\le
C_{8}\left(
\sqrt{\rho\log\frac{2T}{\delta}}
+
\rho\log\frac{2T}{\delta}
\right).
\end{equation}
We denote the event in \eqref{eq:concentration-event} as $\mathcal{E}_{\#}$.
Then $\mathbb{P}(\mathcal{E}_{\#}) \ge 1-\delta$.
Since $\rho\le 1$ and $\rho\log(2T/\delta) \le 1$, after enlarging constants we can bound the right-hand side by $C_{9} \sqrt{\rho\log\frac{2T}{\delta}}.$

Combining this with \eqref{eq:hp_sq_unrolled}, we conclude that on the event $\mathcal{E}_{\#} \cap \mathcal{E}_{t-1}$, for suitable constants ${c},C>0$,
\begin{equation}
\label{eq:iterate-error}
\|\mathbf e_t\|^2
\le
(1-{c}\rho)^t\|\mathbf e_0\|^2
+
C\left(
\sqrt{\rho\log\frac{2T}{\delta}}
+
\frac{\eta^2}{\rho^2}
\right).
\end{equation}

\paragraph{Step 6: Induction on Bootstrap events.}
By the condition in Theorem~\ref{thm:local_tracking},
\begin{equation}
\label{eq:initial-condition}
\|\mathbf e_0\|^2
+
C\left(
\sqrt{\rho\log\frac{2T}{\delta}}
+
\frac{\eta^2}{\rho^2}
\right)
\le \frac1{16}.
\end{equation}
As a result of this condition, $\mathcal{E}_0 = \{\|\mathbf e_0\| \le \frac{1}{4}\}$ is true.
Now, we want to show that for any $t \ge 1$, on the event $\mathcal{E}_{\#} \cap \mathcal{E}_{t-1}$, we must have that the event $\mathcal{E}_{t}$ is also true.
It suffices to show that 
\[\|\mathbf e_t\|^2 \le \frac{1}{16},\]
which naturally follows by combining \eqref{eq:initial-condition} and \eqref{eq:iterate-error}.
We then complete the proof by induction.
\end{proof}

\subsection{Proof of Lemma \ref{lem:error-decomposition}}
\label{proof:error-decomposition}
\begin{proof}[Proof of Lemma \ref{lem:error-decomposition}]
We split this region $\widehat{\mathcal{R}} \Delta \mathcal{R}^{\star}$ into two disjoint parts: $\mathcal{R}^{\star} \cap \widehat{\mathcal{R}}^c$ (where the optimal classifier accepts but the estimated one rejects) and $(\mathcal{R}^{\star})^c \cap \widehat{\mathcal{R}}$ (where the estimated one accepts but the optimal rejects).
It then follows that
$$
\begin{aligned}
&\int_{(\mathcal{R}^{\star} \cap \widehat{\mathcal{R}}^c) \cup ((\mathcal{R}^{\star})^c \cap \widehat{\mathcal{R}})}\left|\frac{p_1(x)}{p_0(x)}-e^{\tau_{\alpha}^{\star}}\right| \mathrm{d} \mathbb{P}_0(x) 
\\& =\int_{\mathcal{R}^{\star} \cap \widehat{\mathcal{R}}^c}\left|\frac{p_1(x)}{p_0(x)} - e^{\tau_{\alpha}^{\star}}\right| \mathrm{d} \mathbb{P}_0(x)+\int_{(\mathcal{R}^{\star})^c \cap \widehat{\mathcal{R}}}\left|\frac{p_1(x)}{p_0(x)}-e^{\tau_{\alpha}^{\star}}\right| \mathrm{d} \mathbb{P}_0(x) \\
& =\int_{\mathcal{R}^{\star} \cap \widehat{\mathcal{R}}^c}\left(e^{\tau_{\alpha}^{\star}}-\frac{p_1(x)}{p_0(x)}\right) \mathrm{d} \mathbb{P}_0(x)+\int_{(\mathcal{R}^{\star})^c \cap \widehat{\mathcal{R}}}\left(\frac{p_1(x)}{p_0(x)}-e^{\tau_{\alpha}^{\star}}\right) \mathrm{d} \mathbb{P}_0(x).
\end{aligned}
$$
Note that
$$
\begin{aligned}
&\int_{\mathcal{R}^{\star}}\left(e^{\tau_{\alpha}^{\star}} - \frac{p_1(x)}{p_0(x)}\right) \mathrm{d} \mathbb{P}_0(x)+\int_{\widehat{\mathcal{R}}}\left(\frac{p_1(x)}{p_0(x)}-e^{\tau_{\alpha}^{\star}}\right) \mathrm{d} \mathbb{P}_0(x)
\\&=
\underbrace{\int_{\mathcal{R}^{\star} \cap \widehat{\mathcal{R}}}\left(e^{\tau_{\alpha}^{\star}}-\frac{p_1(x)}{p_0(x)} + \frac{p_1(x)}{p_0(x)}-e^{\tau_{\alpha}^{\star}}\right) \mathrm{d} \mathbb{P}_0(x)}_{=0}
\\&\quad +\int_{\mathcal{R}^{\star} \cap \widehat{\mathcal{R}}^c}\left(e^{\tau_{\alpha}^{\star}}- \frac{p_1(x)}{p_0(x)} \right) \mathrm{d} \mathbb{P}_0(x)+\int_{(\mathcal{R}^{\star})^c \cap \widehat{\mathcal{R}}}\left(\frac{p_1(x)}{p_0(x)} - e^{\tau_{\alpha}^{\star}}\right) \mathrm{d} \mathbb{P}_0(x).
\end{aligned}
$$
Using the property that $\int_A \frac{p_1(x)}{p_0(x)} \mathrm{d} \mathbb{P}_0(x)=\int_A \mathrm{d} \mathbb{P}_1(x)=\mathbb{P}_1(A)$ for any measurable set $A$, we have
$$
\begin{aligned}
&\int_{(\mathcal{R}^{\star} \cap \widehat{\mathcal{R}}^c) \cup ((\mathcal{R}^{\star})^c \cap \widehat{\mathcal{R}})}\left\lvert\frac{p_1(x)}{p_0(x)}-e^{\tau_{\alpha}^{\star}}\right\rvert \mathrm{d} \mathbb{P}_0(x) \\
& =\int_{\mathcal{R}^{\star}}\left(e^{\tau_{\alpha}^{\star}} - \frac{p_1(x)}{p_0(x)}\right) \mathrm{d} \mathbb{P}_0(x)+\int_{\widehat{\mathcal{R}}}\left(\frac{p_1(x)}{p_0(x)}-e^{\tau_{\alpha}^{\star}}\right) \mathrm{d} \mathbb{P}_0(x)\\
& = \mathbb{P}_1(\widehat{\mathcal{R}}) - \mathbb{P}_1(\mathcal{R}^{\star})  + e^{\tau_{\alpha}^{\star}} \left(\mathbb{P}_0(\mathcal{R}^{\star})- \mathbb{P}_0(\widehat{\mathcal{R}})\right),
\end{aligned}
$$
which, together with the fact that $\mathbb{P}_0((\mathcal{R}^{\star})^c) = \alpha$, implies that $\mathbb{P}_0(\mathcal{R}^{\star})=1-\alpha$. Recalling $\mathbb{P}_0(\widehat{\mathcal{R}})=1-\mathbb{P}_0(\widehat{\mathcal{R}}^c)$, we obtain $\mathbb{P}_0(\mathcal{R}^{\star})- \mathbb{P}_0(\widehat{\mathcal{R}}) = \mathbb{P}_0(\widehat{\mathcal{R}}^c) - \alpha$. 
We then complete the proof.
\end{proof}

\subsection{Proof of Lemma \ref{lem:diff-bound}}
\label{proof:diff-bound}
\begin{proof}[Proof of Lemma \ref{lem:diff-bound}]
By the uniform error bound, we have
\[
\{\mathcal S(X)> \widehat{\tau}_\alpha+r(\delta,\rho,\eta)\}
\subseteq
\{\widehat{\mathcal S}(X)> \widehat{\tau}_\alpha\}
\subseteq
\{\mathcal S(X)> \widehat{\tau}_\alpha-r(\delta,\rho,\eta)\},
\]
which implies
\[
\mathbb P_0\!\left(\mathcal S(X)> \widehat{\tau}_\alpha+r(\delta,\rho,\eta)\right)
\le
\mathbb P_0(\widehat{\mathcal R}^c)
\le
\mathbb P_0\!\left(\mathcal S(X)> \widehat{\tau}_\alpha-r(\delta,\rho,\eta)\right).
\]
Since
$\alpha=\mathbb P_0\left(\mathcal S(X)> \tau_\alpha^\star\right)$, it follows that $\mathbb P_0\!\left(
\mathcal S(X)\in
[\tau_\alpha^\star,\widehat{\tau}_\alpha-r(\delta,\rho,\eta))
\right)
\le
\Delta_{n_2}(\delta)$ whenever $\widehat{\tau}_\alpha-r(\delta,\rho,\eta)\ge \tau_\alpha^\star$, and similarly $\mathbb P_0\!\left(\mathcal S(X)\in [\widehat{\tau}_\alpha+r(\delta,\rho,\eta),\tau_\alpha^\star) \right) \le \Delta_{n_2}(\delta)$ whenever $\widehat{\tau}_\alpha+r(\delta,\rho,\eta)\le \tau_\alpha^\star$.

Now let
\[
u_\delta
:=
\left(
\frac{\Delta_{n_2}(\delta)}{c_0}
\right)^{1/\underline{\gamma}}.
\]
Assume $u_\delta<\varepsilon_0$.
By the one-sided lower bounds in
Assumption~\ref{asump3:local_mass}, for every
$\varepsilon\in[0,\varepsilon_0]$,
\[
\mathbb P_0\!\left(
\tau_\alpha^\star
<
\mathcal S(X)
\le
\tau_\alpha^\star+\varepsilon
\right)
\ge
c_0\varepsilon^{\underline{\gamma}}
\quad \text{and} \quad 
\mathbb P_0\!\left(
\tau_\alpha^\star-\varepsilon
\le
\mathcal S(X)
<
\tau_\alpha^\star
\right)
\ge
c_0\varepsilon^{\underline{\gamma}}.
\]

Suppose first that
$\widehat{\tau}_\alpha-r(\delta,\rho,\eta)
>
\tau_\alpha^\star+u_\delta$.
Then one can choose
$\varepsilon\in(u_\delta,\varepsilon_0]$
such that
\[
\tau_\alpha^\star+\varepsilon
<
\widehat{\tau}_\alpha-r(\delta,\rho,\eta).
\]
The right-sided lower bound therefore gives
\[
\mathbb P_0\!\left(
\tau_\alpha^\star
<
\mathcal S(X)
\le
\tau_\alpha^\star+\varepsilon
\right)
>
\Delta_{n_2}(\delta),
\]
contradicting the preceding upper bound on the probability mass between $\tau_\alpha^\star$ and $\widehat{\tau}_\alpha-r(\delta,\rho,\eta)$. The case $\widehat{\tau}_\alpha+r(\delta,\rho,\eta)<\tau_\alpha^\star-u_\delta$ is analogous, using the left-sided lower bound. Hence,
\[
|\widehat{\tau}_\alpha-\tau_\alpha^\star| \le r(\delta,\rho,\eta)+u_\delta.
\]

Next, consider any $x\in \widehat{\mathcal R}\Delta \mathcal R^\star$. 
\begin{itemize}
    \item[(i)] If $x\in (\mathcal R^\star)^c\cap \widehat{\mathcal R}$, then $\mathcal S(x)> \tau_\alpha^\star$ and $\widehat{\mathcal S}(x)\le\widehat{\tau}_\alpha$. Using the uniform score bound gives
    \[
    \mathcal S(x) \le\widehat{\tau}_\alpha+r(\delta,\rho,\eta)
    \le\tau_\alpha^\star+2r(\delta,\rho,\eta) +\left(\frac{\Delta_{n_2}(\delta)}{c_0}\right)^{1/\underline{\gamma}}.
    \]
    \item[(ii)] If $x\in \mathcal R^\star\cap \widehat{\mathcal R}^c$, then $\mathcal S(x)\le \tau_\alpha^\star$ and $\widehat{\mathcal S}(x)> \widehat{\tau}_\alpha$, and hence
    \[
    \tau_\alpha^\star-2r(\delta,\rho,\eta) -\left(\frac{\Delta_{n_2}(\delta)}{c_0}\right)^{1/\underline{\gamma}}
    <\mathcal S(x)
    \le\tau_\alpha^\star.
    \]

\end{itemize}
Combining the two cases, we complete the proof.
\end{proof}

\subsection{Useful Lemmas}

\begin{lemma}\label{lem:proj_lipschitz}
Let $m>0$. For any $\mathbf x,\mathbf y\in\mathbb R^d$ satisfying
$\|\mathbf x\|\ge m$ and $\|\mathbf y\|\ge m$,
\[
\|\Pi(\mathbf x)-\Pi(\mathbf y)\|
\le \frac{2}{m}\|\mathbf x-\mathbf y\|.
\]
\end{lemma}

\begin{proof}[Proof of Lemma \ref{lem:proj_lipschitz}]
Using $\Pi(\mathbf x)-\Pi(\mathbf y)
=
\frac{\mathbf x-\mathbf y}{\|\mathbf x\|}
+
\mathbf y\!\left(\frac{1}{\|\mathbf x\|}-\frac{1}{\|\mathbf y\|}\right),$ we obtain
\begin{align*}
\|\Pi(\mathbf x)-\Pi(\mathbf y)\|
&\le \frac{\|\mathbf x-\mathbf y\|}{\|\mathbf x\|}
 + \frac{\bigl|\|\mathbf y\|-\|\mathbf x\|\bigr|}{\|\mathbf x\|}\le \frac{2}{m}\|\mathbf x-\mathbf y\|.
\end{align*}
\end{proof}

\begin{lemma}\label{lem:local_contraction}
There exist constants $c_0, \rho_0\in(0,1)$, depending only on $A$, such that for every
$0<\rho\le \rho_0$, every $t\ge 0$, and every $\mathbf u\in \mathbb S^{d-1}$ satisfying
$\|\mathbf u-\boldsymbol{\mu}_t\|\le \frac14$,
\[
\left\|
\Pi\!\bigl((1-\rho)\mathbf u+\rho A\boldsymbol{\mu}_t\bigr)-\boldsymbol{\mu}_t
\right\|
\le
(1-c_0\rho)\|\mathbf u-\boldsymbol{\mu}_t\|.
\]
\end{lemma}

\begin{proof}[Proof of Lemma \ref{lem:local_contraction}]
Fix $t$ and $\mathbf u\in\mathbb S^{d-1}$ with $\|\mathbf u-\boldsymbol{\mu}_t\|\le 1/4$.
Write
\[
\mathbf u = \alpha \boldsymbol{\mu}_t + \beta \boldsymbol{\nu},
\qquad
\boldsymbol{\nu}\perp \boldsymbol{\mu}_t,\qquad
\alpha^2+\beta^2=1.
\]
Since $\|\mathbf u-\boldsymbol{\mu}_t\|^2 = 2(1-\alpha)$, the condition
$\|\mathbf u-\boldsymbol{\mu}_t\|\le 1/4$ implies $\alpha \ge 1-\frac{1}{32} = \frac{31}{32}.$
Define
\[
\mathbf x := (1-\rho)\mathbf u + \rho A\boldsymbol{\mu}_t
          = a \boldsymbol{\mu}_t + b \boldsymbol{\nu},
\]
where $a := (1-\rho)\alpha+\rho A$ and $b := (1-\rho)\beta.$
Then
\[
\Pi(\mathbf x)=\frac{a\boldsymbol{\mu}_t+b\boldsymbol{\nu}}{r}
\quad \text{with} \quad
r:=\sqrt{a^2+b^2}.
\]
Therefore, we have that
\[
\left\|\Pi(\mathbf x)-\boldsymbol{\mu}_t\right\|^2
=
2\left(1-\frac{a}{r}\right)
=
\frac{2b^2}{r(r+a)}.
\]
Using the facts that $\|\mathbf u-\boldsymbol{\mu}_t\|^2 = 2(1-\alpha)$ and $\beta^2 = 1-\alpha^2,$ we get
\[
\frac{\left\|\Pi(\mathbf x)-\boldsymbol{\mu}_t\right\|^2}
{\|\mathbf u-\boldsymbol{\mu}_t\|^2}
=
\frac{(1-\rho)^2(1+\alpha)}{r(r+a)}
=: \Psi(\alpha,\rho).
\]
Now $\alpha\in[31/32,1]$, and one checks that $r^2=(1-\rho)^2+2\rho(1-\rho)A\alpha+\rho^2A^2
$ by definition.
Hence, $\Psi$ is continuous on the compact set $[31/32,1]\times[0,\rho_1]$ for any fixed $\rho_1<1$, and $\Psi(\alpha,0)=1$ for all $\alpha$.
Moreover,
\[
\left.\frac{\partial}{\partial \rho}\Psi(\alpha,\rho)\right|_{\rho=0}
= -A(1+\alpha)
\le -A\cdot \frac{63}{32}
<0,
\]
uniformly over $\alpha\in[31/32,1]$.
Therefore, by compactness and continuity, there exist constants $c_0, \rho_0\in(0,1)$,
depending only on $A$, such that for all $\alpha\in[31/32,1]$ and all $0<\rho\le \rho_0$,
\[
\Psi(\alpha,\rho)\le (1-c_0\rho)^2.
\]
Taking square roots yields that
\[
\left\|
\Pi\!\bigl((1-\rho)\mathbf u+\rho A\boldsymbol{\mu}_t\bigr)-\boldsymbol{\mu}_t
\right\|
\le
(1-c_0\rho)\|\mathbf u-\boldsymbol{\mu}_t\|.
\]
\end{proof}

\begin{lemma}[Weighted martingale concentration]\label{lem:weighted_mds}
Let $\{\Delta_{t+1}\}_{t\ge 0}$ be a martingale difference sequence with respect to
$\{\mathcal F_t\}$.
Assume that there exist constants $B,V>0$ such that, almost surely,
\[
|\Delta_{t+1}| \le B \rho,
\qquad
\mathbb E[\Delta_{t+1}^2\mid \mathcal F_t]\le V \rho^2
\qquad \forall\, t\ge 0.
\]
Fix $c_\star, \rho \in (0, 1)$ and define, for each $t\ge 1$,
\[
S_t := \sum_{s=0}^{t-1}(1-c_\star \rho)^{t-1-s}\Delta_{s+1}.
\]
Then there exists a constant $C>0$, depending only on $B,V,c_\star$, such that for any
$\delta\in(0,1)$,
\[
\mathbb P\!\left(
\sup_{1\le t\le T}|S_t|
\le
C\Bigl(\sqrt{\rho\log\frac{T}{\delta}}+\rho\log\frac{T}{\delta}\Bigr)
\right)
\ge 1-\delta.
\]
\end{lemma}

\begin{proof}[Proof of Lemma \ref{lem:weighted_mds}]
Fix $t\ge 1$. Define
\[
\xi_{t,s+1}:=(1-c_\star\rho)^{t-1-s}\Delta_{s+1},
\qquad 0\le s\le t-1.
\]
Then $\{\xi_{t,s+1}\}_{s=0}^{t-1}$ is a martingale difference sequence with respect to
$\{\mathcal F_s\}$, and
\[
|\xi_{t,s+1}|
\le
(1-c_\star\rho)^{t-1-s}B\rho.
\]
Moreover,
\[
\sum_{s=0}^{t-1}\mathbb E[\xi_{t,s+1}^2\mid \mathcal F_s]
\le
V \rho^2 \sum_{s=0}^{t-1}(1-c_\star\rho)^{2(t-1-s)}
\le
V \rho^2 \sum_{j=0}^{\infty}(1-c_\star\rho)^{2j}.
\]
Since
\[
\sum_{j=0}^{\infty}(1-c_\star\rho)^{2j}
=
\frac{1}{1-(1-c_\star\rho)^2}
\le
\frac{1}{c_\star\rho},
\]
it follows that
\[
\sum_{s=0}^{t-1}\mathbb E[\xi_{t,s+1}^2\mid \mathcal F_s]
\le
\frac{V}{c_\star}\rho.
\]
Applying Freedman's inequality \citep{freedman1975tail} yields that, for any $u>0$,
\[
\mathbb P\!\left(
|S_t|\ge u
\right)
\le
2\exp\!\left(
-\frac{u^2}{2\left(\frac{V}{c_\star}\rho+\frac13 B\rho\,u\right)}
\right).
\]
Taking a union bound over $1\le t\le T$ and choosing
\[
u
=
C\Bigl(\sqrt{\rho\log\frac{T}{\delta}}+\rho\log\frac{T}{\delta}\Bigr).
\]
for a sufficiently large constant $C$ proves the claim.
\end{proof}

\section{Additional Theory}\label{app_prop}
\subsection{Gradient Analysis of the Population Training Objective}\label{app_prop_1}

The training procedure in Phase I follows a two-timescale scheme: the steering vector $\mathbf v$ is updated on a slow timescale, while the class mean directions are updated on a fast timescale. To make the role of this dynamics clear, we analyze the corresponding population objective.

For a fixed steering vector $\mathbf v$, we then define
\[
\mathcal{J}(\mathbf v)
:=
\max_{\boldsymbol{\mu}_0,\boldsymbol{\mu}_1 \in \mathbb{S}^{d-1}}
\mathbb{E}_{x,y}\!\left[
\log q_\beta\!\left(
y \mid f_{\theta,\mathbf v}(x),
\boldsymbol{\mu}_0,
\boldsymbol{\mu}_1
\right)
\right]
=
\mathbb{E}_{x,y}\!\left[
\log q_\beta\!\left(
y \mid f_{\theta,\mathbf v}(x),
\boldsymbol{\mu}_0(\mathbf v),
\boldsymbol{\mu}_1(\mathbf v)
\right)
\right].
\]
where $\boldsymbol{\mu}_c(\mathbf v)$ denotes the population-optimal mean direction for class $c$ induced by the steered representation map $f_{\theta,\mathbf v}$. 

The following proposition characterizes how optimizing $\mathcal J(\mathbf v)$ reshapes the latent geometry.  For this analysis, we work under Assumption~\ref{asump1:true-model}, a uniform class prior $p(y=0)=p(y=1)=1/2$, and the matched case $\beta=\kappa$. Under these conditions, the temperature-scaled posterior $q_\beta$ in~\eqref{eq:postvMF} coincides with the class posterior induced by the shared-concentration vMF model.

\begin{proposition}\label{prop_1}
The gradient of $\mathcal{J}(\mathbf v)$ admits the decomposition
\[
\nabla_{\mathbf v}\mathcal J(\mathbf v)
=
\beta\,\bar{\omega}(\mathbf v)\,\nabla_{\mathbf{v}} 
\|\boldsymbol{\mu}_1(\mathbf{v})-\boldsymbol{\mu}_0(\mathbf{v})\|^2
+
\frac{\beta}{2}
\bigl(\boldsymbol{\mu}_1(\mathbf v)-\boldsymbol{\mu}_0(\mathbf v)\bigr)^\top
\bigl(
\Delta_1(\mathbf v)-\Delta_0(\mathbf v)
\bigr).
\]
Here, $\Delta_0(\mathbf v)$ and $\Delta_1(\mathbf v)$ are residual terms defined in~\eqref{residual_term} as
\begin{equation}\label{residual_term}
\Delta_c(\mathbf v):=
\mathbb E_x\!\left[
q_\beta(1-c\mid x)
\bigl(
\nabla_{\mathbf v} f_{\theta,\mathbf v}(x)
-
\nabla_{\mathbf v}\boldsymbol{\mu}_c(\mathbf v)
\bigr)
\mid y=c
\right],
\end{equation}
and $\bar{\omega}(\mathbf{v}) \geq 0$ is a data-dependent weight given by
\[
\bar{\omega}(\mathbf v) = \frac{1}{8}
\left(1-\mathbb E_x\!\left[
\tanh^2\!\left(\frac{\beta}{2}
\bigl(\boldsymbol{\mu}_1(\mathbf v)-\boldsymbol{\mu}_0(\mathbf v)\bigr)^\top
f_{\theta,\mathbf v}(x)\right)\right]\right).
\]
\end{proposition}

Proposition~\ref{prop_1} decomposes the gradient of the population objective into two components. 
The first term drives an increase in the separation between the class mean directions 
$\|\boldsymbol{\mu}_1(\mathbf v)-\boldsymbol{\mu}_0(\mathbf v)\|^2$, with a data-dependent weight 
that emphasizes samples near the decision boundary. 
The second term captures a residual effect arising from the mismatch between the sample-level 
response $\nabla_{\mathbf v} f_{\theta,\mathbf v}(x)$ and the induced movement of the class means. 

This decomposition clarifies the optimization dynamics: when the residual term is small, the first term pushes the representation toward greater class separation, while the residual term captures a correction that depends on local variability.

\begin{proof}[Proof of Proposition \ref{prop_1}]
Denote by
$$
\mathcal{L}_{\beta}(\mathbf{v},\boldsymbol{\mu}_0,\boldsymbol{\mu}_1)
:=
\log q_\beta\!\left(
y \mid f_{\theta,\mathbf{v}}(x),
\boldsymbol{\mu}_0,
\boldsymbol{\mu}_1
\right).
$$
Then
$\mathcal J(\mathbf v)=\mathbb E_{x,y}\!\left[
\mathcal{L}_{\beta}\bigl(\mathbf v,\boldsymbol{\mu}_0(\mathbf v),\boldsymbol{\mu}_1(\mathbf v)\bigr)\right]$. Differentiating with respect to $\mathbf v$ gives
\[
\nabla_{\mathbf v}\mathcal J(\mathbf v)
=
\mathbb E_{x,y}\!\left[
\nabla_f \mathcal{L}_{\beta}^\top
\nabla_{\mathbf v} f_{\theta,\mathbf v}(x)
\right]
+
\sum_{c\in\{0,1\}}
\Bigl(\nabla_{\mathbf v}\boldsymbol{\mu}_c(\mathbf v)\Bigr)^\top
\nabla_{\boldsymbol{\mu}_c}\mathcal J(\mathbf v).
\]
Since $\boldsymbol{\mu}_c(\mathbf v)$ is the optimal class mean direction on $\mathbb S^{d-1}$, the first-order optimality condition implies that $\nabla_{\boldsymbol{\mu}_c}\mathcal J(\mathbf v)$ is normal to $\mathbb S^{d-1}$ at $\boldsymbol{\mu}_c(\mathbf v)$. On the other hand, $\|\boldsymbol{\mu}_c(\mathbf v)\|=1$ for all $\mathbf v$, and differentiating $\|\boldsymbol{\mu}_c(\mathbf v)\|^2=1$ yields $\boldsymbol{\mu}_c(\mathbf v)^\top \nabla_{\mathbf v}\boldsymbol{\mu}_c(\mathbf v)=0$, which shows that $\nabla_{\mathbf v}\boldsymbol{\mu}_c(\mathbf v)$ lies in the tangent space of $\mathbb S^{d-1}$ at $\boldsymbol{\mu}_c(\mathbf v)$. Consequently,
\[
\Bigl(\nabla_{\mathbf v}\boldsymbol{\mu}_c(\mathbf v)\Bigr)^\top
\nabla_{\boldsymbol{\mu}_c}\mathcal J(\mathbf v)=0,
\]
which implies that
\begin{equation}\label{prop_1_pf_eq_1_new}
\nabla_{\mathbf v}\mathcal J(\mathbf v)
=
\mathbb E_{x,y}\!\left[
\left(\nabla_f \mathcal{L}_{\beta}\right)^\top
\nabla_{\mathbf v} f_{\theta,\mathbf v}(x)
\right].
\end{equation}

For simplicity, write
$q_\beta(y\mid x):=
q_\beta\!\left(
y \mid f_{\theta,\mathbf{v}}(x),
\boldsymbol{\mu}_0(\mathbf v),
\boldsymbol{\mu}_1(\mathbf v)
\right)$. Under the matched case $\beta=\kappa$, this coincides with the class posterior induced by the shared-concentration vMF model, and
\[
q_\beta(y\mid x)=
\frac{\exp\{\beta \boldsymbol{\mu}_y(\mathbf v)^\top f_{\theta,\mathbf{v}}(x)\}}
{\exp\{\beta \boldsymbol{\mu}_0(\mathbf v)^\top f_{\theta,\mathbf{v}}(x)\}
+\exp\{\beta \boldsymbol{\mu}_1(\mathbf v)^\top f_{\theta,\mathbf{v}}(x)\}}.
\]
A direct calculation gives
\[
\nabla_f \mathcal{L}_{\beta}
=
\beta\Bigl(
\boldsymbol{\mu}_y(\mathbf v)
-
q_\beta(0\mid x)\boldsymbol{\mu}_0(\mathbf v)
-
q_\beta(1\mid x)\boldsymbol{\mu}_1(\mathbf v)
\Bigr),
\]
which can be rewritten as
\[
\nabla_f \mathcal{L}_{\beta}
=
\beta\bigl(y-q_\beta(1\mid x)\bigr)
\bigl(\boldsymbol{\mu}_1(\mathbf v)-\boldsymbol{\mu}_0(\mathbf v)\bigr).
\]
Substituting this into \eqref{prop_1_pf_eq_1_new} and conditioning on $y$ yields
\begin{equation}\label{eq:prop_pf_1_new}
\begin{aligned}
\nabla_{\mathbf v}\mathcal J(\mathbf v)
&=
\frac{\beta}{2}
\bigl(\boldsymbol{\mu}_1(\mathbf v)-\boldsymbol{\mu}_0(\mathbf v)\bigr)^\top
\Bigl(
\mathbb E_x\!\left[q_\beta(0\mid x)\nabla_{\mathbf v} f_{\theta,\mathbf v}(x)\mid y=1\right]
\\
&\hspace{4.5cm}
-
\mathbb E_x\!\left[q_\beta(1\mid x)\nabla_{\mathbf v} f_{\theta,\mathbf v}(x)\mid y=0\right]
\Bigr),
\end{aligned}
\end{equation}
where we used $p(y=0)=p(y=1)=1/2$. For each $c\in\{0,1\}$, introduce the residual term
\begin{equation}
\tag{\ref{residual_term}}
\Delta_c(\mathbf v):=
\mathbb E_x\!\left[
q_\beta(1-c\mid x)
\bigl(
\nabla_{\mathbf v} f_{\theta,\mathbf v}(x)
-
\nabla_{\mathbf v}\boldsymbol{\mu}_c(\mathbf v)
\bigr)
\mid y=c
\right].
\end{equation}
Then, we obtain
$$
\mathbb E_x\!\left[
q_\beta(0\mid x)\nabla_{\mathbf v} f_{\theta,\mathbf v}(x)
\mid y=1
\right]
=
\mathbb E_x[q_\beta(0\mid x)\mid y=1]
\nabla_{\mathbf v}\boldsymbol{\mu}_1(\mathbf v)
+
\Delta_1(\mathbf v),
$$
and
$$
\mathbb E_x\!\left[
q_\beta(1\mid x)\nabla_{\mathbf v} f_{\theta,\mathbf v}(x)
\mid y=0
\right]
=
\mathbb E_x[q_\beta(1\mid x)\mid y=0]
\nabla_{\mathbf v}\boldsymbol{\mu}_0(\mathbf v)
+
\Delta_0(\mathbf v).
$$
Substituting these identities into \eqref{eq:prop_pf_1_new} gives
\begin{equation}\label{eq:prop_pf_2_new}
\begin{aligned}
\nabla_{\mathbf v}\mathcal J(\mathbf v)
&=
\frac{\beta}{2}
\bigl(\boldsymbol{\mu}_1(\mathbf v)-\boldsymbol{\mu}_0(\mathbf v)\bigr)^\top
\Bigl(
\mathbb E_x[q_\beta(0\mid x)\mid y=1]
\nabla_{\mathbf v}\boldsymbol{\mu}_1(\mathbf v)
\\
&\hspace{4.5cm}
-
\mathbb E_x[q_\beta(1\mid x)\mid y=0]
\nabla_{\mathbf v}\boldsymbol{\mu}_0(\mathbf v)
\Bigr)
\\
&\quad+
\frac{\beta}{2}
\bigl(\boldsymbol{\mu}_1(\mathbf v)-\boldsymbol{\mu}_0(\mathbf v)\bigr)^\top
\bigl(
\Delta_1(\mathbf v)-\Delta_0(\mathbf v)
\bigr).
\end{aligned}
\end{equation}

Under the uniform class prior and the matched case $\beta=\kappa$, $q_\beta(y\mid x)=p(y\mid x)$, so
\begin{align*}
\mathbb E_x[q_\beta(0\mid x)\mid y=1]
&=\int p(0\mid x)p(x\mid y=1)\,dx\\
&=2\int p(0\mid x)p(1\mid x)p(x)\,dx\\
&=2\mathbb E_x[q_\beta(0\mid x)q_\beta(1\mid x)].
\end{align*}
Similarly,
\[
\mathbb E_x[q_\beta(1\mid x)\mid y=0]
=
2\mathbb E_x[q_\beta(0\mid x)q_\beta(1\mid x)].
\]
Using these identities in \eqref{eq:prop_pf_2_new}, we obtain
\[
\begin{aligned}
\nabla_{\mathbf v}\mathcal J(\mathbf v)
&=
\beta\,
\mathbb E_x[q_\beta(0\mid x)q_\beta(1\mid x)]
\bigl(\boldsymbol{\mu}_1(\mathbf v)-\boldsymbol{\mu}_0(\mathbf v)\bigr)^\top
\bigl(
\nabla_{\mathbf v}\boldsymbol{\mu}_1(\mathbf v)
-
\nabla_{\mathbf v}\boldsymbol{\mu}_0(\mathbf v)
\bigr)
\\
&\quad+
\frac{\beta}{2}
\bigl(\boldsymbol{\mu}_1(\mathbf v)-\boldsymbol{\mu}_0(\mathbf v)\bigr)^\top
\bigl(
\Delta_1(\mathbf v)-\Delta_0(\mathbf v)
\bigr).
\end{aligned}
\]
By using
\[
\nabla_{\mathbf v}
\|\boldsymbol{\mu}_1(\mathbf v)-\boldsymbol{\mu}_0(\mathbf v)\|^2
=
2\bigl(\boldsymbol{\mu}_1(\mathbf v)-\boldsymbol{\mu}_0(\mathbf v)\bigr)^\top
\bigl(
\nabla_{\mathbf v}\boldsymbol{\mu}_1(\mathbf v)
-
\nabla_{\mathbf v}\boldsymbol{\mu}_0(\mathbf v)
\bigr),
\]
we obtain
$$
\nabla_{\mathbf v}\mathcal J(\mathbf v)
=
\frac{\beta}{2}\,
\mathbb E_x[q_\beta(0\mid x)q_\beta(1\mid x)]\,
\nabla_{\mathbf v}
\|\boldsymbol{\mu}_1(\mathbf v)-\boldsymbol{\mu}_0(\mathbf v)\|^2
+
\frac{\beta}{2}
\bigl(\boldsymbol{\mu}_1(\mathbf v)-\boldsymbol{\mu}_0(\mathbf v)\bigr)^\top
\bigl(
\Delta_1(\mathbf v)-\Delta_0(\mathbf v)
\bigr).
$$
Finally, using the logistic form of $q_\beta$,
\begin{equation}
\label{eq:p_1_p_0_new}
\begin{aligned}
\frac{1}{2}\mathbb E_x[q_\beta(0\mid x)q_\beta(1\mid x)]
&=
\mathbb E_x\!\left[
\frac{
\exp\!\left\{
\beta\bigl(\boldsymbol{\mu}_0(\mathbf v)+\boldsymbol{\mu}_1(\mathbf v)\bigr)^\top
f_{\theta,\mathbf v}(x)
\right\}
}{
2\left(
\sum_{c\in\{0,1\}}
\exp\!\left\{
\beta \boldsymbol{\mu}_c(\mathbf v)^\top
f_{\theta,\mathbf v}(x)
\right\}
\right)^2
}
\right]
\\
&=
\frac{1}{8}
\left(1-\mathbb E_x\!\left[
\tanh^2\left(
\frac{\beta}{2}
\bigl(\boldsymbol{\mu}_1(\mathbf v)-\boldsymbol{\mu}_0(\mathbf v)\bigr)^\top
f_{\theta,\mathbf v}(x)
\right)\right]\right),
\end{aligned}
\end{equation}
which yields
$$
\nabla_{\mathbf v}\mathcal J(\mathbf v)
=
\beta\,\bar{\omega}(\mathbf v)\,
\nabla_{\mathbf v}
\|\boldsymbol{\mu}_1(\mathbf v)-\boldsymbol{\mu}_0(\mathbf v)\|^2
+
\frac{\beta}{2}
\bigl(\boldsymbol{\mu}_1(\mathbf v)-\boldsymbol{\mu}_0(\mathbf v)\bigr)^\top
\bigl(
\Delta_1(\mathbf v)-\Delta_0(\mathbf v)
\bigr).
$$
This completes the proof.
\end{proof}

\subsection{Distribution Shift on Two Types of Errors}\label{app_prop_2}
In practice, the deployment distribution may differ from the source distribution used for training and calibration. 
We consider this mismatch as a distribution shift between the source and deployment distributions. 
In this section, we quantify its impact via a Wasserstein perturbation and study how it affects Type I error control.

\begin{assumption}\label{asm_4}
We denote by $\mathbb{P}_c$ the source distribution of class $c$ used for training and calibration, and by $\widetilde{\mathbb{P}}_c$ the corresponding class-conditional distribution under deployment.
Assume there exists a constant $\mathcal{E} > 0$ such that, for a fixed representation map $f_{\theta,\mathbf{v}} : \mathcal{X} \to \mathbb{S}^{d-1}$, the following holds for each class $c \in \{0,1\}$:
\[
\mathcal{W}_1\!\left( (f_{\theta,\mathbf{v}})_\# \mathbb{P}_c,\; (f_{\theta,\mathbf{v}})_\# \widetilde{\mathbb{P}}_c \right) \le \mathcal{E},
\]
where $\mathcal{W}_1(\mu, \nu) := \inf_{\gamma \in \Pi(\mu, \nu)} \mathbb{E}_{(z, \tilde{z}) \sim \gamma} [\|z - \tilde{z}\|]$ denotes the $1$-Wasserstein distance on $\mathbb{S}^{d-1}$, and $(f_{\theta,\mathbf{v}})_\# \mathbb{P}_c$ denotes the pushforward measure of $\mathbb{P}_c$ under $f_{\theta,\mathbf{v}}$.
\end{assumption}

\begin{proposition}\label{prop_3}
Suppose Assumption~\ref{asm_4} holds. For any $\delta \in (0,1)$, with probability at least $1-\delta$ over the randomness of training and calibration samples, the Type-I error under the shifted distribution $\widetilde{\mathbb P}_0$ satisfies
\[
\begin{aligned}
\widetilde{\mathbb P}_0\bigl(
\widehat{\mathcal S}_t(X)>\widehat{\tau}_{\alpha,t}
\bigr)-\alpha\le
\inf_{\varepsilon>0}
\left(
\sqrt{\frac{\log(2/\delta)}{2n_2}}
+\frac1{n_2}
+\mathbb P_0\left(
\widehat{\tau}_{\alpha,t}-\varepsilon
\le
\widehat{\mathcal S}_t(X)
\le
\widehat{\tau}_{\alpha,t}
\right)
+\frac{2\kappa\mathcal E}{\varepsilon}
\right)\wedge 1 .
\end{aligned}
\]
The Type-II error under the shifted distribution $\widetilde{\mathbb P}_1$ satisfies
\[
\widetilde{\mathbb P}_1\bigl(
\widehat{\mathcal S}_t(X)\le\widehat{\tau}_{\alpha,t}
\bigr)
\le
\inf_{\varepsilon>0}
\left(
\mathbb P_1\bigl(
\widehat{\mathcal S}_t(X)
\le
\widehat{\tau}_{\alpha,t}+\varepsilon
\bigr)
+
\frac{2\kappa\mathcal E}{\varepsilon}
\right)\wedge 1 .
\]
\end{proposition}


Proposition~\ref{prop_3} characterizes how the Type-I and Type-II errors depend on the magnitude of the distribution shift $\mathcal{E}$. 
When $\mathcal{E}=0$, the bounds recover the in-distribution guarantees in Theorem~\ref{thm_1}. 
When $\mathcal{E}>0$, both errors incur an additional penalty that scales with $\mathcal{E}$, reflecting performance degradation under shift. 
The parameter $\varepsilon$ captures a trade-off between the local probability mass near the threshold and the shift-induced error. 
Overall, the result indicates that the proposed method is robust to moderate shifts, while its performance degrades gracefully as the distribution shift increases.

\begin{proof}[Proof of Proposition \ref{prop_3}]
For simplicity of notation, we drop the time index $t$ in the proof. Let $f_{\theta, \mathbf{v}}(x) \in \mathbb{S}^{d-1}$ denote the representation for $x \in \mathcal{X}$. The deployed scoring function is defined as $\widehat{\mathcal{S}}(x)=\kappa\left(\widehat{\boldsymbol{\mu}}_1-\widehat{\boldsymbol{\mu}}_0\right)^{\top} f_{\theta, \mathbf{v}}(x)$.
Since the estimated mean directions $\widehat{\boldsymbol{\mu}}_1, \widehat{\boldsymbol{\mu}}_0$ lie on the unit hypersphere $\mathbb{S}^{d-1}$, by the triangle inequality, $\left\|\widehat{\boldsymbol{\mu}}_1-\widehat{\boldsymbol{\mu}}_0\right\| \leq\left\|\widehat{\boldsymbol{\mu}}_1\right\|+\left\|\widehat{\boldsymbol{\mu}}_0\right\| \leq 2$.
The gradient of $\widehat{\mathcal{S}}$ with respect to the representation $f$ is given by
\begin{equation}\label{prop_3_eq_1}
    \|\nabla_f \widehat{\mathcal{S}}(x)\| = \kappa \|\widehat{\boldsymbol{\mu}}_1 - \widehat{\boldsymbol{\mu}}_0\| \le 2\kappa.
\end{equation}
Let $\pi \in \Pi(\mathbb{P}_0, \widetilde{\mathbb{P}}_0)$ be the optimal coupling such that 
$$
\mathbb{E}_{(X, \widetilde{X})\sim \pi}[\|f_{\theta,\mathbf{v}}(X) - f_{\theta,\mathbf{v}}(\widetilde{X})\|] = \mathcal{W}_1\left( (f_{\theta,\mathbf{v}})_\# \mathbb{P}_0, (f_{\theta,\mathbf{v}})_\# \widetilde{\mathbb{P}}_0 \right) \le \mathcal{E}.
$$ 
For any $\varepsilon > 0$, we decompose the Type-I error probability under the shifted distribution as
\begin{equation}\label{prop_3_eq_2}
\begin{aligned}
\widetilde{\mathbb{P}}_0
\big(\widehat{\mathcal{S}}(\widetilde{X})>\widehat{\tau}_\alpha\big)
&\le
\mathbb{P}_\pi\Big(
\widehat{\mathcal{S}}(\widetilde{X})>\widehat{\tau}_\alpha,\,
\|f_{\theta,\mathbf{v}}(X)-f_{\theta,\mathbf{v}}(\widetilde{X})\|
\le \frac{\varepsilon}{2\kappa}
\Big)\\
&\quad+
\mathbb{P}_\pi\Big(
\|f_{\theta,\mathbf{v}}(X)-f_{\theta,\mathbf{v}}(\widetilde{X})\|
>
\frac{\varepsilon}{2\kappa}
\Big).
\end{aligned}
\end{equation}
For the first term on the RHS of~\eqref{prop_3_eq_2}, by using~\eqref{prop_3_eq_1}, we obtain
$$
|\widehat{\mathcal{S}}(X) - \widehat{\mathcal{S}}(\widetilde{X})| \leq 2\kappa \|f_{\theta,\mathbf{v}}(X) -f_{\theta,\mathbf{v}}(\widetilde{X})\| \leq \varepsilon.
$$
Thus, $\widehat{\mathcal{S}}(\widetilde{X}) > \widehat{\tau}_\alpha$ implies $\widehat{\mathcal{S}}(X) > \widehat{\tau}_\alpha-\varepsilon$. By Theorem~\ref{thm_1}, with probability at least $1-\delta$ over the calibration sample $\mathcal{S}_{\mathrm{cal}}$,
\[
\mathbb{P}_0\!\left(
\widehat{\mathcal{S}}(X)>\widehat{\tau}_\alpha
\right)
\le
\alpha
+
\sqrt{\frac{\log(2/\delta)}{2n_2}}
+
\frac{1}{n_2}.
\]
On this high-probability event, we therefore have
\begin{equation}\label{prop_3_eq_3}
\begin{aligned}
\mathbb{P}_\pi\Big(
\widehat{\mathcal{S}}(\widetilde{X})>\widehat{\tau}_\alpha,\,
\|f_{\theta,\mathbf{v}}(X)-f_{\theta,\mathbf{v}}(\widetilde{X})\|
\le \frac{\varepsilon}{2\kappa}
\Big)\le
\mathbb{P}_0\big(
\widehat{\mathcal{S}}(X)>\widehat{\tau}_\alpha
\big)
+
\mathbb{P}_0\big(
\widehat{\tau}_\alpha-\varepsilon
\le
\widehat{\mathcal{S}}(X)
\le
\widehat{\tau}_\alpha
\big).
\end{aligned}
\end{equation}
For the second term on the RHS of~\eqref{prop_3_eq_2}, applying Markov's inequality yields
\begin{equation*}
    \mathbb{P}_\pi\Big(\|f_{\theta,\mathbf{v}}(X) -f_{\theta,\mathbf{v}}(\widetilde{X})\| > \frac{\varepsilon}{2\kappa}\Big) \le \frac{2\kappa}{\varepsilon} \mathbb{E}_\pi\big[\|f_{\theta,\mathbf{v}}(X) -f_{\theta,\mathbf{v}}(\widetilde{X})\|\big] \le \frac{2\kappa \mathcal{E}}{\varepsilon}.
\end{equation*}
Substituting~\eqref{prop_3_eq_3} into~\eqref{prop_3_eq_2} and applying Markov's inequality, we obtain, for any $\varepsilon>0$,
\[
\begin{aligned}
\widetilde{\mathbb P}_0
\bigl(
\widehat{\mathcal S}(\widetilde X)>\widehat\tau_\alpha
\bigr)
&\le
\mathbb P_0
\bigl(
\widehat{\mathcal S}(X)>\widehat\tau_\alpha
\bigr)
+
\mathbb P_0\!\left(
\widehat\tau_\alpha-\varepsilon
\le
\widehat{\mathcal S}(X)
\le
\widehat\tau_\alpha
\right)
+
\frac{2\kappa\mathcal E}{\varepsilon}
\\
&\le
\alpha
+
\sqrt{\frac{\log(2/\delta)}{2n_2}}
+
\frac{1}{n_2}
+
\mathbb P_0\!\left(
\widehat\tau_\alpha-\varepsilon
\le
\widehat{\mathcal S}(X)
\le
\widehat\tau_\alpha
\right)
+
\frac{2\kappa\mathcal E}{\varepsilon}.
\end{aligned}
\]
Taking the infimum over $\varepsilon>0$ gives the stated Type-I bound.

Now we turn to the Type-II error. Let $\pi' \in \Pi(\mathbb P_1,\widetilde{\mathbb P}_1)$ be an optimal coupling satisfying
\[
\mathbb E_{(X,\widetilde X)\sim\pi'}
\left[
\|f_{\theta,\mathbf v}(X)-f_{\theta,\mathbf v}(\widetilde X)\|
\right]
\le \mathcal E.
\]
For any $\varepsilon>0$,
\[
\begin{aligned}
\widetilde{\mathbb P}_1
\bigl(
\widehat{\mathcal S}(\widetilde X)>\widehat\tau_\alpha
\bigr)
&\ge
\mathbb P_{\pi'}\Bigl(
\widehat{\mathcal S}(\widetilde X)>\widehat\tau_\alpha,\,
\|f_{\theta,\mathbf v}(X)-f_{\theta,\mathbf v}(\widetilde X)\|
\le \frac{\varepsilon}{2\kappa}
\Bigr)
\\
&\ge
\mathbb P_{\pi'}\Bigl(
\widehat{\mathcal S}(X)>\widehat\tau_\alpha+\varepsilon,\,
\|f_{\theta,\mathbf v}(X)-f_{\theta,\mathbf v}(\widetilde X)\|
\le \frac{\varepsilon}{2\kappa}
\Bigr)
\\
&\ge
\mathbb P_1\bigl(
\widehat{\mathcal S}(X)>\widehat\tau_\alpha+\varepsilon
\bigr)
-
\mathbb P_{\pi'}\Bigl(
\|f_{\theta,\mathbf v}(X)-f_{\theta,\mathbf v}(\widetilde X)\|
>
\frac{\varepsilon}{2\kappa}
\Bigr)
\\
&\ge
\mathbb P_1\bigl(
\widehat{\mathcal S}(X)>\widehat\tau_\alpha+\varepsilon
\bigr)
-
\frac{2\kappa\mathcal E}{\varepsilon}.
\end{aligned}
\]
Equivalently,
\[
\widetilde{\mathbb P}_1
\bigl(
\widehat{\mathcal S}(\widetilde X)\le\widehat\tau_\alpha
\bigr)
\le
\mathbb P_1
\bigl(
\widehat{\mathcal S}(X)\le
\widehat\tau_\alpha+\varepsilon
\bigr)
+
\frac{2\kappa\mathcal E}{\varepsilon}.
\]
Taking the infimum over $\varepsilon>0$ gives the stated Type-II bound.
\end{proof}

\section{Broader Impacts}\label{app:broader_imp}
The proposed \texttt{S2D} framework contributes to improving the reliability of distinguishing LLM-generated text from human-written text, which has implications for mitigating risks such as misinformation, academic misconduct, and erosion of trust in digital content. By leveraging hidden representations and providing a statistically grounded detection procedure with controlled Type-I error, our method is particularly relevant in high-stakes scenarios where false accusations (e.g., misclassifying human-written text as machine-generated) carry ethical consequences. In addition, the representation-based nature of \texttt{S2D} allows it to operate without requiring access to the generation process, making it broadly applicable in real-world settings where watermarking is unavailable.

More broadly, such capabilities may influence how content is created and consumed, potentially shaping emerging norms around disclosure and authenticity. They may also inform policy discussions on transparency and responsible AI use, including whether and how generated content should be labeled. Overall, advances in this area contribute to a broader effort to maintain trust, accountability, and informed decision-making in an increasingly AI-mediated information landscape.

\newpage
\section*{NeurIPS Paper Checklist}

\begin{enumerate}

\item {\bf Claims}
    \item[] Question: Do the main claims made in the abstract and introduction accurately reflect the paper's contributions and scope?
    \item[] Answer: \answerYes{} 
    \item[] Justification: The main claims made in the abstract are completely based on our work. 

\item {\bf Limitations}
    \item[] Question: Does the paper discuss the limitations of the work performed by the authors?
    \item[] Answer: \answerYes{} 
    \item[] Justification: We discussed limitations and the potential future work together in Section~\ref{sec:discussion}. 

\item {\bf Theory assumptions and proofs}
    \item[] Question: For each theoretical result, does the paper provide the full set of assumptions and a complete (and correct) proof?
    \item[] Answer: \answerYes{} 
    \item[] Justification:  The paper clearly states the assumptions required for each theoretical result and provides corresponding proofs (primarily in the Appendix~\ref{app_pf}). 

    \item {\bf Experimental result reproducibility}
    \item[] Question: Does the paper fully disclose all the information needed to reproduce the main experimental results of the paper to the extent that it affects the main claims and/or conclusions of the paper (regardless of whether the code and data are provided or not)?
    \item[] Answer: \answerYes{} 
    \item[] Justification: Disclosed all the experiment information in Section~\ref{sec:exp} and Appendix~\ref{sec:experiment_details}.

\item {\bf Open access to data and code}
    \item[] Question: Does the paper provide open access to the data and code, with sufficient instructions to faithfully reproduce the main experimental results, as described in supplemental material?
    \item[] Answer: \answerYes{} 
    \item[] Justification: All data are open-source, information is provided in the content, and code link is provided in Section~\ref{sec:exp}. 

\item {\bf Experimental setting/details}
    \item[] Question: Does the paper specify all the training and test details (e.g., data splits, hyperparameters, how they were chosen, type of optimizer) necessary to understand the results?
    \item[] Answer: \answerYes{} 
    \item[] Justification: We detailed all the experiment information in Section~\ref{sec:exp} and Appendix~\ref{sec:experiment_details}.

\item {\bf Experiment statistical significance}
    \item[] Question: Does the paper report error bars suitably and correctly defined or other appropriate information about the statistical significance of the experiments?
    \item[] Answer: \answerYes{} 
    \item[] Justification: In our experiments, we evaluate our method and baseline methods on five runs with different random samplings of the training data. We report the average performance across the datasets. 

\item {\bf Experiments compute resources}
    \item[] Question: For each experiment, does the paper provide sufficient information on the computer resources (type of compute workers, memory, time of execution) needed to reproduce the experiments?
    \item[] Answer: \answerYes{} 
    \item[] Justification: Provided in Section~\ref{sec:experiment_details}.
    
\item {\bf Code of ethics}
    \item[] Question: Does the research conducted in the paper conform, in every respect, with the NeurIPS Code of Ethics \url{https://neurips.cc/public/EthicsGuidelines}?
    \item[] Answer: \answerYes{} 
    \item[] Justification: The research conducted in this paper conforms in every respect with the NeurIPS Code of Ethics.

\item {\bf Broader impacts}
    \item[] Question: Does the paper discuss both potential positive societal impacts and negative societal impacts of the work performed?
    \item[] Answer: \answerYes{} 
    \item[] Justification: Discussed in Appendix~\ref{app:broader_imp}.
    
\item {\bf Safeguards}
    \item[] Question: Does the paper describe safeguards that have been put in place for responsible release of data or models that have a high risk for misuse (e.g., pre-trained language models, image generators, or scraped datasets)?
    \item[] Answer: \answerNA{} 
    \item[] Justification: This paper does not involve the release of data, models, or other assets that pose a high risk of misuse.

\item {\bf Licenses for existing assets}
    \item[] Question: Are the creators or original owners of assets (e.g., code, data, models), used in the paper, properly credited and are the license and terms of use explicitly mentioned and properly respected?
    \item[] Answer: \answerYes{} 
    \item[] Justification: All the data, code, and models used in the paper are properly credited. 

\item {\bf New assets}
    \item[] Question: Are new assets introduced in the paper well documented and is the documentation provided alongside the assets?
    \item[] Answer: \answerNA{} 
    \item[] Justification: The paper does not release new assets. 

\item {\bf Crowdsourcing and research with human subjects}
    \item[] Question: For crowdsourcing experiments and research with human subjects, does the paper include the full text of instructions given to participants and screenshots, if applicable, as well as details about compensation (if any)? 
    \item[] Answer: \answerNA{} 
    \item[] Justification: The paper does not involve crowdsourcing nor research with human subjects. 

\item {\bf Institutional review board (IRB) approvals or equivalent for research with human subjects}
    \item[] Question: Does the paper describe potential risks incurred by study participants, whether such risks were disclosed to the subjects, and whether Institutional Review Board (IRB) approvals (or an equivalent approval/review based on the requirements of your country or institution) were obtained?
    \item[] Answer: \answerNA{} 
    \item[] Justification: The paper does not involve crowdsourcing nor research with human subjects. 

\item {\bf Declaration of LLM usage}
    \item[] Question: Does the paper describe the usage of LLMs if it is an important, original, or non-standard component of the core methods in this research? Note that if the LLM is used only for writing, editing, or formatting purposes and does \emph{not} impact the core methodology, scientific rigor, or originality of the research, declaration is not required.
    \item[] Answer: \answerYes{} 
    \item[] Justification: This work involves the use of large language models to extract hidden representations and to rewrite texts. We provided the detailed usage of LLMs in Section~\ref{sec:exp} and Appendix~\ref{sec:experiment_details}. 

\end{enumerate}

\end{document}